\documentclass[
superscriptaddress,
nofootinbib,
 amsmath,amssymb,
 aps,
pra,
twocolumn,
nofootinbib,
floatfix,
10pt
]{revtex4-2}

\usepackage{graphicx}
\usepackage{dcolumn}
\usepackage{bm}
\usepackage{hyperref}

\usepackage{amsmath}
\usepackage{amssymb}
\usepackage{stackrel}
\usepackage{tikz-cd}
\usepackage{tikz}
\usepackage[cm]{fullpage}
\usepackage{amsthm}
\usepackage{mathtools}
\usepackage{enumitem}
\usepackage{bbold}
\usepackage[normalem]{ulem}
\usepackage[capitalise]{cleveref}
\Crefname{equation}{Eq.}{Eqs.}
\Crefname{figure}{Fig.}{Figs.}
\usepackage{appendix}
\usepackage{comment}
\usepackage{centernot}
\usepackage{braket}
\usepackage{apptools}
\AtAppendix{\counterwithin{Proposition}{section}}
\usepackage{graphicx}
\usepackage{subcaption}

\definecolor{googleblue}{RGB}{34, 0, 204}
\definecolor{panblue}{RGB}{0,24,150}
\definecolor{carmine}{RGB}{150, 0, 24}
\usepackage{hyperref} 
\hypersetup{
unicode=true,
bookmarksnumbered=false,
bookmarksopen=false,
breaklinks=false,
colorlinks=true,
linkcolor=carmine,
citecolor=googleblue,
urlcolor=panblue,
anchorcolor=OliveGreen}

\usepackage{multirow}
\usepackage{array}

\usepackage{microtype} 
\microtypecontext{spacing=nonfrench}
\microtypesetup{
expansion={true,nocompatibility},
protrusion={true,nocompatibility},
activate={true,nocompatibility},
tracking=false,
kerning=true,
spacing={true}
}

\usepackage[all=normal,floats=tight,mathspacing=tight,wordspacing=normal,paragraphs=normal,tracking=tight,charwidths=tight,mathdisplays=normal,sections=normal,margins=normal]{savetrees}

\newtheorem{Theorem}{Theorem}
\newtheorem{Lemma}{Lemma}
\newtheorem{Proposition}{Proposition}

\theoremstyle{definition}
\newtheorem{Definition}{Definition}
\newtheorem{Corollary}{Corollary}

\DeclareMathOperator{\Tr}{Tr}

\DeclareMathOperator{\supp}{supp}
\DeclareMathOperator{\Pa}{Pa}

\newcommand{\LA}{L^{(A)}}
\newcommand{\LC}{L^{(C)}}
\newcommand{\MA}{M^{(A)}}
\newcommand{\MB}{M^{(B)}}
\newcommand{\NB}{N^{(B)}}
\newcommand{\NC}{N^{(C)}}

\renewcommand{\H}{\mathcal{H}}
\newcommand{\E}{\mathcal{E}}

\begin{document}

\preprint{APS/123-QED}
\title{Fully quantum inflation: quantum marginal problem constraints in the service of causal inference}

\author{Isaac D. Smith}
\email{isaac.smith@uibk.ac.at}
\affiliation{University of Innsbruck, Department of Theoretical Physics, Technikerstr. 21A, Innsbruck A-6020, Austria}
\author{Elie Wolfe}
\email{ewolfe@perimeterinstitute.ca}
\author{Robert W. Spekkens}
\email{rspekkens@perimeterinstitute.ca}
\affiliation{Perimeter Institute for Theoretical Physics, 31 Caroline Street North, Waterloo, Ontario Canada N2L, 2Y5}

\date{\today}

\begin{abstract}   
Consider the problem of deciding, for a particular multipartite quantum state, whether or not it is realizable in a quantum network with a particular causal structure. This is a fully quantum version of what causal inference researchers refer to as the problem of causal discovery. In this work, we introduce a fully quantum version of the inflation technique for causal inference, which leverages the quantum marginal problem. The primary example by which we illustrate the utility of this method is testing compatibility of tripartite quantum states with the  quantum network known as the triangle scenario. We show, in particular, how the method yields a complete classification of pure three-qubit states into those that are and those that are not compatible with the triangle scenario.  We also provide some illustrative examples involving mixed states and some where one or more of the systems is higher-dimensional. Furthermore, we examine the question of when the incompatibility of a multipartite quantum state with a causal structure can  be inferred from the incompatibility of a joint probability distribution induced by implementing measurements on each subsystem. Finally, we present a family of networks, which includes the triangle scenario as a special case, for which causal compatibility constraints can be derived.
\end{abstract}

\maketitle
\tableofcontents

\section{Introduction}

By providing a method for formalizing intuitive or explicit cause-effect relationships between variables, causal modelling has become an active subfield of statistics and machine learning~\cite{peters2017elements}.  More recently, notions of causality have begun to play an instructive role also in quantum physics.  One of the best clues for deducing the precise manner in which quantum physics implies a departure from the principles underlying classical physics, Bell's theorem~\cite{Bell_Aspect_2004,sep-bell-theorem}, admits a distinctly causal interpretation~\cite{wood2015lesson, fritz2012beyond}. This newfound interest in the intersection of causality and quantum physics has led   to a better understanding of quantum-classical gaps in an array of causal structures~\cite{fritz2016beyond}, the investigation of potential benefits for causal inference~\cite{ried2015quantum,chiribella2019quantum}, and the formalization of the notion of an intrinsically quantum causal model~\cite{Allen_17,Costa_2016,barrett2019quantum}.
 
In a classical causal model~\cite{Pearl_2009,spirtes2001causation}, a causal structure is represented by a directed acyclic graph (DAG) where each node of the DAG represents a random variable and each directed edge represents a potential causal influence between these.  The parameters of the model can be taken to be a set of conditional probability distributions, one for each variable conditioned on its causal parents in the DAG.  A given probability distribution over the subset of observed variables is said to be {\em compatible} with the causal structure if it can be obtained from some choice of the parameters after marginalizing over the unobserved (i.e., latent) variables. A central problem in classical causal inference, which we term the {\em causal compatibility problem}, is deciding whether a given distribution over observed variables is compatible with a given causal structure.  

In this work, we are interested in a fully quantum version of the latter problem, where the distribution over observed variables is replaced by a quantum state over a multipartite quantum system.  That is, we are interested in the question: for a given quantum state on a multipartite system, is it compatible with a particular causal structure holding among the subsystems?
 
To make this question more precise, it is necessary to review the quantum generalization of the notion of a causal model. 
In a {\em quantum} causal model~\cite{Allen_17}, a causal structure is still represented by a DAG, but now each node of the DAG represents a quantum system.  The parameters of the model can be taken to be a set of quantum channels, one for each quantum system where the output of the channel is that quantum system and the inputs are the quantum systems that are its causal parents in the DAG.  Here, the distinction between visible and latent systems is best understood as the visible nodes being the ones that are probed in a quantum experiment.  If the visible quantum systems under consideration are such that there is no cause-effect relationship holding among any of them, then one can assign a joint quantum state to the visible systems.\footnote{Otherwise, one would need to consider some analogue of a joint quantum state over systems that are cause-effect related, which is a thorny problem~\cite{horsman2017can,fullwood2022quantum}.}  This is the case we consider here. A given joint quantum state on the visible systems is said to be {\em compatible} with the causal structure if it can be obtained from some choice of the parameters after marginalizing (i.e., taking the partial trace) over the latent quantum systems. 

The problem  we address, which we term the {\em fully quantum causal compatibility problem} is that of deciding whether a given quantum state over the visible systems is compatible with a given causal structure.

We make use of a method known as the inflation technique~\cite{WolfeSpekkensFritz_2019}, the first applications of which were focused on the case where the visible nodes are classical, and where the latent nodes can be classical or nonclassical.  Here, we apply the technique to the case where both visible and latent nodes are quantum. The inflation technique can be understood at a high level as a method for mapping the problem of determining whether a state  on the visible nodes is compatible with a causal structure $G$ to the problem of determining whether a particular set of  marginals of the state  is compatible with a causal structure $G'$ that is termed an ``inflation'' of $G$. More specifically, the \textit{failure} of a solution to the latter problem demonstrates the \textit{in}compatibility of the state with $G$.\footnote{It is known that the inflation technique completely solves the compatibility problem when the visible and latent nodes are classical~\cite{navascues2020inflation}.  For the case of visible nodes that are classical and latent nodes that are quantum, while the original technique can deliver some necessary conditions for compatibility, finding necessary and sufficient conditions requires the extension of the technique described in Ref.~\cite{wolfe_q_inflation}, with the proof that it completely solves the compatibility problem under certain constraints provided by Ref.~\cite{ligthart2023convergent,ligthart2023inflation}. } 

In tackling the problem of determining whether a particular set of  marginals of the state is compatible with a causal structure $G'$, one can leverage constraints coming from the marginal problem. For the case of visible nodes that are classical, it is the classical marginal problem that is pertinent, namely, the problem of deciding whether a given set of distributions on various subsets of variables are the marginals of a single joint distribution over all the variables. For the case of visible nodes that are {\em quantum}, it is the quantum marginal problem that is pertinent, namely, the problem of deciding whether a given set of {\em quantum states} on various subsets of quantum systems are the marginals of a single joint quantum state over all the quantum systems. 

This connection to the quantum marginal problem means that one can leverage preexisting results in the literature for the purpose of deciding compatibility of quantum states with causal structures having quantum latents. Despite the fact that the quantum marginal problem is known to be difficult in general~\cite{liu2006consistency}, there has nevertheless been progress in a number of cases.  See, e.g., the thorough list of references in Ref.~\cite[Section 1.2]{fraser2023estimation}. Of particular relevance for our purposes are the family of operator inequalities introduced by Hall~\cite{william_hall_05,hall_william_07} (see also~\cite{butterley2006compatibility}). 

\begin{figure}[htbp]
\begin{subfigure}{0.9\linewidth}
        \centering
        \includegraphics[width=\linewidth]{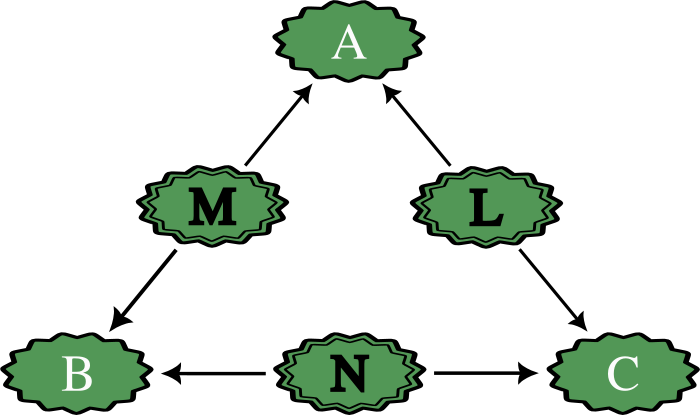}
        \caption{}
    \end{subfigure}
    \vspace{0.5cm}\\
    \begin{subfigure}{0.9\linewidth}
        \centering
        \includegraphics[width=\linewidth]{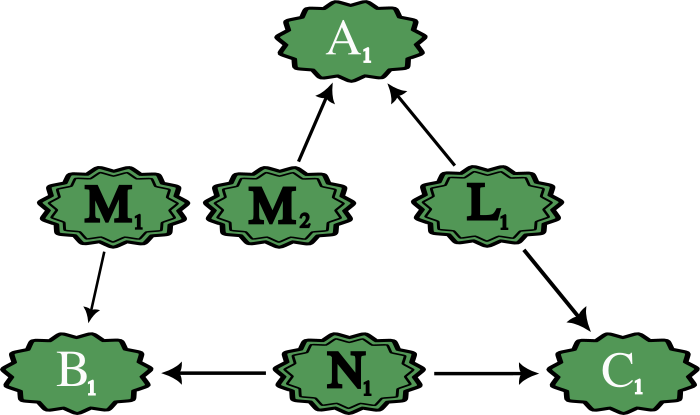}
        \caption{}
    \end{subfigure}
    \vspace{0.5cm}\\
    \begin{subfigure}{0.9\linewidth}
        \centering
        \includegraphics[width=\linewidth]{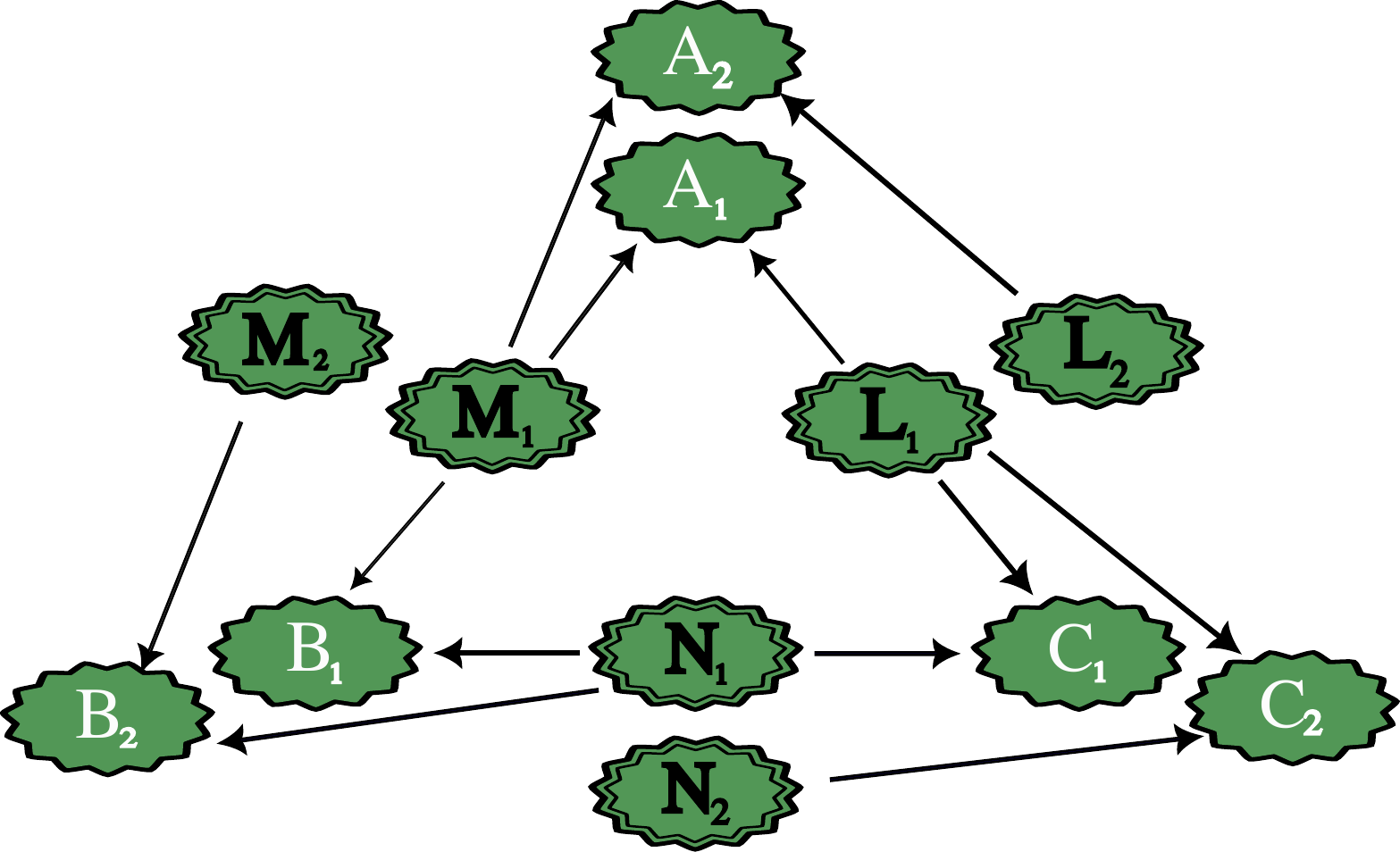}
        \caption{}
    \end{subfigure}
\begin{center}
\caption{ (a) The DAG describing the triangle scenario with quantum semantics for the visible and latent nodes.  (b) The $AB$-Cut inflation of the triangle scenario. (c) The Spiral inflation of the triangle scenario.  In all subfigures, visible nodes are indicated by a single-stroke border while latent nodes are indicated by a dual-stroke border.  The textured edges and green colour of the nodes is to indicate that these are quantum systems rather than classical variables.  The convention for the latter is seen in~\Cref{fig:classicaltriangle}.}
\label{fig:triangleandinflations}
\end{center}
\end{figure}

The main aim of this work is to showcase the utility of the inflation technique in conjunction with these operator inequalities for questions of compatibility of a joint quantum state with a given causal structure.  

For illustrative purposes, we primarily focus on the case of a tripartite system and the well-studied causal structure known as the triangle scenario, i.e., the structure that is  depicted in~\Cref{fig:triangleandinflations}(a). In particular, we use the {\em Cut inflation} of the triangle scenario \cite{WolfeSpekkensFritz_2019}, depicted in~\Cref{fig:triangleandinflations}(b). 
Because the Cut inflation is a so-called {\em nonfanout} inflation (a notion that will be explained further on), it is applicable in the fully quantum case. Furthermore, for the Cut inflation, it is possible to leverage the marginal problem constraint derived by Hall~\cite{william_hall_05}  to obtain strong witnesses of incompatibility. 

We will apply the technique for a few different choices of the dimensionalities of the visible quantum systems.  We focus on the case of three qubits, but we consider also the case of three qutrits and the case of two qubits and a ququart.  We also apply the technique not only to assessments of compatibility of pure states, but mixed states as well.  Furthermore, we show that it is possible to leverage the same marginal problem constraints to obtain causal compatibility constraints for a range of network causal structures beyond the triangle scenario.

A number of benefits of this approach are worth elucidating here. Firstly, this technique is well-suited to numerically specified quantum states. For example, given a state $\rho$ that is specified by experimental state tomography or related methods, it is computationally simple to test compatibility for a specific choice of inflation and marginal inequality. A verdict of incompatibility could then witness the divergence of the actual causal structure of the experiment from the expected one, indicating a deficiency in the experimentalist's understanding of their set-up. This approach is also suitable for considering symbolically specified quantum states, as it utilizes \emph{analytic} quantum marginal problem criteria, in contrast to the semidefinite programming approach of Ref.~\cite{navascues2020genuine}.

One motivation for studying compatibility of pure states with the triangle scenario (and $n$-node generalizations thereof where every subset of ${n{-}1}$ nodes have a latent common cause acting on them) is the definition of genuine $n$-partite entanglement proposed in Refs.~\cite{schmid2020understanding,navascues2020genuine}, based on entanglement under Local Operations and Shared Randomness (LOSR) rather than Local Operations and Classical Communication (LOCC).\footnote{When multipartite entanglement is studied within the framework of resource theories, one can make different choices for the free operations relative to which it is considered a resource \cite{schmid2020understanding}.  We summarize the main distinctions here.  If each party can implement any unitary operation in their local lab,   then we refer to the free operations as 
   {\em Local Unitaries}  (LU). If they can implement any quantum operation, then we refer to the free operations as   {\em Local Operations} (LO). 
  If all parties have access to a common source of classical randomness (i.e., a dice roll or a coin flip) that produces an outcome which may influence the operation that each party implements, then  we refer to the free operations as 
{\em  Local Operations with Shared Randomness} (LOSR). If, in addition to this shared randomness, the parties have access to bipartite entangled quantum states distributed pairwise between parties, then we refer to the free operations as
{\em   Local Operations with Shared Randomness and $2$-Way Shared Entanglement} (LOSR2WSE). If no shared randomness is available but pairwise entanglement is, then we refer to the free operations as LO2WSE and if the allowed operations are restricted to being unitary, then we refer to the free operations as LU2WSE. If the parties do not have access to pairwise entanglement but do have access to classical communication channels between them (which subsumes having access to shared randomness), then we refer to the free operations as {\em Local Operations and Classical Communication} (LOCC).} This definition is distinct from the conventional definition (the one based on biseparability), and overcomes various conceptual problems of the latter, as argued in Ref.~\cite{schmid2020understanding}.   Consider the case of $n=3$ for simplicity. The LOSR-based definition in this case is that a state is genuine tripartite entangled if it cannot be realized by resources of bipartite entanglement.  
More precisely, it is proposed that a genuine tripartite entangled state is one that cannot be realized by Local Operations and Shared Randomness with 2-Way Shared Entanglement (LOSR2WSE). If the state to be prepared is  {\em pure}, the shared randomness does not add any additional power\footnote{To see this, we suppose that one can prepare the pure state $|\Psi\rangle_{ABC}$ using LOSR2WSE and show that it follows that one can also prepare it without making use of the shared randomness, i.e., with just LO2WSE.  Suppose the shared randomness consists of sampling a variable $\Lambda$ from a distribution $P_{\Lambda}$ and that in the LOSR2WSE protocol, the state prepared for the value $\Lambda=\lambda$ is denoted $|\psi_{\lambda}\rangle_{ABC}$, such that the overall state prepared in the protocol is $\rho_{ABC}= \sum_{\lambda} P_{\Lambda}(\lambda) \ket{\psi_{\lambda}}\!\! \bra{\psi_{\lambda}}_{ABC}$.  By virtue of the convex extremality of the state $|\Psi\rangle_{ABC}$, the only way to have $\rho_{ABC}= |\Psi\rangle\! \langle\Psi |_{ABC}$ is if  $|\psi_{\lambda}\rangle_{ABC}=  |\Psi\rangle_{ABC}$ for all $\lambda$ that are assigned nonzero probability by  $P_{\Lambda}$.  But this implies that the protocol allows for the preparation $|\Psi\rangle_{ABC}$ without use of the shared randomness.}, so we have that a pure state is genuinely tripartite entangled when it cannot be realized by Local Operations with 2-Way Shared Entanglement (LO2WSE).  But these operations are precisely what can be achieved in the triangle scenario, so a pure state is genuinely tripartite entangled when it is incompatible with the triangle scenario.  Characterizing the boundary between pure states that are genuinely tripartite entangled and those that are not, therefore, provides a motivation for characterizing the boundary between pure states that are compatible with the triangle scenario and those that are not.

 The notion of compatibility with the triangle scenario studied here is similar to, but distinct from, the notion of compatibility studied in Ref.~\cite{Kraft_21} (where they refer to the triangle scenario as the {\em independent triangle network}).  This is because the latter article considered only unitaries for the channels mapping the latent systems to the visible systems, rather than allowing arbitrary operations as we do here.\footnote{Ref.~\cite{Kraft_21} did not allow for tracing out subsystems, which is why the two approaches are distinct in spite of the Stinespring dilation theorem.} That is, Ref.~\cite{Kraft_21} considers state-preparability by Local {\em Unitaries} with Shared Randomness and 2-Way Shared Entanglement (LUSR2WSE), while here we consider state-preparability by LO2WSE. When considering the preparability of {\em pure} multipartite states, which we do in~\Cref{sec:qubit_nodes} for the triangle scenario, the shared randomness is irrelevant. In such a case, then, the differences between the results of Ref.~\cite{Kraft_21} and our own consists in the differences between state-preparability by LU2WSE and state-preparability by LO2WSE. Because the set of local unitaries is a strict subset of the set of local operations, any state shown to be preparable by LU2WSE is also preparable by LO2WSE, but it is unclear whether or not the opposite implication holds. As such, it is possible that there is a strict inclusion relation between the set of LU2WSE-preparable states and the set of LO2WSE-preparable states. Consequently, any demonstration of incompatibility relative to LU2WSE does not necessarily imply incompatibility relative to LO2WSE. Below, we develop methods that allow for demonstrating the latter type of incompatibility.

The remainder of this article is structured as follows.~\Cref{sec:prelims} provides a review of the necessary background on causal models, classical and quantum.~\Cref{sec:inflation} reviews the inflation technique for causal inference and shows how to leverage a particular family of operator inequalities related to the quantum marginal problem to derive constraints on compatibility of a quantum state with a causal structure. Application of these new constraints are presented in~\Cref{sec:qubit_nodes} and~\Cref{sec:higher_dim_nodes} for qubit and higher-dimensional visible nodes respectively, with proofs presented in the appendices. In~\Cref{sec:beyond_triangle}, we present a family of networks where causal compatibility constraints can be obtained via the methods developed in earlier sections.~\Cref{sec:dicussion} discusses avenues for generalizing our results and conclusions.

\section{Preliminaries} \label{sec:prelims}

Throughout this manuscript, we make frequent reference to directed acyclic graphs (DAGs) as the primary method for representing causal structure \cite{Pearl_2009,spirtes2001causation}, so it is worthwhile establishing some of the relevant notation for these DAGs here. DAGs will typically be denoted by the letter $G$ and variations thereof, such as $G'$. The set of nodes of $G$ will be denoted $\mathsf{Nodes}(G)$, and partitions into two subsets: the visible nodes, 
denoted  $\textsf{Vnodes}(G)$, and the latent nodes, denoted $\textsf{Lnodes}(G)$. A DAG $G$ with a partitioning of nodes in this way was termed a {\em partitioned DAG} in Ref.~\cite{ansanelli2024everything}. All nodes will be labelled by capital letters $A$, $B$, $C$, etc., with letters at the start of the alphabet ($A$, $B$, $C$) usually denoting visible nodes and with letters around the middle ($L$, $M$, $N$) denoting latent nodes. For each node $X  \in \textsf{Nodes}(G)$, $\Pa(X) \subset \mathsf{Nodes}(G)$ denotes the set of parents of $X$ in $G$, that is, the set of nodes of $G$ for which there is a directed edge from the node to $X$. Similarly, for each $X \in \textsf{Nodes}(G)$, $\textrm{Ch}(X) \subset \mathsf{Nodes}(G)$ denotes the set of children of $X$ in $G$, that is, the set of nodes of $G$ for which there is a directed edge from $X$ to the node. A node $X \in \mathsf{Nodes}(G)$ is called an {\em ancestor} of a node $Y \in \mathsf{Nodes}(G)$ if there is a directed path from $X$ to $Y$ in $G$.

\subsection{Classical causal models}
 \label{subsec:cc_inflation}

To begin with, consider what is arguably the simplest case:  {\em classical} semantics for both the visible and the latent nodes. In the case of the triangle scenario, this is depicted in \Cref{fig:classicaltriangle}.  This means that both the visible and the latent nodes are associated to random variables. A joint state on the visible nodes is therefore a joint {\em probability distribution} over the associated random variables.  
  
\begin{figure}[htbp]
\begin{center}
\includegraphics[width=0.30\textwidth]{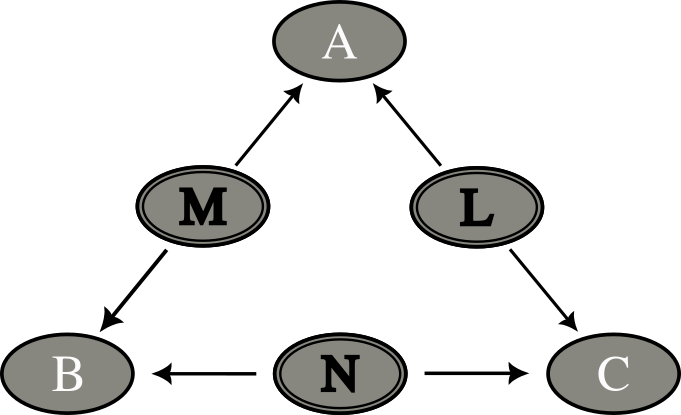}
\caption{The DAG describing the triangle scenario with classical semantics for the visible nodes, labelled $A$, $B$ and $C$, and the latent nodes, labelled $L$, $M$ and $N$.}
\label{fig:classicaltriangle}
\end{center}
\end{figure}

Consider a DAG $G$ describing the causal structure. We can also presume that $G$ is in a canonical form wherein latent nodes are necessarily exogenous, that is, they have no parents in $G$.\footnote{ There is no loss of generality in doing so in the case of classical causal modelling because Ref.~\cite{Evans16} showed that every partitioned DAG is observationally equivalent to one wherein all the latent nodes are exogenous.}  

For a classical causal model, the state on the visible nodes is a joint distribution  $P_{\textsf{Vnodes}(G)}$.  To specify a choice of parameter values in the causal model is to specify a distribution over each latent variable, i.e., a set $\{ P_{L} : L \in \textsf{Lnodes}(G) \}$ (here, the cardinality of the set of values of each $L$ is assumed to be finite but allowed to be arbitrarily large), and a conditional probability for each visible variable given its causal parents $\{  P_{A |\Pa(A)} : A \in \textsf{Vnodes}(G)\}$   (where the parental sets may include both visible and latent variables)\footnote{A comment regarding notation: if $\Pa(A)$ is empty, then $P_{A|\Pa(A)} = P_{A}$.}.

\begin{Definition}[Classical $G$-compatibility of a joint distribution] 
Consider a joint distribution $P_{\textsf{Vnodes}(G)}$ over the variables associated to the visible nodes of a causal structure $G$.  We say that $P_{\textsf{Vnodes}(G)}$ is {\em classically $G$-compatible} if there exists a choice of the parameters in the classical causal model, that is, a choice of $\{ P_{L} : L \in \textsf{Lnodes}(G) \}$ (distributions on the classical systems associated to the latent nodes) and of $\{  P_{A |\Pa(A)} : A \in \textsf{Vnodes}(G)\}$ (conditional probability distributions defining a stochastic map whose inputs are the classical systems associated to the latent nodes and whose outputs are the classical systems associated to the visible nodes) such that these realize $P_{\textsf{Vnodes}(G)}$:
 \begin{equation}\label{eq:classicalcausalmodeljoint}
 P_{\textsf{Vnodes}(G)} = \sum_{L\in\textsf{Lnodes}(G)}\prod_{A \in \textsf{Vnodes}(G)}P_{A|\Pa(A)}   \prod_{L\in\textsf{Lnodes}(G)} P_{L}.
\end{equation}
where $\sum_{L\in\textsf{Lnodes}(G)}$ indicates marginalisation over the variables assigned to latent nodes. 
\end{Definition}

We now consider the concrete example of the triangle scenario.  We use the notation for nodes presented in Fig.~\ref{fig:classicaltriangle}, that is, the visible nodes are denoted $A,B,C$ and the latent nodes are denoted $L,M,N$, and the variables associated to these nodes are denoted in the same manner as the nodes themselves. A joint distribution $P_{A BC}$ is compatible with the triangle scenario in a classical causal model if there exists a choice of parameters $P_{A|LM}, P_{B|MN}, P_{C|LN}, P_{L}, P_{M}$, and $P_{N}$ such that
\begin{align}
P_{A BC} = \sum_{LMN} P_{A|LM} P_{B|MN} P_{C|LN} P_{L} P_{M} P_{N}.
\end{align}

\subsection{Quantum causal models}\label{Qcausalmodels}

We now consider quantum causal models~\cite{Allen_17,Costa_2016,barrett2019quantum}, where both the latent and visible nodes have quantum semantics.  
 
In the case of the triangle scenario, this is depicted in Fig.~\ref{fig:triangleandinflations}(a).

For simplicity, we restrict our attention to causal structures wherein no two visible nodes appear in a cause-effect relationship (i.e., one visible node is never in the causal ancestry of another visible node) and wherein all latent nodes are exogenous (i.e., have no parents in the causal structure). This pair of restrictions together imply that the causal structure has just two layers: a layer of exogenous latent nodes and a layer of visible nodes. Such causal structures are sometimes termed {\em networks}~\cite{fritz2012beyond,wolfe_q_inflation} and we shall adopt this terminology here.\footnote{This terminology can also be motivated by recent developments of quantum networks within the context of the quantum internet~\cite{kimble2008quantum,wehner2018quantum}. For example, several proposals exist for satellite-based ``downlink'' quantum networks (see, e.g., Refs.~\cite{Liao_18,bourgoin2013comprehensive}) whereby satellites act as sources distributing entangled states to separated receiver stations on the ground. In such cases, if the satellites are considered to be the latent nodes of the network and the ground stations the visible nodes, and if the lack of communication between these nodes can be justified on practical grounds (i.e., due to the separation of satellites in orbit or the stations on the ground \cite{simon2017towards}), then the topology of these downlink networks matches the type of causal structures termed networks here.}

In contrast with a classical causal model, which assigns a {\em variable} to each node (the important aspect of which is the cardinality of its set of values), a quantum causal model assigns a {\em quantum system} to each node (the important aspect of which is the dimension of Hilbert space associated to it).  Consequently, to specify a choice of parameter values in a quantum causal model is to specify two types of parameters: (i) quantum states on the systems associated to the latent nodes (where the dimension of the Hilbert space associated to each latent node is assumed to be finite but allowed to be arbitrarily large)\footnote{The assumption that the latent spaces are finite-dimensional derives from nuances regarding the definition of causal compatibility that occur if this assumption is dropped --- see the comments below Def.~\ref{def:G_comp_QCM}.}, and (ii) quantum channels (i.e., completely positive and trace-preserving linear maps) from these to the quantum systems associated to the visible nodes.  

It follows that the entity whose compatibility with a causal structure is to be assessed is no longer a joint probability distribution, but a joint quantum state, specifically, $\rho_{\textsf{Vnodes}(G)}$, i.e., a positive trace-one operator on the Hilbert space $\bigotimes_{A\in \textsf{Vnodes}(G)}\mathcal{H}_{A}$.

Because of the no-broadcasting theorem in quantum theory~\cite{Barnum_96}, if a given latent node influences more than one visible node, one cannot imagine that  the state of that node is broadcast to each of the children, as one could classically.  The main differences between the definitions of the notion of a quantum causal model proposed in various works, such as Ref.~\cite{Costa_2016} and Refs.~\cite{Allen_17,barrett2019quantum}, relate to how to meet this challenge.  However, disagreements about the definition are not relevant to the compatibility problem. This is because the definition of the causal compatibility problem allows the Hilbert spaces associated to the latent nodes to be of arbitrary (finite)  dimensionality, and  the differences between these approaches arise only when one fixes the dimensionality. We here adopt the following approach: Suppose that the set of visible nodes that are children of a latent node $L$ are denoted $\textrm{Ch}(L)$.  Then the Hilbert space associated to the latent node $L$,  $\mathcal{H}_{L}$, is presumed to factorize into a set of Hilbert spaces indexed by its children.  That is, $\mathcal{H}_{L} = \bigotimes_{A \in\textrm{Ch}(L)} \mathcal{H}_{L^{(A)}}$. The first type of parameter in the causal model is a quantum state $\rho_{L}$, i.e., a positive trace-one operator on the space of linear operators on the Hilbert space associated to $L$,  $\rho_{L} \in \mathcal{L}(\mathcal{H}_{L})$.  Given the factorization just described, the state $\rho_{L}$ can be conceptualized as $\rho_{L^{(A)}, L^{(B)}\dots L^{(Z)}} \in \mathcal{L}(\mathcal{H}_{L^{(A)}} \otimes \mathcal{H}_{L^{(B)}}\otimes  \dots \otimes \mathcal{H}_{L^{(Z)}})$ where $A,B,\dots, Z \in \textrm{Ch}(L)$. Such a state can clearly describe correlations between the different subsystems of $L$. 

The second type of parameter of the causal model specifies how each visible node  depends on its causal parents. For each visible node  $A\in \textsf{Vnodes}(G)$, one specifies a quantum channel whose output space is $\mathcal{L}(\mathcal{H}_{A})$, the operator space associated to $A$.  Let $\widetilde{\textrm{Pa}}(A)$ be defined as $\{ L^{(A)} : L \in \textrm{Pa}(A) \}$.  
That is, for each latent node $L$ among the parents of $A$, include within $\widetilde{\textrm{Pa}}(A)$ only $L^{(A)}$, the subsystem of $L$ that influences $A$. Thus, whereas $\textrm{Pa}(A)$ describes the collection of latent nodes that are parents of $A$, $\widetilde{\textrm{Pa}}(A)$ describes the collection of {\em subsystems} of the latent nodes that have an influence on $A$. The input space of the quantum channel whose output space is $\mathcal{L}(\mathcal{H}_{A})$ is $\mathcal{L}(\mathcal{H}_{\widetilde{\textrm{Pa}}(A)})$.  The channel can consequently be denoted $\mathcal{E}_{A|\widetilde{\textrm{Pa}}(A)}$. There is one such channel for each visible node. 

We can now define what it means for a joint quantum state on the visible nodes of a  network $G$ to be compatible with $G$ for quantum semantics of the latent nodes. 
\begin{Definition}[$G$-compatibility in quantum causal models] \label{def:G_comp_QCM}
Let $G$ be a  network and consider a joint quantum state $\rho_{\textsf{Vnodes}(G)}$ over the systems associated to the visible nodes of $G$. The state $\rho_{\textsf{Vnodes}(G)}$ is said to be $G$-compatible with the network $G$ within a quantum causal model if there exist a choice of the parameters of the model that realizes $\rho_{\textsf{Vnodes}(G)}$. More precisely, the joint state $\rho_{\textsf{Vnodes}(G)}$ is $G$-compatible if, for each latent node $L\in \textsf{Lnodes}(G)$, there exists a choice of quantum state $\rho_{L}$, and for each visible node $A \in \textsf{Vnodes}(G)$, there exists a choice of quantum channel $\mathcal{E}_{A|\widetilde{\textrm{Pa}}(A)}$, such that
\begin{equation}
\rho_{\textsf{Vnodes}(G)} =  \bigotimes_{A \in \textsf{Vnodes}(G)} \mathcal{E}_{A|\widetilde{\textrm{Pa}}(A)} \left( \bigotimes_{L\in \textsf{Lnodes}(G)} \rho_{L} \right).
\end{equation}
\end{Definition}

We briefly comment on the  choice of having the latent spaces factorize and what that means for the notion of $G$-compatibility defined above. An alternative choice would be to not assume that latent spaces factorize but instead enforce that the channel parameters $\mathcal{E}_{A|\textrm{Pa}(A)}$ mutually commute on common input spaces. As was demonstrated in Ref.~\cite{Allen_17}, this alternative is well-motivated by a quantum generalisation of the notion of common cause in classical causal modelling. However, in the context of quantum causal compatibility, if the latent spaces are finite dimensional (but of arbitrarily large dimension) as we consider here, there is, in fact, no distinction between the factorizing case and the commuting case. That is, any state $\rho$ that is $G$-compatible with respect to the former notion is also $G$-compatible with respect to the latter, and vice versa. If the latent spaces are instead infinite dimensional, then the results of, e.g., Refs.~\cite{slofstra2019set,ji2021mip} indicate that there may indeed be a distinction between the factorizing and commuting cases in such a scenario. For our present purposes, we consider finite (but arbitrary large) dimension of latent spaces only, which can be understood as defining the scope of application of the techniques developed below.

 To better understand the notion of $G$-compatibility in a quantum causal model, it is worthwhile once more to consider the concrete example of the triangle scenario, using the notational convention of Fig.~\ref{fig:triangleandinflations}(a). In this case, the Hilbert spaces of each latent node factorize into a pair of subspaces, i.e., $\H_{L} \cong \H_{\LA} \otimes \H_{\LC}$, $\H_{M} \cong \H_{\MA} \otimes \H_{\MB}$, $\H_{N} \cong \H_{\NB} \otimes \H_{\NC}$.  A quantum state $\rho_{ABC}$ on $\H_{A} \otimes \H_{B} \otimes \H_{C}$ is \textit{triangle-compatible} if there exist quantum states of the latent nodes $\varrho_{L} = \varrho_{\LA \LC}$, $\varrho_{M} = \varrho_{\MA \MB}$, $\varrho_{N} = \varrho_{\NB \NC}$ and quantum channels $\E_{A|\LA \MA}: \H_{\LA} \otimes \H_{\MA} \rightarrow \H_{A}$, $\E_{B|\MB \NB}: \H_{\MB}\otimes \H_{\NB} \rightarrow \H_{B}$, and $\E_{C|\LC \NC}: \H_{\LC} \otimes \H_{\NC} \rightarrow \H_{C}$ such that\footnote{We note that the definition of a state being compatible with the triangle scenario using quantum latents, asks about the existence of arbitrary quantum channels $\E_{A|\LA \MA}, \E_{B|\MB \NB}$ and $\E_{C|\LC \NC}$ in Eq.~\eqref{eq:q_triangle_compat}.  In Refs.~\cite{Kraft_21,luo2021new}, by contrast, these quantum channels are restricted to be unitary channels.  Consequently, the notion of compatibility defined here is more permissive than the notion of compatibility proposed in Refs.~\cite{Kraft_21,luo2021new}.  Because we cannot see a good reason to restrict these channels to unitaries, we consider the notion of compatibility  described above  to be the fundamental one, rather than that of Refs.~\cite{Kraft_21,luo2021new}.  Nonetheless, there are cases where the two definitions coincide, such as when studying compatibility of pure states. See Sec.~\ref{subsec:qubit_pure} for a discussion of this case.} 
\begin{align}\label{eq:q_triangle_compat}
 \rho_{ABC} = & \E_{A|\LA \MA} \otimes \E_{B|\MB \NB}  \\\nonumber
&  \otimes  \E_{C|\LC \NC}(\varrho_{\LA \LC} 
\otimes \varrho_{\MA \MB} \otimes \varrho_{\NB \NC}). 
\end{align}

As an aside, note that the partitioning of the latent systems into subsystems, each of which influences a different child node of the latent node, also provides an alternative way of defining a classical causal model.  Take the triangle scenario as an example.  The parametrization of the classical causal model that we provided earlier was 
  $P_{A|LM}, P_{B|MN}, P_{C|LN}, P_{L}, P_{M}$, and $P_{N}$ such that
\begin{align}
P_{A BC} = \sum_{LMN} P_{A|LM} P_{B|MN} P_{C|LN} P_{L} P_{M} P_{N}.
\end{align}
However, one can also partition each latent variable into a pair of latent variables and make each child depend on just one element of the pair. For instance, one defines $L=(\LA,\LC)$ and one makes $A$ depend only on $\LA$  and $C$ depend only on $\LC$.  The parametrization of the classical causal model we obtain in this way is 
  $P_{A|\LA \MA}, P_{B|\MB \NB}, P_{C|\LC \NC}, P_{\LA \LC}, P_{\MA \MB}$, and $P_{\NB \NC}$  such that
\begin{align}
 P_{A BC} &= \sum_{L M,N}
  P_{A|\LA \MA} P_{B|\MB \NB} P_{C|\LC \NC}\nonumber\\
&  \times P_{\LA \LC} P_{\MA \MB} P_{\NB \NC}. 
\end{align}
This second parametrization is clearly subsumed as a special case of the first. To see that the first is also subsumed as a special case of the second, it suffices to note that one can take $\LA$  and $\LC$ to be two copies of the $L$ of the first model and to take the distribution over $\LA \LC$ to be what one obtains from the distribution $P_L$ by applying a copy operation, so that $P_{\LA \LC} = P_{L} \delta_{\LA, L} \delta_{\LC, L}$.  A similar observation holds for any causal structure. 

This demonstrates that one can recover a classical causal model from a quantum causal model by particularizing to density operators  that are diagonal in some fixed product basis over the subsystems, since these are equivalent to probability distributions, and channels that act as stochastic maps relative to these bases, since these are equivalent to conditional probability distributions.

\section{Fully quantum inflation}
\label{sec:inflation}

\subsection{A description of the technique}

The inflation technique for causal inference~\cite{WolfeSpekkensFritz_2019} provides a method for resolving questions of causal compatibility.  This method was inspired by ideas from the literature on Bell's theorem, namely, derivations of Bell's theorem that leverage results from the classical marginal problem.  

At a high level, the inflation technique can be conceptualized as a means for converting facts about a marginal problem into facts about causal compatibility.  We explore a particular version of this problem in the present paper, namely, how to make use of facts about the {\em quantum} marginal problem to obtain conclusions about causal compatibility between a causal structure and a joint quantum state. 

In the following, we will generalize the inflation technique to the case of quantum causal models. 
Note that previous work allowed the {\em latent nodes} to be quantum, but a description of the inflation technique for the case where the {\em visible} nodes are quantum has not been provided previously. It is worth noting a subtlety regarding terminology at this point. Whether a causal model is termed `classical' or `quantum' refers to whether all the nodes are classical or all are quantum. The term `quantum inflation', on the other hand, has been used~\cite{wolfe_q_inflation} to describe the inflation technique wherein the latent nodes are quantum but the visible nodes are still classical  (and in particular the case where one can treat fanout inflations).  One consequently cannot refer to the inflation technique applied to quantum causal models as simply `quantum inflation'. Therefore, we introduce the term `fully quantum inflation' to refer to the technique when both the latent and the visible nodes are quantum.

It is useful to contrast the classical and quantum cases.  We will not, however, review the inflation technique for classical causal models.  Rather, we will present the technique only for quantum causal models.  The description of the technique for classical causal models can be recovered from the quantum one, as noted at the end of Sec.~\ref{Qcausalmodels}, by simply restricting attention to density operators that are diagonal in some fixed product basis over the subsystems and channels that act as stochastic maps relative to these bases. 

The first step in using the inflation technique consists of producing, from the original DAG $G$, a new DAG $G'$, termed the {\em inflation} of $G$, in the following way (the reader is referred to Ref.~\cite{WolfeSpekkensFritz_2019} for a more complete treatment). The set of nodes of $G'$, denoted $\mathsf{Nodes}(G')$,  are labelled in such a way that each can be associated to a node of $G$, with the possibility that more than one node of $G'$ is associated to  a single node in $G$. That is, the inflation $G'$ comes with a fixed surjective map $\mathsf{DropCopyIndex}_{G' \rightarrow G} : \mathsf{Nodes}(G') \rightarrow \mathsf{Nodes}(G)$. If the nodes of $G$ are labelled as $A,B,\dots,Z$ then the nodes of $G'$ can be labelled as $A_{i}, B_{j}, \dots, Z_{k}$ where $i$, $j$ and $k$ are indices in the natural numbers whose ranges depend on the choice of $G'$. The surjective map is then defined via $\mathsf{DropCopyIndex}_{G' \rightarrow G}(A_{i}) = A$ for all $i$ and similarly for the other nodes. We refer to the nodes $A_{i}$ of $G'$ that are associated to a given node $A$ of $G$ as {\em copies} of the latter, and the subscript $i$ is referred to as the {\em copy-index}. 

For certain subsets of nodes of $G'$ and of $G$, the restriction of the map $\mathsf{DropCopyIndex}_{G' \rightarrow G}$ to these subsets forms a bijection. Suppose that ${\bf X} \subseteq \mathsf{Nodes}(G)$ and ${\bf X}' \subseteq \mathsf{Nodes}(G')$ are such that, for every $A \in {\bf X}$ there is precisely one value of $i$ such that $A_{i} \in {\bf X}'$. In such a case, the restriction of $\mathsf{DropCopyIndex}_{G' \rightarrow G}$ to ${\bf X}'$ forms a bijection between ${\bf X}'$ and ${\bf X}$, and we write ${\bf X} \sim {\bf X}'$. This bijection can be considered as capturing ``sameness of the labels of the nodes up to the identity of the copy-index'' (where no copy-index is considered to be one possible choice). 

Furthermore, for certain choices of $G'$, ${\bf X}' \in \mathsf{Nodes}(G')$ and ${\bf X} \in \mathsf{Nodes}(G)$ such that ${\bf X} \sim {\bf X}'$, the bijection can be considered to be structure-preserving in the following sense. For a given graph $G$ and a subset ${\bf X} \in \mathsf{Nodes}(G)$, we denote by $\mathsf{SubDAG}_G ({\bf X})$ the sub-DAG of $G$ consisting of those nodes of $G$ contained in ${\bf X}$ and those directed edges of $G$ that connect pairs of nodes within ${\bf X}$ (analogous notation is used for $G'$ and ${\bf X}' \in \mathsf{Nodes}(G')$). We then say that the map $\mathsf{DropCopyIndex}_{G' \rightarrow G}$ is a structure-preserving bijection between $\mathsf{SubDAG}_{G'}({\bf X}')$ and $\mathsf{SubDAG}_G ({\bf X})$ if the restriction of $\mathsf{DropCopyIndex}_{G' \rightarrow G}$ to ${\bf X}'$ is a bijection between ${\bf X}'$ and ${\bf X}$, and if, for $A_{i}, B_{j} \in {\bf X}'$, the directed edge from $A_{i}$ to $B_{j}$ exists in $\mathsf{SubDAG}_{G'}({\bf X}')$ if and only if the directed edge from  $A$ to $B$ exists in $\mathsf{SubDAG}_{G}({\bf X})$. In such a case we write $\mathsf{SubDAG}_{G'} ({\bf X}' ) \sim \mathsf{SubDAG}_G ({\bf X})$. This bijection captures ``sameness of the sub-DAGs (both labels and connectivity) up to the identity of the copy-index''. 

Next, let us consider subsets of nodes ${\bf U}' \subseteq \mathsf{Nodes}(G')$ and ${\bf U} \subseteq \mathsf{Nodes}(G)$. Let ${\bf X}'_{\bf{U}'} \subseteq \mathsf{Nodes}(G')$ be the set of nodes comprised of ${\bf U}'$ and all ancestors of the nodes of ${\bf U}'$ in $G'$ and similarly for ${\bf X}_{\bf{U}} \subseteq \mathsf{Nodes}(G)$. If $\mathsf{SubDAG}_{G'} ({\bf X}'_{\bf{U}'}) \sim \mathsf{SubDAG}_G ({\bf X}_{\bf{U}})$ in the sense outlined above, then we write
\begin{align}
\label{eq:definjectable}
\textsf{ansubgraph}_{G'}(\bm{U}')\sim\textsf{ansubgraph}_G(\bm{U}),
\end{align}
where $\textsf{ansubgraph}_G(\bm{U})$ is called the ancestral subgraph of $\bm{U}$ in $G$. In this case, the ancestral subgraph of ${\bf U}'$ in $G'$ and that of ${\bf U}$ in $G$ are the same up to the identity of the copy-indices.

There are two important notions from the inflation technique that are defined with respect to the correspondence in \Cref{eq:definjectable}. The first is the notion of injectable set, defined as follows. If ${\bf U}'$ consists entirely of visible nodes of $G'$, i.e., ${\bf U}' \subseteq \textsf{Vnodes}(G')$, and ${\bf U} \subseteq \textsf{Vnodes}(G)$ is such that $\textsf{ansubgraph}_{G'}(\bm{U}')\sim\textsf{ansubgraph}_G(\bm{U})$, then ${\bf U}'$ is called an {\em injectable set} and ${\bf U}$ is the {\em image of this injectable set} under the dropping of copy indices. The set of injectable sets of $G'$ is denoted $\mathsf{InjectableSets}(G')$ and the set of all their images in $G$ is denoted $\mathsf{ImagesInjectableSets}(G)$.

The second is the very notion of an inflation of a DAG. A DAG $G'$ is an {\em inflation of $G$} if and only if, for every $A \in \mathsf{Nodes}(G)$ and $A_{i} \in \mathsf{Nodes}(G')$ (so that $\textsf{DropCopyIndex}_{G' \rightarrow G}(A_{i}) = A$), $\mathsf{ansubgraph}_{G'}(A_i)\sim
\mathsf{ansubgraph}_G(A)$.  In other words, for $G'$ to be an inflation of $G$, every {\em singleton} set of visible nodes in $G'$ must be an injectable set.  The set of DAGs that are inflations of $G$ is denoted $\mathsf{Inflations}(G)$.

To explain these concepts further, let us consider an inflation used in the remainder of this work, namely the {\em $AB$-Cut inflation} of the triangle scenario, depicted in~\Cref{fig:triangleandinflations}(b). In this case, the inflated DAG contains just a single copy of each of the visible nodes, $A,B,C$, and a single copy of two of the latent nodes,  namely, $L$ and $N$, but two copies of the latent node $M$.  The latter are denoted by $M_{1}$ and $M_{2}$, while the former are denoted by $A_{1},B_{1},C_{1},L_{1},N_{1}$. Consider next the requirement of matching the ancestral subgraphs of visible nodes. The ancestry of each of the visible nodes in the original triangle consists only of the node's parents, so that the ancestral sub-DAG of node $A$ is the $v$-shaped ``collider'' DAG with arrows pointing from latent nodes $L$ and $M$ to $A$ (and similarly for $B$ with parents $L,N$ and for $C$ with parents $M,N$). In the Cut inflation, the situation is the same: the visible node $A_{1}$ has ancestry given by the $v$-shaped DAG with parents $L_{1}$ and $M_{2}$, $B_{1}$ has parents $M_{1}$ and $N_{1}$, and $C_{1}$ has parents $L_{1}$ and $N_{1}$. By removing the copy-indices (the subscripts) these ancestral sub-DAGs exactly match those of $A$, $B$ and $C$. By the same token, it can be seen that $\{A_{1}, C_{1}\}$ is an injectable set: the ancestral subgraph of $\{A_{1}, C_{1}\}$ in $G'$ is the $w$-shaped subgraph made up of $A_{1}, C_{1}$, their parents $L_{1}, M_{2}, N_{1}$ and the arrows between them, which corresponds precisely to the subgraph made up of $A,B,L,M,N$ in $G$. Likewise, $\{B_{1}, C_{1}\}$ is also an injectable set.

However, in the $AB$-Cut inflation of the triangle scenario, the set $\{A_{1},B_{1}\}$ is {\em not} injectable. This can be seen by noting that, in $G'$, $A_{1}$ and $B_{1}$ share no common ancestor, whereas $A$ and $B$ {\em do} share an ancestor in $G$ (i.e., $M$). While this ancestral independence disallows $\{A_{1}, B_{1}\}$ from being an injectable set, it {\em is} an example of a related notion, that of an {\em expressible} set. In the context of this work, namely, when considering DAGs $G$ that have the connectivity of a network, an {\em expressible} set of $G' \in \mathsf{Inflations}(G)$ is any set that can be written as the union of injectable sets that share no common ancestors in $G'$. That is, a subset $\bm{Y}' \subseteq \textsf{Vnodes}(G')$ is an expressible set, i.e., $\bm{Y}'\in\textsf{ExpressibleSets}(G')$, if there exist $\bm{U}'_{(1)},\dots, \bm{U}'_{(y)} \in\textsf{InjectableSets}(G')$ such that $\bm{Y}' = \cup_{j=1}^{y} \bm{U}'_{(j)} $ and such that for any pair $j,k$ such that $j\ne k$, $\textsf{ansubgraph}_{G'}(\bm{U}'_{(j)})\cap \textsf{ansubgraph}_{G'}(\bm{U}'_{(k)})= \emptyset$.\footnote{In particular, this means that $\bm{U}'_{(j)}$ and $\bm{U}'_{(k)}$ are $d$-separated by the empty set in $G'$. Also note that, in Ref.~\cite{WolfeSpekkensFritz_2019}, a more general definition of expressible set was given, with the notion presented here being referred to as ``ai-expressibility'' (with ``ai'' standing for ``ancestrally independent''). However, in the case where $G$ is a network, the set of ai-expressible sets is equal to the set of all expressible sets. As we focus exclusively on networks in this article, we refrain from providing the full definition here.}
All injectable sets are trivial examples of expressible sets.

Up until this point in our description of the inflation technique, we have not needed to stipulate whether we are considering a classical or quantum causal model (or indeed a causal model for some other generalized probabilistic theory), a distinction that we refer to as the {\em semantics} of the causal model.   The semantics specifies the nature of the parameters of the causal model.  From this point on, however, what we say is specific to quantum causal models.  As noted above, the case of classical causal models can be recovered by restricting attention to density operators that are diagonal relative to some fixed product basis over the subsystems and channels that act as stochastic maps relative to this basis.

We now introduce a map from the {\em parameters} of a causal model on $G$, denoted \texttt{par}, to the parameters  of a causal model on $G'$ (an inflation of $G$), denoted \texttt{par}$'$.  In other words, we describe how inflation acts on the parameters of a model. We denote this map from \texttt{par} to \texttt{par}$'$ by $\mathsf{Inflation}_{G\to G'}$.\footnote{Note that in Ref.~\cite{WolfeSpekkensFritz_2019}, the pair consisting of $G$ and a particular choice of parameter values, \texttt{par}, was referred to as {\em a causal model}.  For physicists, the term `model' appears natural as the description of a {\em particular} choice of parameter values.  For statisticians, however, the term `model' naturally corresponds to a {\em set} of possible such choices.  To reduce confusion, we here use the term `causal model' in a manner that is more in keeping with the statistician's usage.}

The map $\mathsf{Inflation}_{G\to G'}$ taking a choice of parameters on $G$ to a choice of parameters on  $G'\in \textsf{Inflations}(G)$ is defined by the following conditions: (i) For every visible node $A_i$ in $G'$, the quantum channel relating $A_i$ to its parents within $G'$ is the same as the quantum channel relating $A$ to its parents within $G$, 
\begin{align}\label{eq:funcdependences}
 \forall A_i \in \textsf{Vnodes}(G'):\; \mathcal{E}_{A_i| \widetilde{\textrm{Pa}}_{G'}(A_i)} =\mathcal{E}_{A|\widetilde{\textrm{Pa}}_{G}(A)}.
\end{align}
(ii) For every latent node $L_{j}$ in $G'$, let $\textrm{Ch}(L_{j})$ denote the set of children of $L_{j}$ in $G'$ and let $\widehat{\textrm{Ch}}(L_{j}) := \mathsf{DropCopyIndex}_{G' \rightarrow G}(\textrm{Ch}(L_{j})) \subseteq \textrm{Ch}(L)$ denote the set of visible nodes of $G$ that are the same up to copy-index as the nodes of $\textrm{Ch}(L_{j})$. The state on $L_{j}$ is then taken to be the reduced state of the state $\rho_{L}$ on $L$ in $G$ on the subsystems indexed by the elements of $\widehat{\textrm{Ch}}(L_{j})$. That is,
\begin{align}\label{eq:funcdependences2}
 \forall L_j \in \textsf{Lnodes}(G'):\; \rho_{L_j } = \textrm{Tr}_{L^{(A)} : A \in \textrm{Ch}(L) \setminus \widehat{\textrm{Ch}}(L_{j})}[\rho_{L}]
\end{align}
with the understanding that, if $\widehat{\textrm{Ch}}(L_{j}) = \textrm{Ch}(L)$, then $\rho_{L_{j}} = \rho_{L}$.\footnote{The map $\mathsf{Inflation}_{G\to G'}$ has been defined in this way since we focus on networks $G$ in this work, meaning that no latent node of $G$ has any parents. It should be noted, however, that there are DAGs $G$ where latent nodes {\em do} have parents in $G$. To cover this more general case, the inflation map can be defined without distinguishing between latent and visible nodes as is done above: for each $X_{i} \in \textsf{Nodes}(G')$, $\mathsf{Inflation}_{G\to G'}$ assigns to $X_{i}$ the parameter $\mathcal{E}_{X_i| \widetilde{\textrm{Pa}}_{G'}(X_i)} = \mathcal{E}_{X|\widetilde{\textrm{Pa}}_{G}(X)}$ with the understanding that (a) if $\widetilde{\textrm{Pa}}_{G'}(X_i)$ is empty, the map is a state-preparation, and (b) if the children of $X_{i}$ in $G'$ is a subset of the set of children of $X$ in $G$ (up to sameness of copy-indices), then there is a partial trace over the unused subspaces of $X_{i}$ to produce the reduced state on those subspaces corresponding to $\textrm{Ch}(X_{i})$, as in (ii) above.} 

So, why is the inflation map --- and the broader machinery of the inflation technique --- relevant for problems of causal compatibility? In essence, it allows a question of compatibility with $G$ to be mapped to a question of compatiblity with an inflation $G'$ of $G$. The formal statement of this is given in Proposition~\ref{mainlemma} below. 

Before stating and proving the proposition, let us establish some of the notation to be used throughout the remainder of this section. As above, $\bm{Y}'$ will continue to denote an expressible set of $G'$, i.e., $\bm{Y}' \in \mathsf{ExpressibleSets}(G')$. The set of elements $\bm{U}'_{(j)} \in \mathsf{InjectableSets}(G')$, for $j = 1, ..., y$, that are ancestrally independent and make up $\bm{Y}'$ is denoted $\mathsf{InjectableComponents}(\bm{Y}')$ (note that $y \in \mathbb{N}$ depends on $\bm{Y}'$). The images under the $\mathsf{DropCopyIndex}$ map of the $\bm{U}'_{(j)}$ in $\mathsf{ImagesInjectableSets}(G)$ for a given $\bm{Y}'$ are denoted $\bm{U}_{(j)}$, and the set of these is denoted $\mathsf{ImagesInjectableComponents}(\bm{Y}')$. We have the following:

\begin{Proposition} \label[proposition]{mainlemma}  
Let $G'$ be an inflation of the network $G$.
In this case, if the family of states $\{ \rho_{\bm{U}} : \bm{U} \in \mathsf{ImagesInjectableSets}(G) \}$ is compatible with $G$, then the family of states $\{ \sigma_{\bm{Y}'} : \bm{Y}' \in \mathsf{ExpressibleSets}(G') \}$, where 
 \begin{align}
 \sigma_{\bm{Y}'} := \bigotimes_{\bm{U} \in \mathsf{ImagesInjectableComponents}(\bm{Y}')} \rho_{\bm{U}},
 \end{align} 
 is compatible with $G'$. 
\end{Proposition}

\begin{proof} 
Suppose that $\{ \rho_{\bm{U}} : \bm{U} \in \mathsf{ImagesInjectableSets}(G) \}$ is compatible with $G$. That means there exists a choice of parameters $\texttt{par}$ on $G$ that realises a state $\rho_{\textsf{Vnodes}(G)}$ such that 
\begin{align}
\rho_{\bm{U}} = \textrm{Tr}_{\textsf{Vnodes}(G) \setminus \bm{U}}[\rho_{\textsf{Vnodes}(G)}]
\end{align} 
for all $\bm{U} \in \mathsf{ImagesInjectableSets}(G)$. Consider the parameters $\texttt{par}'$ on $G'$ given by $\texttt{par}' = \textsf{Inflation}_{G\rightarrow G'}(\texttt{par})$, which realise some state $\sigma_{\textsf{Vnodes}(G')}$ on the visible nodes of $G'$. For any $\bm{X'} \in \mathsf{InjectableSets}(G')$, with image $\bm{X} \in \mathsf{ImagesInjectableSets}(G)$, we have that 
\begin{align}
\sigma_{\bm{X}'} := \textrm{Tr}_{\textsf{Vnodes}(G')\setminus \bm{X'}}[\sigma_{\textsf{Vnodes}(G')}] \equiv \rho_{\bm{X}}. \label{eq:equivalence_via_inj_set}
\end{align}
This follows from the definition of injectable set and the definition of the map $\textsf{Inflation}_{G\rightarrow G'}$.

Suppose $\bm{Y}' \in \mathsf{ExpressibleSets}(G')$ is such that 
\begin{align}
{}&\mathsf{InjectableComponents}(\bm{Y}') \nonumber\\
&\quad \quad = \{ \bm{U}'_{(j)} \in \mathsf{InjectableSets}(G') |j = 1, \dots, y\}.
\end{align}
As a consequence of ancestral independence, we get that
\begin{align}
\sigma_{\bm{Y'}} &:= \textrm{Tr}_{\textsf{Vnodes}(G') \setminus \bm{Y}'}[\sigma_{\textsf{Vnodes}(G')}] \\
&= \bigotimes_{j=1}^{y} \; \textrm{Tr}_{\textsf{Vnodes}(G') \setminus \bm{U}'_{(j)}}[\sigma_{\textsf{Vnodes}(G')}]\\
&= \bigotimes_{j=1}^{y} \; \sigma_{\bm{U}'_{(j)}}.
\end{align}
Since each $\bm{U}'_{(j)}$ is injectable, with image $\bm{U}_{(j)}$, applying Eq.~(\ref{eq:equivalence_via_inj_set}) gives that $\sigma_{\bm{Y'}} = \bigotimes_{j=1}^{y} \rho_{\bm{U}_{(j)}}$. As the states $\sigma_{\bm{Y}'}$ for $\bm{Y}' \in \mathsf{ExpressibleSets}(G')$ are all marginals of the same state $\sigma_{\textsf{Vnodes}(G')}$ which is compatible with $G'$, it is also the case that the family $\{ \sigma_{\bm{Y}'} : \bm{Y}' \in \mathsf{ExpressibleSets}(G') \}$ is compatible with $G'$.
\end{proof}

An example illustrating this proposition for the case where $G$ is the triangle scenario and $G'$ is the $AB$-Cut inflation is given in \Cref{app:example_prop_1}.

Next, it is useful to define the notion of a causal compatibility inequality.

Let $G$ be a causal structure and let $X$ be a family of subsets of the visible nodes of $G$.  Let $I_X$ denote an inequality that operates on the corresponding family of states,  $\{ \rho_{\bm{X}}: \bm{X} \in X\}$.  Then $I_X$ is a {\em causal compatibility inequality for the causal structure $G$} (or a {\em $G$-compatibility inequality} for short) whenever it is satisfied by every family of states $\{ \rho_{\bm{X}}: \bm{X} \in X\}$ that is compatible with $G$.

Note that while violation of a causal compatibility inequality for $G$ witnesses the incompatibility of a state with the causal structure $G$, satisfaction of the inequality does not guarantee compatibility: satisfaction of the inequality merely provides a {\em necessary} condition for compatibility. 

Inflation provides a means of ``pulling back'' causal compatibility inequalities from the inflation DAG $G'$ to answer questions of causal compatibility with the original DAG $G$. This is formalized in the following corollary of Proposition~\ref{mainlemma}.

\begin{Corollary} \label[corollary]{maincorollary}
Let $G'$ be an inflation of a network $G$ and let $\{\rho_{\bm{U}} : \bm{U} \in \mathsf{ImagesInjectableSets}(G)\}$ be a family of states that is compatible with $G$. Let $T'  \subseteq \mathsf{ExpressibleSets}(G')$. If $I_{T'}$ is a causal compatibility inequality for $G'$ operating on families of states $\{ \tau_{\bm{Y}'} : \bm{Y}' \in T'\}$, then the family of states $\{ \sigma_{\bm{Y}'} : \bm{Y}' \in T' \}$, defined via 
\begin{align}
\sigma_{\bm{Y}'} := \bigotimes_{\bm{U} \in \mathsf{ImagesInjectableComponents}(\bm{Y}')}\rho_{\bm{U}_{(j)}}, \label{eq:sigma_Y_imagesinjcomp}
\end{align}
satisfies $I_{T'}$.
\end{Corollary}

\begin{proof} By Proposition~\ref{mainlemma}, if $\{\rho_{\bm{U}} : \bm{U} \in \mathsf{ImagesInjectableSets}(G)\}$ is compatible with $G$, then $\{ \sigma_{\bm{Y}'} : \bm{Y}' \in \mathsf{ExpressibleSets}(G') \}$, with 
\begin{align}
\sigma_{\bm{Y}'} := \bigotimes_{\bm{U} \in \mathsf{ImagesInjectableComponents}(\bm{Y}')}\rho_{\bm{U}_{(j)}}, 
\end{align}
is compatible with $G'$. Since the compatibility of $\{ \sigma_{\bm{Y}'} : \bm{Y}' \in \mathsf{ExpressibleSets}(G') \}$ with $G'$ implies that $\{ \sigma_{\bm{Y}'} : \bm{Y}' \in T'\}$ is compatible with $G'$ for each $T' \subseteq \mathsf{ExpressibleSets}(G')$, and since $I_{T'}$ is satisfied for all families of states $\{ \tau_{\bm{Y}'} : \bm{Y}' \in T'\}$ that are compatible with $G'$, it must also be satisfied for $\{ \sigma_{\bm{Y}'} : \bm{Y}' \in T' \}$.
\end{proof}
The utility of this corollary for questions of compatibility with the original causal model $G$ arises from the contrapositive of the above statement: if $I_{T'}$ is {\em not} satisfied for $\{ \sigma_{\bm{Y}'} : \bm{Y}' \in T' \}$ with $\sigma_{\bm{Y}'}$ as in Eq.~(\ref{eq:sigma_Y_imagesinjcomp}), then $\{\rho_{\bm{U}} : \bm{U} \in \mathsf{ImagesInjectableSets}(G)\}$ cannot be compatible with $G$. It is in this direction that we most use the above results in the remainder of this work.

It should be noted that all of these results are simply straightforward quantum generalizations of the ones presented in Ref.~\cite{WolfeSpekkensFritz_2019} for the case of classical causal models.

\subsubsection{The distinction between fanout and nonfanout inflations}

An inflation is called \textit{fanout} if it contains a latent node that influences visible nodes that are the same up to copy-indices, that is, if it contains a latent node that  influences more than one copy of a visible node from the original DAG. Otherwise, the inflation is termed \textit{nonfanout}.\footnote{Equivalently, one can define an inflation as nonfanout if the {\em descendent subgraph} of every latent node in the inflation DAG is the same up to copy-indices as  the  {\em descendent subgraph} of the corresponding latent node in the original DAG.}
For the triangle scenario, an example of a fanout inflation is the spiral inflation, depicted in~\Cref{fig:triangleandinflations}(c), while an example of a nonfanout inflation is the Cut inflation, depicted in~\Cref{fig:triangleandinflations}(b).

If one is considering the causal compatibility problem in quantum causal models, then it is not clear how to make use of a fanout inflation. In particular, the interpretation of a fanned-out latent node cannot be that it is copied and made to influence each of several children that are the same up to copy-indices, because the quantum no-broadcasting theorem prohibits such copying~\cite{Barnum_96}. As such, we here restrict attention to nonfanout inflations when deriving causal compatibility constraints in the case of quantum causal models.
 
In the case where the visible nodes are classical, so that the compatibility problem concerns a joint probability distribution over the visible nodes, then it is an interesting problem to determine, for a given causal structure, whether there is a difference between what can be realized with classical versus quantum latent nodes. Nonfanout inflations derive constraints on compatible distributions that hold for both types of latent nodes.  It is only with fanout inflations that one can derive constraints that can separate the cases of classical and quantum latents.  For the Bell scenario, for instance, it is a fanout inflation that allows one to derive Bell inequalities, which are the constraints that establish the existence of a quantum-classical gap in this case.
 
It is worth noting that some causal structures do not admit of any nontrivial nonfanout inflations. For example, the causal structure known as the bilocality scenario (see e.g.,~\cite{ligthart2023inflation}) is such an example that is moreover a member of the class of causal structures considered here (i.e., it is a network).

\subsection{Example of Cut inflation of triangle scenario: classical case}

As noted above, one obtains a description of how inflation can be used to derive causal compatibility inequalities for {\em classical} causal models by simply restricting to density operators that are diagonal relative to some fixed product basis over the subsystems and channels that act as stochastic maps relative to this basis. 

In this article, we will primarily focus on using the inflation technique to derive causal compatibility constraints for the triangle scenario using the cut inflation.  We do so first for classical causal models, i.e., using the classical version of the inflation technique, in order to facilitate a comparison with the case of quantum causal models which is our main focus.

\subsubsection{Leveraging classical marginal inequalities}
\label{sec:classicalcompatibility}
 
We consider the $AB$-Cut inflation of the triangle scenario, depicted in~\Cref{fig:triangleandinflations}(b). As noted earlier, the injectable sets of visible nodes in this case are, in addition to all the singleton sets, the two-node sets $\{ A_1, C_1\}$ and $\{ B_1, C_1\}$. Meanwhile, the two-node set $\{A_{1}, B_{1}\}$ is not injectable, but it is expressible. By the classical version of~\Cref{maincorollary}, which follows from taking all density operators to be diagonal in some product basis (see corollary 6 of Ref.~\cite{WolfeSpekkensFritz_2019}), any causal compatibility inequality for the inflation DAG that involves only the expressible sets $\{ A_1\}, \{ B_1\}, \{ C_1\}$, $\{ A_1, B_1\}$, $ \{ A_1, C_1\}$ and $\{B_1, C_1 \}$ defines (by the dropping of copy-indices) a causal compatibility inequality for the original DAG. By considering these sets as marginal contexts, we turn  to the question of how to leverage results from the classical marginal problem to derive causal compatibility inequalities on the inflation DAG that are of this type.

We begin by recalling the classical marginal problem on a set of variables ${\bf V}$.  Any subset of variables is termed a {\em marginal context}. The set of {\em all} marginal contexts, which is simply the power set of the set ${\bf V}$, we denote by $2^{\bf V}$.  An instance of the classical marginal problem is to determine, for a given set of marginal contexts $S \subseteq 2^{\bf V}$, and a given set of distributions on each of these contexts, $\{ Q_{\bf X}: {\bf X} \in S\}$, whether there is a joint distribution $Q_{\bf V}$ such that each of the $Q_{\bf X}$ is recovered as the marginal  of $Q_{\bf V}$ on ${\bf X}$, i.e., such that $Q_{\bf X}=\sum_{{\bf V}/{\bf X}} Q_{\bf V}$.   In this case, we say that the set of distributions $\{ Q_{\bf X}: {\bf X} \in S\}$ satisfies {\em marginal compatibility}. We are here introducing an important notational convention.  When we denote distributions by `$Q$', such as $\{ Q_{\bf X} :{\bf X} \in S\}$, we mean a set of distributions on marginal contexts that may or may not be marginals of a single distribution, whereas if we denote distributions by `$P$', such as $\{ P_{\bf X} :{\bf X} \in S\}$, then they {\em are} assumed to be marginals of a single distribution.

The most obvious constraint that a set $\{ Q_{\bf X}: {\bf X} \in S\}$ must satisfy in order to be marginally compatible is termed the {\em equimarginal property} and asserts that for any inclusion relation holding among the marginal contexts, i.e., ${\bf X}' ,{\bf X} \in S$ such that ${\bf X}' \subset {\bf X}$, we require $Q_{\bf X'}:=\sum_{{\bf X}/{\bf X'}} Q_{\bf X}$.  This constraint is straightforward to check, so we will restrict our attention to {\em nontrivial} constraints.    That is, we will only ask about the marginal compatibility of sets $\{ Q_{\bf X}: {\bf X} \in S\}$ that satisfy the equimarginal property.

For our example, the set of variables is ${\bf V}= \{A,B,C\}$, and the set of marginal contexts of interest is the powerset of ${\bf V}$ excluding ${\bf V}$ itself, so that $S= 2^{\bf V}\!\setminus\!{\bf V} :=\{ \{A\}, \{B\},\{C\},\{AB\}, \{AC\},\{BC\}\}$.   The input to our classical marginal problem in this case is an equimarginal set of distributions $\{ Q_{\bf X}:{\bf X}\in  2^{\bf V}\!\setminus\!{\bf V}   \} = \{Q_A, Q_B, Q_C, Q_{AB}, Q_{AC}, Q_{BC} \}$. 

A necessary condition for such an equimarginal family of distributions to be marginally compatible is  that the following inequality be satisfied~\cite{butterley2006compatibility}:
\begin{align}
\textbf{1} - Q_A -Q_B-Q_C + Q_{AB}+ Q_{AC}+ Q_{BC} \ge \textbf{0},\label{BellWignerInequality}
\end{align}
where the distributions appearing here are conceptualized as functions over the full sample space, $\bf{1}$ denotes the function that takes the value 1 everywhere, and $\bf{0}$ denotes the function that takes the value 0 everywhere.  We refer to this as a {\em  marginal compatibility inequality}. 

To emphasize that this is to be conceptualized as an inequality on a family of distributions $\{ Q_{\bf X}:{\bf X}\in 2^{\bf V}\!\setminus\!{\bf V}\}$,
It is useful to write it as:
\begin{align}
 \Delta(\{ Q_{\bf X}:{\bf X}\in 2^{\bf V}\!\setminus\!{\bf V} \}) \ge \textbf{0}.
 \end{align}
where $  \Delta(\{ Q_{\bf X}:{\bf X}\in 2^{\bf V}\!\setminus\!{\bf V} \}) := \textbf{1} - Q_A -Q_B-Q_C + Q_{AB}+ Q_{AC}+ Q_{BC} $.  

Also, if one evaluates this functional inequality at  $A=a, B=b, C=c$ for each set of values $a,b,c$, one obtains a {\em set} of inequalities on the real parameters appearing in these distributions, specifically, 
\begin{align}\forall a,b,c: &1- Q_A(a) -Q_B(b)-Q_C(c) \nonumber\\
&+ Q_{AB}(ab)+ Q_{AC}(ac)+Q_{BC}(bc) \ge 0.
\end{align} 
If the constraint is presented in this form, then we have a set of inequalities, rather than a single inequality.  In this case, we term them {\em marginal compatibility inequalities} (plural).

Now consider a distribution $P_{A_1 B_1 C_1}$ on the variables appearing in the $AB$-Cut inflation DAG, namely, $A_1$, $B_1$ and $C_1$. The marginal compatibility inequality above implies that 
\begin{align}
&\textbf{1} - P_{A_1} -P_{B_1}-P_{C_1} + P_{A_1 B_1}+ P_{A_1 C_1}+ P_{B_1 C_1}\ge \textbf{0}. \label{BellWignerInequalityprime}
\end{align} 
Note that this constraint follows simply from classical probability theory and has nothing to do with the causal structure of the $AB$-Cut inflation DAG.  However, we will now combine it with a constraint that {\em does} follow from this causal structure.

Specifically, we consider the $d$-separation relations among the visible nodes. In brief, the $d$-separation relations in a given DAG $G$ are the graphical description of the conditional independence of variables (or sets thereof) associated to the nodes of $G$ for any distribution compatible with $G$. These descriptions are given in terms of paths between the nodes in question (see, e.g., Refs.~\cite{peters2017elements,Pearl_2009} for a formal definition of $d$-separation). In the case of the $AB$-Cut inflation, there is one such relation, namely, that $A_1$ and $B_1$ are $d$-separated given the empty set. This holds because the only path between $A_{1}$ and $B_{1}$ in the $AB$-Cut inflation passes through $C_{1}$, which is a collider---see Refs.~\cite{peters2017elements,Pearl_2009}. As a consequence, the variables $A_{1}$ and $B_{1}$ are marginally independent (which is equivalent to independence conditioned on the empty set) in any distribution compatible with the $AB$-Cut inflation. That is, for any such distribution $P_{A_1 B_1 C_1}$, we require 
\begin{align}\label{eq:dsepABcut}
P_{A_1 B_1}=P_{A_1}P_{B_1}.
\end{align}

To get a nontrivial causal compatibility inequality on the marginals of $P_{A_1 B_1 C_1}$, it suffices to combine the equality constraint of Eq.~\eqref{eq:dsepABcut} with the marginal compatibility inequality of Eq.~\eqref{BellWignerInequality} to obtain:
\begin{align}
\textbf{1} - P_{A_1} -P_{B_1}-P_{C_1} + P_{A_1}P_{B_1}+ P_{A_1 C_1}+ P_{B_1 C_1} \ge \textbf{0}.\label{BellWignerInequalityprimeprime}
\end{align}
(Note the difference to Eq.~\eqref{BellWignerInequalityprime}, namely, that $P_{A_1}P_{B_1}$ appears in place of $P_{A_1 B_1}$.)

From this, we can infer a causal compatibility inequality for the triangle scenario DAG.  It suffices to note that the inequality of Eq.~\eqref{BellWignerInequalityprimeprime} is a polynomial function of $P_{A_1} , P_{B_1}, P_{C_1}, P_{A_1 C_1}$ and $P_{B_1 C_1}$ and that these are all marginals on {\em injectable sets} of visible nodes in the inflation DAG.  By~\Cref{maincorollary}, we infer that the inequality having the same functional form as Eq.~\eqref{BellWignerInequalityprimeprime} but where $A_1$, $B_1$ and $C_1$ are replaced by $A$, $B$ and $C$ respectively, namely, 
\begin{align}
\textbf{1} - P_A - P_B - P_C + P_{A}P_{B} + P_{AC} + P_{BC} \ge \textbf{0},\label{ABCutBellWignerInequality}
\end{align}
is a causal compatibility inequality for the triangle scenario in terms of the marginals $P_A, P_B, P_C, P_{AC},$ and $P_{BC}$ of $P_{ABC}$. If a given joint distribution $P_{ABC}$ violates the inequality of Eq.~\eqref{ABCutBellWignerInequality}, then it is incompatible with the triangle scenario. Following the terminology of a ``$G$-compatibility inequality'' introduced above, it is appropriate to refer to such an  inequality on $P_{ABC}$ as a {\em triangle-compatibility inequality}.

One can of course consider Cut Inflations where the Cut is between $B$ and $C$ or between $A$ and $C$, rather than being between $A$ and $B$.  We get analogous inequalities for each case, namely, for each   $x,y,z \in \{A,B,C\}$ such that $x\ne y\ne z$,
\begin{align}
\textbf{1} - P_x  - P_y - P_z + P_{x}P_{y}+ P_{xz}+ P_{yz} \ge \textbf{0}.
\label{xyCutBellWignerInequality}
\end{align}
A violation of any of these inequalities suffices to demonstrate triangle-incompatibility. In general, it will be clear from the context and the notation which inequality is being used (such as via the use of subscripts as in, e.g., \eqref{TriangleCompatibilityInequality2} below).

Because the inequality of Eq.~\eqref{xyCutBellWignerInequality} expresses a necessary condition on the joint distribution $P_{ABC}$ for triangle-compatibility using the $xy$-Cut inflation, it is useful to write it as 
\begin{align}
I_{xy}(P_{ABC}) \ge \bf{0},\label{TriangleCompatibilityInequality2}
\end{align}
where
\begin{align*}
&I_{xy}(P_{ABC}) := \textbf{1} - P_x - P_y - P_z + P_{x}P_{y} + P_{xz} + P_{yz},
\end{align*}
and where $P_x$ is the marginal on the singleton set $x$, $P_{xz}$ is the marginal on the two-variable set $xz$, etcetera. 

\subsubsection{Some distributions whose triangle-incompatibility is witnessed}

We now explore some of the conclusions that can be inferred from this  triangle-compatibility inequality.

As a simple example, consider the following tripartite distribution, which we will refer to as the {\em GHZ distribution} (because of its similarity with the Greenberger-Horne-Zeilinger (GHZ) state in quantum theory \cite{greenberger1989going}):
\begin{align}\label{eq:GHZdistn}
P_{ABC}^{(\rm GHZ)}
 = \frac{1}{2} \left( [000]_{ABC} +  [111]_{ABC} \right),
\end{align}
where $[abc]_{ABC}$ denotes the point distribution wherein all the probabilistic weight is on $A=a$, $B=b$, and $C=c$.\footnote{Formally, $[abc]_{ABC}$ can be understood as a vector in the $(|A|\times |B| \times |C| - 1)$-dimensional probability simplex (in $\mathbb{R}^{|A| \times |B| \times |C|}$) containing all distributions $P_{ABC}$, where $|X|$ denotes the (finite) number of values that the variable $X$ can take. By choosing a labelling $f: A \times B \times C \rightarrow \{1, \dots, |A|\times|B|\times|C|\}$, the vector $[abc]_{ABC}$ corresponds to a unit vector along the axis in $\mathbb{R}^{|A|\times |B| \times |C|}$ identified by $f(a,b,c)$ (i.e., it corresponds to a vertex of the probability simplex).}
It is straightforward to see that $P^{\rm GHZ}_{ABC}$ violates a triangle-compatibility inequality. For example, using Eq.~(\ref{ABCutBellWignerInequality}) and taking $A=0, B=0, C=1$, we find that the left-hand-side evaluates to $1 - \tfrac{1}{2} -  \tfrac{1}{2}-\tfrac{1}{2}+ \tfrac{1}{4} + 0 +0 =-\tfrac{1}{4}$.  Since this is negative, the inequality is violated. It follows that the GHZ distribution is  incompatible with the triangle scenario.
 
It is worth recalling that violations of a $G$-compatibility inequality by a distribution witnesses the incompatibility of that distribution with the causal structure $G$, but satisfaction of such inequalities by a distribution does not guarantee that it is compatible. A good example of this is 
the following distribution, 
\begin{align}\label{eq:Wdistn}
P_{ABC}^{({\rm W})} = \frac{1}{3} \left( [100]_{ABC} +  [010]_{ABC} +  [001]_{ABC} \right).
\end{align}
which is termed the ``W distribution'' (this distribution is likewise inspired by a quantum state, namely, the $W$-state of Ref.~\cite{Durr_00}). 
$P^{\rm W}_{ABC}$ can easily be verified to {\em satisfy} the  triangle-compatibility inequalities (satisfaction of the inequality for any of the Cuts implies satisfaction of the inequalities for the other Cuts due to the symmetries in the distribution), but it is nonetheless incompatible with the triangle scenario. Triangle-incompatibility of the W distribution is established in Ref.~\cite{WolfeSpekkensFritz_2019} using the Spiral inflation of the triangle scenario, which is depicted in Fig.~\ref{fig:triangleandinflations}(c).\footnote{\label{footnote:ringinflation} The triangle-incompatibility of the W distribution can also be established using a nonfanout inflation, known as the Ring inflation, thereby establishing the incompatibility of the W distribution with the triangle scenario even when the latent nodes are quantum, a fact that will be of use to us further on. (For a description of the Ring inflation see Refs.~\cite{gisin2020constraints, navascues2020genuine}.)
}

\subsection{Example of Cut inflation of triangle scenario: quantum case}
\label{subsec:fully_q_inflation}

Again, we consider  the $AB$-Cut inflation of the triangle scenario, depicted in~\Cref{fig:triangleandinflations}(b). Analogously to the classical case,~\Cref{maincorollary} guarantees that if we can identify a causal compatibility inequality on the inflation DAG that involves only the 
 expressible sets $\{ A_1\}, \{ B_1\}, \{ C_1\}, \{A_{1}, B_{1}\}, \{ A_1, C_1\}$ and $ \{ B_1, C_1\}$, then dropping the copy-indices yields a causal compatibility inequality for the triangle scenario.  We find such a causal compatibility inequality on the inflation DAG by leveraging the results from the quantum marginal problem.  This is the sense in which the fully quantum inflation technique allows questions of causal compatibility to be related to quantum marginal problems, allowing progress on the latter to be brought to bear on the former.

\subsubsection{Leveraging quantum marginal inequalities}\label{subsec:q_marginal_ineqs}

This subsection presents a family of operator inequalities describing solutions to a particular instance of the quantum marginal problem and demonstrates how these can be leveraged in fully quantum inflation. Specifically, the extra constraints enforced by the causal structure of the inflation DAG (such as factorization constraints) can be substituted into these operator inequalities to yield modified operator inequalities that constitute causal compatibility inequalities for the inflation DAG. 

The quantum version of the marginal problem is defined precisely as the classical version, but where we consider joint quantum states rather than joint probability distributions, and marginalization corresponds to the partial trace operation.
Consequently, if ${\bf V}$ denotes the full set of quantum systems under consideration, an instance of the quantum marginal problem is to determine, for a given set of marginal contexts $S \subseteq 2^{\bf V}$, and a given set of quantum states on each of these contexts, $\{ \sigma_{\bf X}: {\bf X} \in S\}$, whether there is a joint quantum state $\sigma_{\bf V}$ such that each of the $\sigma_{\bf X}$ is recovered as the marginal  of $\sigma_{\bf V}$ on ${\bf X}$, i.e., such that $\sigma_{\bf X}= {\rm Tr}_{{\bf V} \setminus {\bf X}} [\sigma_{\bf V}]$.   In this case, we say that the set of states $\{ \sigma_{\bf X}: {\bf X} \in S\}$ satisfy {\em marginal compatibility}. The notations $\sigma$ and $\rho$ will parallel those of $Q$ and $P$ in the classical case: $\{ \sigma_{\bf X} :{\bf X} \in S\}$ denotes  a set of states on marginal contexts that may or may not be marginals of a single joint state, whereas if we write $\{ \rho_{\bf X} : {\bf X} \in S\}$, then they are assumed to be marginals of a single joint state.

In Ref.~\cite{william_hall_05}, Hall introduced a necessary condition  that the marginals of any quantum state on a (finite-dimensional) composite Hilbert space consisting of an odd number of subsystems must satisfy (c.f.~\cite[Thm 1]{william_hall_05} and the discussion thereafter). 
Suppose ${\bf V}$ denotes a set of quantum systems and let $\sigma_{\bf X} \in \mathcal{L}(\bigotimes_{X\in {\bf X} }\H_{X})$
 denote a joint quantum state on  ${\bf X}\subseteq {\bf V}$ (i.e., a positive semidefinite operator with unit trace). 
The necessary condition on a set $\{\sigma_{\bf X}: {\bf X}\in 2^{\bf V}\!\setminus\!{\bf V}  \}$ of states to be the marginals of a single state $\sigma_{\bf V}$ when the set ${\bf V}$ is of odd cardinality  is the following operator inequality:
\begin{align}
 \Delta(\{\sigma_{\bf X}: {\bf X}\in 2^{\bf V}\!\setminus\!{\bf V}  \}) \ge 0 \label{eq:Hall_op_ineq}
\end{align}
where 
\begin{align}
 \Delta(\{\sigma_{\bf X}: {\bf X}\in 2^{\bf V}\!\setminus\!{\bf V}  \}) := 
 \sum_{{\bf X}\in 2^{\bf V}\!\setminus\!{\bf V}} (-1)^{|{\bf X}|}\sigma_{\bf X} \otimes \mathbb{1}_{{\bf V}\!\setminus\! {\bf X}}, \label{eq:Hall_ineqs}
\end{align}
with $|{\bf X}|$ denoting the cardinality of the set ${\bf X}$, with $\mathbb{1}_{{\bf V}\!\setminus\! {\bf X}}$ denoting the identity operator on the subsystems in ${\bf V}\!\setminus\! {\bf X}$, and with $\ge$ denoting the standard partial order relation over Hermitian operators (i.e., $M \ge N$ if and only if $M-N$ is a positive semi-definite operator), so that an operator inequality of the form $O \ge 0$ signifies that $O$ is positive semidefinite.\footnote{For pedagogical purposes, we here explain why there is a difference between the cases of even and odd numbers of systems. 
 In Ref.~\cite{william_hall_05}, it is shown that the operator
 \begin{align}
    \Lambda(\{\sigma_{\bf X}: {\bf X}\in 2^{\bf V}  \}) := \sum_{{\bf X}\in 2^{\bf V}} (-1)^{|{\bf X}|}\sigma_{\bf X} \otimes \mathbb{1}_{{\bf V}\!\setminus\! {\bf X}}
 \end{align}
 is positive for any cardinality of ${\bf V}$.  However, this is not the same as  \eqref{eq:Hall_ineqs} being positive for any cardinality of ${\bf V}$, since the sum in the above expression includes the case where  ${\bf X} = {\bf V}$
  while the sum in \eqref{eq:Hall_ineqs} doesn't. For odd cardinality of ${\bf V}$, we have that
 \begin{align}
    \Lambda(\{\sigma_{\bf X}: {\bf X}\in 2^{\bf V}  \}) = \Delta(\{\sigma_{\bf X}: {\bf X}\in 2^{\bf V}\!\setminus\!{\bf V}  \}) - \sigma_{{\bf V}}
 \end{align}
while for even cardinality of ${\bf V}$, we have
 \begin{align}
    \Lambda(\{\sigma_{\bf X}: {\bf X}\in 2^{\bf V}  \}) = \Delta(\{\sigma_{\bf X}: {\bf X}\in 2^{\bf V}\!\setminus\!{\bf V}  \}) + \sigma_{{\bf V}}
 \end{align}
 In the odd cardinality case, $\Delta(\{\sigma_{\bf X}: {\bf X}\in 2^{\bf V}\!\setminus\!{\bf V}  \})$ is a sum of two positive operators and so we can conclude that it is positive, while in the even cardinality case, we cannot make this inference.  
}

 In this work, we focus on the case of ${\bf V} = \{ A,B,C\}$,  and thus on the operator
\begin{align}
 \Delta&(\{\sigma_{\bf X}: {\bf X}\in 2^{\bf V}\!\setminus\!{\bf V}  \})\nonumber\\
 &  = \mathbb{1}_{ABC} - \sigma_{A}\otimes \mathbb{1}_{BC} - \sigma_{B}\otimes \mathbb{1}_{AC} - \sigma_{C}\otimes \mathbb{1}_{AB} \nonumber\\
& \quad   + \sigma_{AB}\otimes \mathbb{1}_{C} + \sigma_{AC}\otimes \mathbb{1}_{B} + \sigma_{BC}\otimes \mathbb{1}_{A}. \label{eq:Hall_tripartitewithidentities}
\end{align}
It is convenient to suppress tensor products with identity operators, as well as the subscripts on the remaining identity operator on the entire space, and write simply 
\begin{align}
 \Delta&(\{\sigma_{\bf X}: {\bf X}\in 2^{\bf V}\!\setminus\!{\bf V}  \})\nonumber\\
 &  = \mathbb{1} - \sigma_{A} - \sigma_{B} - \sigma_{C} 
 + \sigma_{AB} + \sigma_{AC} + \sigma_{BC}. \label{eq:Hall_tripartite}
\end{align}
Because ${\bf V}$ is of odd cardinality, we infer from Hall's result that $\Delta(\{\sigma_{\bf X}: {\bf X}\in 2^{\bf V}\!\setminus\!{\bf V}  \})$ must be positive semidefinite when the elements of $\{\sigma_{\bf X}: {\bf X}\in 2^{\bf V}\!\setminus\!{\bf V}  \}$ are marginals of a single state $\sigma_{{ \bf V}}$, so that 
\begin{align}
 \mathbb{1} - \sigma_{A} - \sigma_{B} - \sigma_{C} 
 + \sigma_{AB} + \sigma_{AC} + \sigma_{BC}\ge 0.
\label{quantumBellWigner}
\end{align}
In this form, the operator inequality  is clearly the quantum counterpart of Eq.~\eqref{BellWignerInequality}.  We refer to it as a {\em quantum marginal compatibility inequality.} 

Now consider a state $\rho_{A_1 B_1 C_1}$ on the systems that appear  in the $AB$-Cut inflation DAG, namely, $A_1$, $B_1$ and $C_1$.  The quantum marginal compatibility inequality above implies that 
\begin{align}
&\mathbb{1}
 - \rho_{A_1} -\rho_{B_1}-\rho_{C_1} + \rho_{A_1 B_1}+ \rho_{A_1 C_1}+ \rho_{B_1 C_1}\ge 0.
\label{quantumBellWignerInequality}
 \end{align} 
As with the constraint of Eq.~\eqref{BellWignerInequalityprime} in the classical case, this constraint has nothing to do with the causal structure of the $AB$-Cut inflation DAG, but we can combine it with such a constraint to obtain a causal compatibility inequality.

Again, we leverage the fact that, in the $AB$-Cut inflation DAG, $A_1$ and $B_1$ are d-separated given the empty set. It follows that any state $\rho_{A_1 B_1 C_1}$ that is compatible with this DAG must satisfy
\begin{align}\label{eq:quantumdsepABcut}
\rho_{A_1 B_1}=\rho_{A_1} \otimes \rho_{B_1}.
\end{align}

To get a nontrivial causal compatibility inequality for the $AB$-Cut inflation DAG, we combine the equality constraint of Eq.~\eqref{eq:quantumdsepABcut} with the quantum marginal compatibility inequality of Eq.~\eqref{quantumBellWignerInequality} to obtain:
 \begin{align}
\mathbb{1} - \rho_{A_1} -\rho_{B_1}-\rho_{C_1} + \rho_{A_1}\otimes \rho_{B_1}+ \rho_{A_1 C_1}+ \rho_{B_1 C_1} \ge 0.\label{quantumBellWignerInequalityprimeprime}
\end{align}

Finally, we can use this together with~\Cref{maincorollary} to obtain a causal compatibility inequality for the original triangle scenario DAG.  Because the inequality is a polynomial function of $\rho_{A_1} , \rho_{B_1}, \rho_{C_1}, \rho_{A_1 C_1}$ and $\rho_{B_1 C_1}$ and because these are all marginals on {\em injectable sets} of visible nodes in the inflation DAG, the inequality having the same functional form as Eq.~\eqref{BellWignerInequalityprimeprime} but where $A_1$, $B_1$ and $C_1$ are replaced by $A$, $B$ and $C$ respectively, 
 \begin{align}
\mathbb{1} - \rho_{A} -\rho_{B}-\rho_{C} + \rho_{A}\otimes \rho_{B}+ \rho_{A C}+ \rho_{B C} \ge 0\label{quantumABCutBellWignerInequality}
\end{align}
is a causal compatibility inequality for the triangle scenario. 

We refer to this operator inequality as a {\em triangle-compatibility inequality} for the quantum state $\rho_{ABC}$.   It is convenient to write this as
\begin{align}
I_{AB}(\rho_{ABC}) \ge 0,
\end{align}
where
 \begin{align}\label{eq:witness}
I_{AB}(\rho_{ABC})  := \mathbb{1} - \rho_{A} - \rho_{B} - \rho_{C} + \rho_{A} \otimes \rho_{B} + \rho_{AC} + \rho_{BC}.
\end{align}

Summarizing, if the state $\rho_{ABC}$ is triangle-compatible, then $I_{AB}(\rho_{ABC})$
 is a positive operator. Conversely, if $I_{AB}(\rho_{ABC})$ is found to be non-positive, then $\rho_{ABC}$ is triangle-incompatible. Note that $I_{AB}(\rho_{ABC})$
can be a positive operator even though $\rho_{ABC}$ is triangle-incompatible, which is why the nonpositivity of  $I_{AB}(\rho_{ABC})$ is only a necessary and not a sufficient condition for triangle-incompatibility.   
We refer to the operator $I_{AB}(\rho_{ABC})$ as a {\em triangle-incompatibility witness}.

As we will also be considering Cut inflations between the visible nodes $A$ and $C$ and between $B$ and $C$, 
it is convenient to define, for a given $\rho_{ABC}$,
\begin{align}
I_{xy}(\rho_{ABC}) = \mathbb{1} - \rho_x  - \rho_y - \rho_z + \rho_{x}\otimes \rho_{y}+ \rho_{xz}+ \rho_{yz} \label{eq:q_I_xy}
\end{align}
where  $x,y,z \in \{A,B,C\}$ such that $x\ne y\ne z$ and where we are suppressing tensor products with identity operators. For each type of cut, we have a corresponding triangle-incompatibility  inequality, $I_{xy}(\rho_{ABC})\ge0$.

It is worth noting that $I_{xy}(\rho_{ABC})$ can be equivalently expressed as
\begin{align}
I_{xy}(\rho_{ABC})   = \Delta(\{\rho_{\bf X}: {\bf X}\in 2^{\bf V}\!\setminus\!{\bf V}  \}) + \rho_{x} \otimes \rho_{y} - \rho_{xy}.
\end{align}

\subsubsection{Some technical results}

Before proceeding to our general results, we note a few useful facts about the  operators $I_{xy}(\rho_{ABC})$. Since $\Delta(\{\rho_{\bf X}: {\bf X}\in 2^{\bf V}\!\setminus\!{\bf V}  \})$ is guaranteed to be positive, the only chance for demonstrating that $I_{xy}(\rho_{ABC})$ is non-positive must derive from the non-positivity of $\rho_{x} \otimes \rho_{y} - \rho_{xy}$. We note the following:
\begin{Lemma} \label{lem:xy_minus_x_y} 
If $\rho_{xy} \neq \rho_{x} \otimes \rho_{y}$ then $\rho_{x} \otimes \rho_{y} - \rho_{xy} \ngeq  0$
\end{Lemma}

The proof is given in~\Cref{app:proof_lem_xy_minus_x_y}. Considered as an operator on $\H_{x} \otimes \H_{y}$, let us write
\begin{align}
\rho_{x} \otimes \rho_{y} - \rho_{xy} = \nu_{xy}^{+} - \nu_{xy}^{-}\label{nus}
\end{align}
where the $\{ \nu_{xy}^{\pm} \}_{xy\in \{AB,AC,BC\}}$ are positive semidefinite and where $\nu_{xy}^+$ is mutually orthogonal to $\nu_{xy}^-$ relative to the Hilbert-Schmidt inner product. The above lemma guarantees that $\nu_{xy}^{-}$ is non-zero and also positive semi-definite whenever $\rho_{xy} \neq \rho_{x} \otimes \rho_{y}$. Clearly, if $\rho_{ABC}$ is such that there exists some $\ket{\Phi} \in \supp(\nu_{xy}^{-} \otimes \mathbb{1}_{z}) \cap \ker(\Delta(\{\rho_{\bf X}: {\bf X}\in 2^{\bf V}\!\setminus\!{\bf V}  \}))$ for some  $x,y,z \in \{A,B,C\}$ such that $x \neq y \neq z$ then
\begin{align}
\braket{\Phi| I_{xy}(\rho_{ABC})|\Phi} < 0,
\end{align}
demonstrating that $\rho_{ABC}$ is triangle-incompatible. Despite the fact that the existence of such a $\ket{\Phi}$ is not \textit{a priori} necessary for demonstrating the triangle-incompatibility of $\rho_{ABC}$, it is sufficient, so this constitutes  a fruitful method for deriving such results, as we will see further on. 
  We state it as a proposition for later reference:
\begin{Proposition} \label{prop:supp_ker_intersect}
Let $\rho_{ABC}$ be a quantum state on the set of visible nodes ${\bf V}=\{A,B,C\}$, and let $\Delta(\{\rho_{\bf X}: {\bf X}\in 2^{\bf V}\!\setminus\!{\bf V}  \})$ and $\nu_{xy}^{-}$ be defined in terms of it via Eqs.~\eqref{eq:Hall_tripartite} and \eqref{nus}.
 If 
\begin{align}
\supp(\nu_{xy}^{-} \otimes \mathbb{1}_{z}) \cap \ker(\Delta(\{\rho_{\bf X}: {\bf X}\in 2^{\bf V}\!\setminus\!{\bf V}  \})) \neq \emptyset \label{eq:supp_ker_intersect}
\end{align}
for some $x,y,z \in \{A,B,C\}$ such that $x\neq y \neq z$,
 then $\rho_{ABC}$ is triangle-incompatible. 
\end{Proposition}

One further advantage of evaluating compatibility with the triangle scenario by evaluating whether the operator $I_{xy}(\rho_{ABC})$ is  positive semi-definite consists in the finer-grained understanding of \textit{how} incompatibility of a given state can be shown in certain cases. If there exists a pure product state $\ket{\Phi} = \ket{\phi}_{A}\ket{\phi'}_{B}\ket{\phi''}_{C}$ such that $\braket{\Phi| I_{xy}(\rho_{ABC}) |\Phi} < 0$, then it is possible to get a verdict of triangle-incompatibility for $\rho_{ABC}$ purely by local measurement statistics, i.e., by measuring locally on $A$ according to the basis determined by $\ket{\phi}_{A}$, and similarly for $B$ and $C$.
 We will see shortly why the question of whether one can witness incompatibility using the distribution obtained from local measurements is significant.

\subsubsection{Some states whose triangle-incompatibility is witnessed}
 \label{subsec:state_vs_distr}\label{sec:witnessable}

 We begin by considering quantum states $\rho_{ABC}$ that simply encode a classical distribution.  Any quantum state that is diagonal in the computation basis, for instance, merely encodes a classical distribution over binary variables.  The quantum state encoding the GHZ distribution is:
\begin{align}\label{eq:incGHZstate}
\rho^{({\rm GHZdistn})} &= \frac{1}{2} \left( |000\rangle \! \langle 000|_{ABC} +|111\rangle \! \langle 111|_{ABC} \right)\nonumber\\
&:= \sum_{a,b,c} P^{({\rm GHZ})}(a,b,c)  |abc\rangle \! \langle abc|_{ABC},
\end{align}
while the one encoding the W distribution is
\begin{align}\label{eq:incWstate}
\rho^{({\rm Wdistn})} &:= \frac{1}{3} ( |100\rangle\!\langle 100|_{ABC} +|010\rangle\!\langle 010|_{ABC} \nonumber\\
&\quad \quad \quad +|001\rangle\!\langle 001|_{ABC} )\nonumber\\
&= \sum_{a,b,c} P^{({\rm W})}(a,b,c)  |abc\rangle\! \langle abc|_{ABC}.
\end{align}
Note that asking about the triangle-compatibility of a state that encodes a distribution in a quantum causal model is not equivalent to asking about the triangle-compatibility of the corresponding distribution in a classical causal model.  This is because in the first case, one allows the latent nodes to be quantum while in the second case they are constrained to be classical, and    it can happen that a given distribution is realizable in the first but not the second type of model.  Indeed, for the case of the triangle scenario, the existence of a gap in the sorts of distributions that are realizable using quantum latent nodes versus those that are realizable using classical latent nodes was shown in  Ref.~\cite{fritz2012beyond} and further studied in Refs.~\cite{fraser2018causal,polino2023experimental}.

To settle questions about  $G$-compatibility of a state that encodes a distribution, it suffices to use a version of the inflation technique that tests for the $G$-compatibility of a distribution on the visible nodes when the latent nodes are allowed to be quantum.  The standard inflation technique provides the means of doing so: as long as one considers a nonfanout inflation, the technique can witness $G$-incompatibility of a distribution even if the latent nodes are quantum~\cite{WolfeSpekkensFritz_2019}.  One can also consider latent nodes that are quantum using fanout inflations by making use of the quantum inflation technique introduced in Ref.~\cite{wolfe_q_inflation}.  The point is that fully quantum inflation is not required for states that encode a distribution.  

In the case of the GHZ distribution, its triangle-incompatibility was established in Ref.~\cite{WolfeSpekkensFritz_2019} using the standard inflation technique with the Cut inflation.  Given that the latter is nonfanout, its incompatibility holds even for quantum latent nodes. 

As noted previously, the triangle-incompatibility of the W distribution can also be established using a nonfanout inflation, namely, the Ring inflation (see the discussion in footnote~\ref{footnote:ringinflation}), so the triangle-incompatibility of the W distribution also holds even for quantum latent nodes.

 Whether fully quantum inflation might provide a {\em more efficient means} of witnessing triangle-incompatibility for such states  remains a question for future research.  In any case, we move on now to consider quantum states that are not merely encoding probability distributions.

Consider the GHZ state
\begin{align}\label{eq:GHZstate}
|{\rm GHZ}\rangle := \frac{1}{\sqrt{2}} \left( |000\rangle_{ABC} +|111\rangle_{ABC} \right). 
\end{align}
and the W state,
\begin{align}\label{eq:Wstate}
|{\rm W}\rangle := \frac{1}{\sqrt{3}} \left( |100\rangle_{ABC} +|010\rangle_{ABC} +|001\rangle_{ABC} \right),
\end{align}
which differ from the states in Eqs.~\eqref{eq:incGHZstate} and \eqref{eq:incWstate} by being coherent superpositions of the relevant states rather than incoherent mixtures. 
Both of these are found to still violate the triangle-compatibility inequality obtained from fully quantum inflation, i.e.,  Eq.~\eqref{quantumABCutBellWignerInequality}, and so both 
are triangle-incompatible.  This can be verified by direct computation (or inferred from Theorem~\ref{thm:main}).  

We again pause to ask whether fully quantum inflation was {\em necessary} to reach the conclusion of triangle-incompatibility for these states.  

In fact, fully quantum inflation is {\em not} needed to see that the GHZ state is triangle-incompatible. It suffices to note that if one implements local measurements in the computation basis on the GHZ state, one prepares the GHZ distribution, or, equivalently, the quantum state that encodes the GHZ distribution, Eq.~\eqref{eq:incGHZstate}.  Such local measurements are consistent with the causal structure of the triangle scenario, so the possibility of transforming the GHZ state into the state encoding the GHZ distribution in this way means that if the GHZ state is realizable in a quantum causal model within the triangle scenario then so is the state encoding the GHZ distribution.  The contrapositive of this statement is that if the state encoding the GHZ distribution is triangle-incompatible then so is the GHZ state. But we established earlier in this section that the state encoding the GHZ distribution is indeed triangle-incompatible.  Note, moreover, that the triangle-incompatibility of the state encoding the GHZ distribution can be established with the Cut inflation. 

In the case of the W state, fully quantum inflation is also not needed to establish triangle-incompatibility.  Again, we can implement local measurements in the computation basis to convert the W state into the state encoding the W distribution and then make use of the fact that the state encoding the W distribution is shown to be triangle-incompatible using the standard inflation technique with a nonfanout inflation, as noted above.  

Unlike the situation with the GHZ distribution, the conclusion of triangle-incompatibility for the W distribution cannot be obtained using the Cut inflation DAG; one must instead use a nonfanout inflation that lies higher in the hierarchy of inflations~\cite{Navascues_20}, termed the Ring inflation.  Nonetheless, the question remains of whether some {\em other} choice of local measurement on the W state might lead  to a distribution whose triangle-incompatibility is witnessable by the Cut inflation.\footnote{In general, for those states whose triangle-incompatibilty is distribution-witnessable by the Cut inflation,  we do not expect {\em all} choices of local measurements to yield a distribution that witnesses this triangle-incompatibility.  For instance, even though the triangle-incompatibility of the GHZ state is distribution-witnessable using the $AB$-Cut inflation by implementing a measurement of the computation basis on each qubit, it is not distribution-witnessable if one instead measures the  $|+\rangle, |-\rangle$ basis on each qubit. }  
This does indeed turn out to be the case! By locally measuring each subsystem of the W state in the Pauli-$X$ basis, a probability distribution is produced that can be demonstrated to be triangle-incompatible using the $AB$-Cut inflation (some details of this calculation are presented in~\Cref{app:distn_witn_W_state}).

Generalizing  beyond the triangle scenario example,  two questions arise regarding whether a conclusion of $G$-incompatibility for some quantum state  could have been reached without requiring the `big guns' of fully quantum inflation.  The first question is whether it could have been reached just using the standard inflation technique for some nonfanout inflation DAG~\cite{WolfeSpekkensFritz_2019} or using the quantum inflation technique for some fanout inflation DAG~\cite{wolfe_q_inflation}.   The second question is whether this same conclusion could have been reached 
using the standard inflation technique with the very same inflation DAG that was used to establish the conclusion using the fully quantum inflation technique.  More precisely, the second question is: if the $G$-incompatibility of a {\em state} can be witnessed by an inflation $G'$, is there always some choice of local measurements yielding a distribution whose $G$-incompatibility can also be witnessed by $G'$?  

When one can prove incompatibility of a state by leveraging the incompatibilty of a distribution obtained from this state by local measurements, we will say that the state incompatibility is {\em distribution-witnessable}.  
We will also say that the proof of incompatibility of a state has {\em piggybacked} on a proof of incompatibility of a distribution. 

Using this terminology, the first question is whether there are examples of $G$-incompatibility results for states that are not distribution-witnessable, and the second question is whether there are examples of $G$-incompatibility results for states that are not distribution-witnessable {\em using the same inflation DAG}. 

At issue here is whether a test of state incompatibility leveraging the quantum marginal problem, i.e., the fully quantum inflation technique, is  strictly more powerful than a test based on results about the classical marginal problem.

In this article, we will focus on the second question.   In Section~\ref{sec:redux}, we will show that it receives a negative answer:   there are instances of $G$-incompatibility results for states that are not distribution-witnessable using the same inflation DAG.

We do not settle the first question here.  Consequently, it might
be the case that any given $G$-incompatibility result for a state obtained using an inflation DAG $G'$ can always be inferred from the $G$-incompatibility of a distribution obtained from the state by some local measurements, by using a different inflation DAG $G''$
that is at a higher level than $G'$ in the inflation heirarchy of Ref.~\cite{Navascues_20}.  It is worth noting, however, that even if this is the case,  the utility of this sort of piggybacking technique  is limited.  First, one requires an algorithm for determining the correct choice of local measurements.     Second, the computational cost of a compatibility test increases significantly as one moves up the inflation hierarchy of Ref.~\cite{Navascues_20}.  Both facts imply that there is a computational advantage to testing state compatibility directly using fully quantum inflation.

\section{Qubit visible nodes} \label{sec:qubit_nodes}

In this section, we begin to showcase the utility of these new ways of witnessing the incompatibility of a quantum state with a given causal structure. We consider the triangle scenario where each of the visible nodes corresponds to a qubit  and we investigate the compatibility of both pure (\Cref{subsec:qubit_pure}) and mixed (\Cref{subsec:qubit_mixed}) three-qubit states. In~\Cref{sec:redux}, we consider the question of  which states have distribution-witnessable incompatibility for a given inflation DAG, and show that there are examples of incompatibility that are not distribution-witnessable.

\subsection{The quantum compatibility problem for pure states} \label{subsec:qubit_pure}

It is straightforward to see that any tripartite pure state that factorizes (i.e., is a product state) across some bipartition ($A|BC$, $B|AC$, or $C|AB$) is compatible with the triangle scenario. For example, suppose that $\rho_{ABC}$ is the state $\ket{\phi}\!\!\bra{\phi}_{AB} \otimes \ket{\phi'}\!\!\bra{\phi'}_{C}$ for some $\ket{\phi}_{AB}$ and $\ket{\phi'}_{C}$. It suffices to illustrate parameters for the causal model such that the condition of Eq.~\eqref{eq:q_triangle_compat} can be satisfied:     (i) let  $\H_{\MB} \cong \H_{B}$, $\H_{\MA} \cong \H_{A}$ and $\H_{\NC} \cong \H_{C}$,  where $\cong$ denotes the existence of an isometry, let $\H_{\NB}$ and $\H_{L}$ be arbitrary (for example, they could be 1-dimensional), and take $\varrho_{\MA \MB} = \ket{\phi}\!\!\bra{\phi}_{\MA \MB}$,  $\varrho_{\NB \NC} =  \varrho_{\NB} \otimes \ket{\phi'}\!\!\bra{\phi'}_{\NC} $ with $\varrho_{\NB}$ arbitrary, and $\varrho_{L}$ similarly arbitrary; 
 (ii) take
            $\E_{A|\LA \MA} = {\rm Tr}_{\LA} \otimes \mathcal{I}_{A|\MA} $, 
           i.e., an identity channel from $\H_{\MA}$ to $\H_{A}$ and a trace over $\H_{\LA}$, 
   $\E_{B|\MB \NB}= \mathcal{I}_{B|\MB} \otimes {\rm Tr}_{\NB}$, i.e., 
           an identity channel from $\H_{\MB}$ to $\H_{B}$ and a trace over $\H_{\NC}$, and 
         $\E_{C|\LC \NC} = {\rm Tr}_{\LC} \otimes \mathcal{I}_{C|\NC}$, i.e., an identity channel from $\H_{\NC}$ to $\H_{C}$ and a trace over $\H_{\LC}$.

Furthermore, among three-qubit pure states, those that factorize across a bipartition (i.e., the biseparable ones) are the {\em only} states that are triangle-compatible. In other words, in the case of three-qubit pure states, the boundary between triangle-compatible and triangle-incompatible coincides precisely with the boundary between those that factorize across a bipartition and those that do not. 
\begin{Theorem} \label{thm:main} 
For $A$, $B$ and $C$ qubits and a joint quantum state that is pure, $\rho_{ABC} = \ket{\psi}\!\!\bra{\psi}_{ABC}$, the state is triangle-compatible    if and only if it is factorizing across a bipartition (i.e., it is a biseparable pure state). 
\end{Theorem} 
The proof  of the `if' half is given above.  The proof of the `only if' half, which is given in~\Cref{app:proof_cor_pure_qubit_incompat}, consists in applying~\Cref{lem:xy_minus_x_y} to each of the marginals of $\rho_{ABC}$ and finding that the sufficient condition of~\Cref{prop:supp_ker_intersect}, namely Eq.~\eqref{eq:supp_ker_intersect}, holds, from which the result follows.

As noted in the introduction, Ref.~\cite{Kraft_21}  also considered the problem of compatibility of tripartite states with the triangle scenario, but restricting attention to local unitaries (from subsystems of the latent nodes to the visible nodes) rather than arbitrary local channels. In other words, as noted in the introduction, Ref.~\cite{Kraft_21}  studied the set of LUSR2WSE-preparable states. For the purpose of preparing {\em pure} tripartite states, the shared randomness is irrelevant, so that the free operations used in Ref.~\cite{Kraft_21} and in the above theorem are LU2WSE and LO2WSE respectively.  It is possible that in the case of preparing {\em pure} tripartite states, arbitrary quantum channels offer no additional power relative to unitary channels. If a proof of this latter claim were found, then our Theorem~\ref{thm:main} could be obtained as a corollary of Observation 5 from Ref.~\cite{Kraft_21}. Unfortunately, we have not been able to settle this question one way or the other.   In any event, the proof technique used here to establish Theorem~\ref{thm:main} is quite different from that used to establish Observation 5 in Ref.~\cite{Kraft_21}.

\subsection{Examples of incompatibility that are not distribution-witnessable}
\label{sec:redux}

In Sec.~\ref{sec:witnessable}, we noted that in assessing incompatibility of a quantum state with a given network via some inflation, it is not necessarily the case that this conclusion could be reached by first converting the state into a distribution through local measurements and then witnessing the incompatibility of this distribution using the same inflation DAG.  We now prove this result, which we formalize by the following proposition:
\begin{Proposition} \label{prop:state_vs_distr} 
There exist states that are $G$-incompatible, but for which this incompatibility is not distribution-witnessable with respect to the same inflation of $G$.
\end{Proposition}

Certain details of the proof are given in~\Cref{app:prop_state_vs_distr}, but the main thrust is given here. The proof considers the family of pure states $\rho_{ABC}(t) = \ket{\Psi(t)}\!\!\bra{\Psi(t)}$ where
\begin{align}
\ket{\Psi(t)} = \frac{\sqrt{t-2}}{\sqrt{t}} \ket{100}  + \frac{1}{\sqrt{t}}\ket{001} + \frac{1}{\sqrt{t}} \ket{010},
\end{align}
with $t$ is a real parameter in the interval $[3,\infty)$. We demonstrate that the triangle-incompatibility witness $I_{AB}(\rho_{ABC}(t))$ is not positive for all $t$ under consideration, meaning that every $\rho_{ABC}(t)$ is triangle-incompatible. To give a  concrete example: taking $t_{*}$ such that 
$\frac{\sqrt{t_{*}-2}}{\sqrt{t_{*}}}= 0.9$ (so that the amplitude of each of the other two terms is $\frac{1}{\sqrt{t_{*}}}= \sqrt{0.095}$), we have that 
\begin{widetext}
\begin{align}
I_{AB}(\rho_{ABC}(t_{*})) = \begin{bmatrix}
0.73305 & 0 & 0 & 0 & 0 & 0.277399 & 0 & 0 \\
 0 & 0.01805 & 0.095 & 0 & 0 & 0 & 0 & 0 \\
 0 & 0.095 & 0.17195 & 0 & 0 & 0 & 0 & 0.277399 \\
 0 & 0 & 0 & 0.07695 & 0 & 0 & 0 & 0 \\
 0 & 0 & 0 & 0 & 0.07695 & 0 & 0 & 0 \\
 0.277399 & 0 & 0 & 0 & 0 & 0.17195 & 0.095 & 0 \\
 0 & 0 & 0 & 0 & 0 & 0.095 & 0.01805 & 0 \\
 0 & 0 & 0.277399 & 0 & 0 & 0 & 0 & 0.73305 
\end{bmatrix}
\end{align}
\end{widetext}
which has a negative eigenvalue $-0.0529889$ of multiplicity $2$.

To demonstrate that there exist values for $t$ such that the triangle-incompatibility of $\rho_{ABC}(t)$ is not distribution-witnessable, let us introduce some notation for probability distributions arising from single-qubit measurements on $\rho_{ABC}(t)$. Let $P_{\phi, \chi, \psi}^{(t)}$ denote the distribution over three binary variables produced by measuring subsystem $A$ in the basis $\{\ket{\phi_{A}}\!\!\bra{\phi_{A}}, \ket{\phi_{A}^{\perp}}\!\!\bra{\phi_{A}^{\perp}} \}$, subsystem $B$ in the basis $\{\ket{\chi_{B}}\!\!\bra{\chi_{B}}, \ket{\chi_{B}^{\perp}}\!\!\bra{\chi_{B}^{\perp}} \}$ and subsystem $C$ in the basis $\{\ket{\psi_{C}}\!\!\bra{\psi_{C}}, \ket{\psi_{C}^{\perp}}\!\!\bra{\psi_{C}^{\perp}} \}$. We associate the outcome `$0$' to the first projection and `$1$' to the second projection in each case, so that
\begin{align}
P_{\phi, \chi, \psi}^{(t)}(000) := \bra{\phi_{A}\chi_{B}\psi_{C}}\rho_{ABC}(t)\ket{\phi_{A}\chi_{B}\psi_{C}}, \\
P_{\phi, \chi, \psi}^{(t)}(001) := \bra{\phi_{A}\chi_{B}\psi_{C}^{\perp}}\rho_{ABC}(t)\ket{\phi_{A}\chi_{B}\psi_{C}^{\perp}},
\end{align}
and so on. With $P_{\phi, \chi, \psi}^{(t)}$ defined in this way, we can consider whether or not $I_{AB}(P_{\phi, \chi, \psi}^{(t)}) \ge \bm{0}$. It is worth noting that 
\begin{align}
I_{AB}(P_{\phi, \chi, \psi}^{(t)}) = \bra{\phi_{A}\chi_{B}\psi_{C}}I_{AB}(\rho_{ABC}(t))\ket{\phi_{A}\chi_{B}\psi_{C}}. \label{eq:distr_state_witness_correspondence}
\end{align}

Ultimately, we aim to show that, for some $t$, $I_{AB}(P_{\phi, \chi, \psi}^{(t)}) \ge \bm{0}$ for {\em every} choice of $\phi$, $\chi$ and $\psi$. Note, however, that there is redundancy in considering all such choices since e.g., the probability vectors $P_{\phi, \chi, \psi}^{(t)}$ and $P_{\phi^{\perp}, \chi, \psi}^{(t)}$ are permutations of one another, so 
\begin{align}
I_{AB}(P_{\phi, \chi, \psi}^{(t)}) \ge \bm{0} \iff I_{AB}(P_{\phi^{\perp}, \chi, \psi}^{(t)}) \ge \bm{0}. 
\end{align}
Accordingly, it suffices to demonstrate that $I_{AB}(P_{\phi, \chi, \psi}^{(t)}(000)) \ge 0$, for every $\phi$, $\chi$ and $\psi$. To this end, let us define
\begin{align}
\iota_{AB}(t) &:= \min_{\ket{\phi_{A}},\ket{\chi_{B}}, \ket{\psi_{C}}} I_{AB}(P_{\phi,\chi,\psi}^{(t)}(000)). \label{eq:min_product_meas}
\end{align}
If $\iota_{AB}(t) \ge 0$, then $I_{AB}(P_{\phi, \chi, \psi}^{(t)}) \ge \bm{0}$ for all $\phi, \chi, \psi$, and the triangle-incompatibility of $\rho_{ABC}(t)$ is not distribution-witnessable (using the $AB$-Cut and the witnesses derived from the inequalities from Ref.~\cite{william_hall_05}).

The problem with trying to determine $\iota_{AB}(t)$ directly is that the optimization is over a non-convex set, namely the set of three-qubit pure product states. It is, however, possible to  consider a relaxation of this optimization. Let $S$ denote the set of all three-qubit states with positive partial transpose over each subsystem, that is:
\begin{align}
S := \{ \varrho_{ABC} : \varrho_{ABC}^{\top_{A}} \geq 0, \varrho_{ABC}^{\top_{B}} \geq 0, \varrho_{ABC}^{\top_{C}} \geq 0 \}.
\end{align}
In particular, $S$ is convex and contains the set of pure product states. Then, we define
\begin{align}
\tilde{\iota}_{AB}(t) &:= \min_{\varrho_{ABC} \in S} \Tr[\varrho_{ABC} I_{AB}(\rho_{ABC}(t))]. 
\end{align}
In light of Eq.~(\ref{eq:distr_state_witness_correspondence}), we are guaranteed that 
\begin{align}
\iota_{AB}(t) \geq \tilde{\iota}_{AB}(t), \label{eq:tilde_leq_non_tilde}
\end{align}
so non-negativity of the latter ensures non-negativity of the former.

The primary convenience of considering this relaxed optimization is that it can be formulated as a semidefinite program (SDP)\footnote{In its complex form, a semidefinite program can be formulated as 
\begin{gather*}
\min_{X} \Tr[Y_{0}X] \text{ subject to} \\
X \ge 0; \quad \Tr[Y_{i}X] \le b_{i}, i = 1, ..., m; 
\end{gather*}
with $X$ and the $Y_{i}$ Hermitian for $i = 0, ..., m$. In the case considered here, $I_{AB}(\rho_{ABC}(t))$ is Hermitian as it is a linear combination of Hermitian operators and the set $S$ is convex (the partial transpose is a linear operation), which induces linear constraints on the variable $\varrho_{ABC}$.}. The results of this optimization are plotted in~\Cref{fig:SDP_relax}, where rather than use the unbounded parameter $t$, we plot the result as a function of the amplitude of the $|011\rangle$ term, that is, as a function of $\frac{\sqrt{t-2}}{\sqrt{t}}$, which is a real parameter in the interval $[1/\sqrt{3}, 1)$.

\begin{figure}
\centering
\includegraphics[width=\columnwidth]{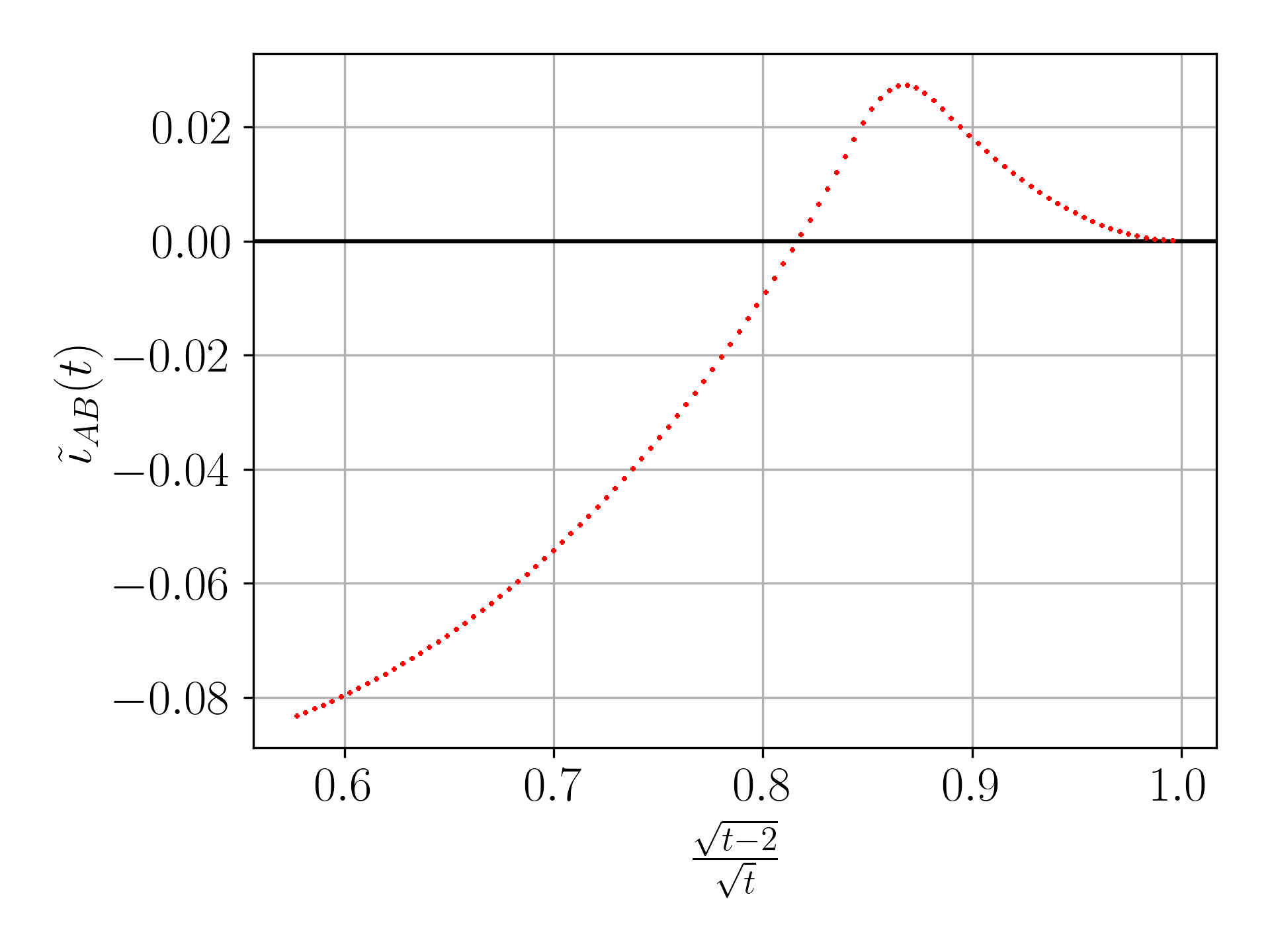}
\caption{The values of the SDP relaxation of the problem of finding local projective measurements witnessing the non-positivity of the $AB$-Cut-Bell-Wigner operators for the parameterized family of states $\rho_{ABC}(t)$. The code generating the depicted data is given in Ref.~\cite{smith_2025_15039119}.} 
\label{fig:SDP_relax}
\end{figure}

As can be seen in the figure, the value of $\tilde{\iota}_{AB}(t)$ is strictly positive when $\frac{\sqrt{t-2}}{\sqrt{t}} \gtrapprox 0.82$. By Eq.~(\ref{eq:tilde_leq_non_tilde}), this implies that the value of $\iota_{AB}(t)$ is also strictly positive in this region, and consequently that the incompatibility of states in this region are not distribution-witnessable by the $AB$-Cut inflation. The specific choice of $t_{*}$ considered above lies in this region, so $\rho_{ABC}(t_{*})$ is an example of a state that is triangle-incompatible, but for which this incompatibility is not distribution-witnessable (using the same Cut inflation and the witnesses $I_{xy}$).

\subsection{The compatibility problem for mixed states} \label{subsec:qubit_mixed}

For the case of pure three-qubit states considered above, all product and biseparable states were shown to be triangle-compatible in~\Cref{subsec:qubit_pure}. However, we have already seen that the analogous statement is not true for mixed states: recall from above that the separable states $\rho^{({\rm GHZdistn})}$ and $\rho^{({\rm Wdistn})}$ that encode the classical distributions $P^{({\rm GHZ})}$ and $P^{({\rm W})}$ respectively, are triangle-incompatible. In general, any state $\rho^{({\rm distn})}$ that encodes a classical distribution in this way is triangle-incompatible if and only if the corresponding distribution is triangle-incompatible.

Since characterizing the precise set of  probability distributions on three binary variables that are compatible with the triangle scenario in a quantum causal model is computationally challenging (for every non-fanout inflation it involves an infinite hierarchy of semi-definite programs~\cite{navascues2020inflation}),  a full classification of the  three-qubit mixed states that are  compatible with the triangle scenario is at least as difficult.  

Nonetheless, one can use the technique described here to easily derive witnesses of incompatibility for mixed quantum states.   

Some of the results of Ref.~\cite{Navascues_20} also provide a means for witnessing triangle-incompatibility, using ideas from the inflation technique.   In fact, Ref.~\cite{Navascues_20} sought to witness incompatibility for a slightly different network: the triangle supplemented by three-way shared randomness. But if one cannot realize a state with the shared randomness, then one clearly cannot realize that state without it, so any incompatibility in the former case implies triangle-incompatibility. 
Specifically, the witnesses considered were based on the fidelities between $\rho_{ABC}$ 
  and the GHZ and W states (defined, respectively, in Eqs.~\eqref{eq:GHZstate} and\eqref{eq:Wstate}). One witnesses the incompatibility of $\rho_{ABC}$ with the shared-randomness-supplemented triangle (and hence with the triangle) whenever these fidelities satisfy 
\begin{gather}\label{eq:fidelitybounds}
\braket{GHZ|\rho_{ABC}|GHZ} \ge \frac{1+\sqrt{3}}{4} \simeq 0.6830, \\
\braket{W|\rho_{ABC}|W} \ge 0.7602,
\end{gather}
Our technique can witness triangle-incompatibilities that are not detected by the fidelity witnesses of Ref.~\cite{Navascues_20}.  This is the case even if we limit ourselves to the lowest level of the hierarchy of nonfanout inflations (i.e., we just consider the Cut inflation), while the bounds on the fidelities just mentioned are derived by considering higher levels of this hierarchy.

Consider the following mixed state:
\begin{align}
\omega_{ABC} = 0.3\ket{\psi_{1}}\!\!\bra{\psi_{1}} + 0.7\ket{0\!+\!+}\!\!\bra{0\!+\!+}    
\end{align}
where
\begin{align}
\ket{\psi_{1}} := 0.9\ket{000} + \sqrt{0.19/2}\ket{101} + \sqrt{0.19/2}\ket{110}.
\end{align}
By a straightforward calculation, one can compute the operator  $I_{AB}(\omega_{ABC}) $ and determine that its eigenvalues are
\begin{align}
\{0.78839,  0.195995, 0.0351218, -0.0195072 \},
\end{align}
each with multiplicity $2$. The occurrence of the negative eigenvalue clearly indicates the non-positivity of $I_{AB}(\omega_{ABC}) $ and hence the triangle-incompatibility of $\omega_{ABC}$. 

On the other hand, the fidelities between $\omega_{ABC}$ and the GHZ and W states  are
\begin{gather}
\braket{GHZ|\omega_{ABC}|GHZ} = 0.209, \\
\braket{W|\omega_{ABC}|W} = 0.233333,
\end{gather} 
and hence these fidelities do not satisfy the condition in Eq.~\eqref{eq:fidelitybounds} for witnessing incompatibility. 

\subsubsection{Probing triangle-incompatibility of one-parameter families of mixed states}

We saw for the case of pure three-qubit states that the boundary between compatible and incompatible states coincided exactly with the boundary between biseparable states and not biseparable states. We also saw above that in the case of mixed states, the compatibility-incompatibility boundary is far less clear as there are separable states that are triangle-incompatible.

One method to probe this boundary further in the case of mixed states is by adding noise to a pure state and seeing how the positivity of a triangle-incompatibility witness behaves. Here, we consider a simple one parameter noise model that adds white noise to a given pure state. Explicitly, for a given pure three-qubit state $\ket{\Psi}$, we consider the family indexed by $p \in [0,1]$ defined by
\begin{align}
\hat{\rho}(\ket{\Psi},p) := p\ket{\Psi}\!\!\bra{\Psi} + \frac{1-p}{8}\mathbb{1}_{ABC} \label{eq:depolarised}
\end{align}
where $\mathbb{1}_{ABC}$ is the identity on $\H_{A} \otimes \H_{B} \otimes \H_{C}$. For any $p < 1$, $\hat{\rho}(\ket{\Psi},p)$ is a full rank state. We have the following result:
 
\begin{Proposition} \label{prop:mixing_value} For any $\ket{\Psi}$, there exists a $q \in [3 - 2\sqrt{2},1)$ such that for all $0 \leq p \leq q$, $\hat{\rho}(\ket{\Psi},p)$ is not demonstrably triangle-incompatible using the triangle-incompatibility witnesses of Eq.~(\ref{eq:q_I_xy}), i.e., $I_{xy}(\hat{\rho}(\ket{\Psi},p)) \ge 0$ for each $xy \in \{AB,AC,BC\}$.
\end{Proposition}
The proof is given in~\Cref{app:proof_prop_mixing_value}. The result holds for any choice of the pure state $\ket{\Psi}$ appearing in Eq.~(\ref{eq:depolarised}), which naturally includes all $\ket{\Psi}$ that are not biseparable (which we know to be demonstrably triangle-incompatible using one of the witnesses $I_{xy}$ for $xy \in \{AB,AC,BC\}$). 

In general, the value of $q$ will be higher than $3 - 2\sqrt{2}$. For instance, if we take $\ket{\Psi} = \ket{GHZ} = \frac{1}{\sqrt{2}}(\ket{000} + \ket{111})$, making $\hat{\rho}$ a generalized three-qubit Werner state~\cite{Siewert_12,Werner_89}, we get that $q = \frac{1}{2}$. In this case, all the witnesses $I_{xy}(\hat{\rho}(\ket{GHZ},p))$ are non-positive for $p > \frac{1}{2}$, thereby witnessing the triangle-incompatibility of $\hat{\rho}(\ket{GHZ},p)$ for these values (see~\Cref{app:proof_prop_mixing_value}). 

Does this tell us anything regarding the relation (or lack thereof) between the entangled-separable boundary and the boundary between triangle-compatible and triangle-incompatible states? As demonstrated in Ref.~\cite{Siewert_12}, the state $\hat{\rho}(\ket{GHZ},p)$ is entangled for any $p \gtrapprox 0.69554$ (as witnessed by the positivity of the three-tangle for those values). It follows that the entanglement boundary and triangle-incompatibility boundary are indeed distinct for this simple class of states: for $\frac{1}{2} < p < 0.6955$, the states $\hat{\rho}(\ket{GHZ},p)$ are triangle-incompatible but not entangled.~\Cref{app:proof_prop_mixing_value} includes some details of the calculations of these $q$-values, also for the case where $\ket{\Psi} = \ket{W} = \frac{1}{\sqrt{3}}(\ket{001} + \ket{010} + \ket{100})$ which produces $q \approx 0.627$.

An alternative family of Werner-like three-qubit mixed states was studied in Ref.~\cite{Toth_06}. This single-parameter family is defined by
\begin{align}
\rho^{(3,c)} &:= \frac{1}{8} \mathbb{1}_{ABC} + \sum_{k=X,Y,Z} \frac{1}{24}  \mathbb{1}_{A} \otimes \sigma_{B}^{(k)} \otimes \sigma_{C}^{(k)} \nonumber \\
& \quad - \frac{c}{16}\left( \sigma_{A}^{(k)} \otimes \mathbb{1}_{B} \otimes \sigma_{C}^{(k)} + \sigma_{A}^{(k)} \otimes \sigma_{B}^{(k)} \otimes \mathbb{1}_{C} \right)
\end{align}
where $c$ is a real parameter, the $\sigma^{(k)}$ denote the Pauli matrices, and the subscripts indicate the subsystems they act upon. This family is interesting since it permits a local hidden variable model for projective measurements whenever $c \leq 1$ (c.f.~\cite[Thm 2,][]{Toth_06}) and is not biseparable for $c > \frac{1}{3}(\sqrt{13} - 1)  \approx 0.8685$. In particular, this means that, for $0.8685 < c < 1$, the states $\rho^{(3,c)}$ are multipartite entangled, but nevertheless, the probability distributions associated to projective measurements on $\rho^{(3,c)}$ admit local hidden variable models. If $c = 0$, one observes that the $\rho^{(3,c)}$ factorizes across the $A|BC$ partition (and moreover is maximally mixed on the $A$ subsystem). By analogous arguments to the provided for biseparable pure states, this state can be seen to be triangle-compatible. It turns out that this state is the only member of the above family that is; all the rest are triangle-incompatible. 
\begin{Proposition} \label{prop:Toth_Acin_states} The state $\rho^{(3,c)}$ is triangle-incompatible if and only if $c \neq 0$.
\end{Proposition}
The proof is given in~\Cref{app:proof_prop_Toth_Acin_states}. For this class of states then, it appears that the incompatibility-compatibility boundary has less to do with the separability than the {\em factorizability} of the state in question.

In sum, apart from providing further evidence that the triangle-incompatibility witnesses introduced above are broadly applicable to mixed states, the results presented in this subsection point to the greater nuance involved in attempting to categorize the triangle-compatibility of mixed as opposed to pure states.

\section{Visible nodes beyond qubits} 
\label{sec:higher_dim_nodes}

The operator inequality in \eqref{eq:Hall_op_ineq} was established by Hall for any finite dimensional Hilbert space. Thus, the methodology of analyzing the positivity of the associated triangle-incompatibility witnesses remains valid in the cases where the visible nodes are not qubits. In this section, we provide a couple of examples of states on higher-dimensional systems that are triangle-incompatible.

\subsection{Two qubits and one ququart, pure states}

Our first higher-dimensional example is to consider pure states where $\H_{A} \cong \H_{B} \cong \mathbb{C}^{2}$ and $\H_{C} \cong \mathbb{C}^{4}$. Unlike the three-qubit case, it is readily seen that for such a choice of systems, there are triangle-compatible states beyond those that factorize across a bipartition.  It suffices to note that the ququart can be conceptualized as a pair of qubits, say $C^{(1)}$ and $C^{(2)}$, and it is clear that in the triangle scenario one can prepare an entangled state on $AC^{(1)}$ and an entangled state on $BC^{(2)}$. 

This of course does not imply that {\em any} pure qubit-qubit-ququart state can be prepared. For example, since a qubit can be isometrically embedded in a ququart, there are qubit-qubit-ququart states that are local isometry-equivalent to the three-qubit pure states that are triangle-incompatible, and hence are themselves   triangle-incompatible. However, there exist qubit-qubit-ququart states that do not reduce to three-qubit states and we now demonstrate that we can use the triangle-incompatibility witnesses presented above to prove the triangle-incompatibility of some such states, including those exhibiting GHZ-like entanglement.

To consider pure states on this qubit-qubit-ququart system, it is possible to use the higher-dimensional generalised Schmidt decomposition given in Ref.~\cite{carteret2000multipartite}. A pure state $\ket{\Psi} \in \H_{A} \otimes \H_{B} \otimes \H_{C}$ can be written as
\begin{align}
\ket{\Psi} &= \alpha_{0}e^{i\phi_{0}}\ket{000} + \alpha_{1}\ket{011} + \alpha_{2}\ket{101} + \alpha_{3}e^{i\phi_{1}}\ket{102} \nonumber \\
&\quad + \alpha_{4}\ket{110} + \alpha_{5}\ket{111} + \alpha_{6}\ket{112} +\alpha_{7}\ket{113} \label{eq:hd_gen_schmidt_224}
\end{align}
where $\alpha_{l} \in \mathbb{R}_{\geq 0}$, $\sum_{l}\alpha_{l}^{2} = 1$ and $\phi_{r} \in [0, 2\pi)$ (see~\Cref{app:hdim_gen_schmidt} for a summary of how this form is obtained from the general form in~\cite{carteret2000multipartite}). We then have the following:
\begin{Proposition} \label{prop:GHZ_224} If $\rho_{ABC} = \ket{\Psi}\!\!\bra{\Psi}$ where
\begin{align}
\ket{\Psi} = \alpha_{0}e^{i\phi_{0}}\ket{000} + \alpha_{4}\ket{110} + \alpha_{5}\ket{111} + \alpha_{6}\ket{112} +\alpha_{7}\ket{113}
\end{align}
with $0 < \alpha_{0}^{2} + \alpha_{4}^{2} < 1$ (i.e. at least one of $\alpha_{5}$, $\alpha_{6}$ and $\alpha_{7}$ is non-zero), then $\rho_{ABC}$ is triangle-incompatible and this incompatibility is distribution-witnessable.
\end{Proposition}

The proof is given in~\Cref{app:proof_prop_GHZ_224} and proceeds by demonstrating that $I_{AC}(\ket{\Psi}\!\!\bra{\Psi})$ is not positive for the relevant values of the $\alpha$ parameters. It should be noted that the above result is consistent with known results establishing the impossibility of generating qudit graph states, such as the GHZ state, in networks with bipartite sources (see e.g.,~\cite{wang2024quantum}), but extends also to other local-unitary inequivalent states that exhibit GHZ-like entanglement.

\subsection{Three qutrits, mixed states}

For the case of three qubits, we saw a number of distinctions between pure and mixed states vis-a-vis triangle-compatibility. Here, we give one final example in a similar vein for the case of three qutrits. In this example, we consider a pure three-qutrit state and a mixed three-qutrit state which have the same level of multipartite entanglement as quantified by the generalized geometric measure of entanglement (GGM)~\cite{Sen(De)_10,shimony1995degree,barnum2001monotones}, where 
for the former, triangle-incompatibility can be demonstrated using the triangle-incompatibility witness while for the latter it cannot.

The GGM is a measure of entanglement for multipartite quantum states. For a given $n$-partite pure state $\ket{\psi}$, the GGM of $\ket{\psi}$ is essentially a measure of how close $\ket{\psi}$ is to a non-genuinely multipartite entangled state. Formally, the GGM is defined to be \cite{Sen(De)_10,shimony1995degree,barnum2001monotones,Das_16}
\begin{align}
\textrm{GGM}(\ket{\psi}) := 1 - \max_{\ket{\phi}}|\braket{\psi|\phi}|^{2},
\end{align}
where the maximization is over all states $\ket{\phi}$ that are not genuinely multipartite entangled. For mixed states, the GGM is given by a convex roof construction, i.e., for a mixed state $\rho$, the GGM is given by
\begin{align}
\textrm{GGM}(\rho) := \min_{\{p_{i}, \ket{\psi_{i}} \}} \sum_{i} p_{i}\textrm{GGM}(\ket{\psi_{i}}),
\end{align}
where the minimization is over all decompositions of $\rho$ into convex mixtures of pure states $\ket{\psi_{i}}$. In Ref.~\cite{Das_16}, the GGM was calculated for a number of mixed states exhibiting certain symmetries. In particular, these mixed states were obtained from specific pure states via a certain choice of twirling operation, allowing the GGM of the former to be related to that of the latter. One example of this pure state - mixed state pairing is given by
\begin{gather}
\ket{\Psi_{3,3}^{(3)}} := \sqrt{p_{0}} \ket{\Psi_{0}} + \sqrt{p_{1}}\ket{\Psi_{1}} + \sqrt{1-p_{0}-p_{1}}\ket{\Psi_{2}}, \\
\rho_{3,3}^{(3)} := p_{0}\ket{\Psi_{0}}\!\!\bra{\Psi_{0}} + p_{1}\ket{\Psi_{1}}\!\!\bra{\Psi_{1}} + (1-p_{0}-p_{1})\ket{\Psi_{2}}\!\!\bra{\Psi_{2}},
\end{gather}
where $p_{0},p_{1}, p_{0}+p_{1} \in [0,1]$ and where $\ket{\Psi_{i}}$ are the three-qutrit states defined to be an equal superposition of computational basis states $\ket{xyz}$ such that $x+y+z = i \bmod 3$, that is,
\begin{align}
\ket{\Psi_{0}} := \frac{1}{3}&\left(\ket{000} + \ket{111} + \ket{222} \right. \nonumber\\
&\left. \ket{012} + \ket{021} + \ket{102} \right. \nonumber\\
&\left. \ket{120} + \ket{201} + \ket{210} \right),\\
\ket{\Psi_{1}} := \frac{1}{3}&\left( \ket{001} + \ket{010} + \ket{100} \right. \nonumber\\
&\left. \ket{022} + \ket{202} + \ket{220} \right. \nonumber\\
&\left. \ket{112} + \ket{121} + \ket{211} \right),\\
\ket{\Psi_{2}} := \frac{1}{3}&\left( \ket{011} + \ket{101} + \ket{110} \right. \nonumber\\
&\left. \ket{002} + \ket{020} + \ket{200} \right. \nonumber\\
&\left. \ket{122} + \ket{212} + \ket{221} \right).
\end{align}
Let $Z_{3}$ denote the qutrit analogue of the qubit Pauli-$Z$ operator, i.e., $Z_{3} = \sum_{j=0}^{2} e^{\frac{2\pi i j}{3}}\ket{j}\!\!\bra{j}$. For the twirling operator $\boldsymbol{P}: \mathcal{L}(\H_{A} \otimes \H_{B} \otimes \H_{C}) \rightarrow \mathcal{L}(\H_{A} \otimes \H_{B} \otimes \H_{C})$ defined via
\begin{align}
\rho \mapsto \frac{1}{3}\sum_{k=0}^{2}\left( Z_{3}^{k} \otimes Z_{3}^{k} \otimes Z_{3}^{k}\right)\rho (Z_{3}^{k} \otimes Z_{3}^{k} \otimes Z_{3}^{k})^{\dagger} 
\end{align}
where $Z_{3}^{k}$ denotes the $k$th power of $Z_{3}$, we have that $\boldsymbol{P}(\ket{\Psi_{3,3,}^{(3)}}\!\!\bra{\Psi_{3,3}^{(3)}}) = \rho_{3,3}^{(3)}$. 

As a consequence of the above relation via $\boldsymbol{P}$, the GGM of $\ket{\Psi_{3,3}^{(3)}}$ and that of $\rho_{3,3}^{(3)}$ coincide. However, when it comes to proving the triangle-compatibility of these states by the witness $I_{AB}$, there is a difference between them. It turns out that $I_{AB}(\rho_{3,3}^{(3)})$, $I_{AC}(\rho_{3,3}^{(3)})$, and $I_{BC}(\rho_{3,3}^{(3)})$ have no dependence on $p_{0}, p_{1}$ and moreover have eigenvalues $\frac{7}{9}$ with multiplicity $3$, $\frac{4}{9}$ with multiplicity $12$ and $\frac{1}{9}$ also with multiplicity $12$. We can thus not draw any conclusions about the triangle-compatibility of $\rho_{3,3}^{(3)}$ using this method alone. 

Conversely, taking e.g., $p_{0} = 2p_{1} = 0.5$, one finds that $I_{AB}(\ket{\Psi_{3,3}^{(3)}}\!\!\bra{\Psi_{3,3}^{(3)}})$ has an eigenvalue equal to $-0.01348$ and so $\ket{\Psi_{3,3}^{(3)}}$ is triangle-incompatible for these values of $p_{0}, p_{1}$. Note that other values of $p_{0},p_{1}$ also exist for which incompatibility can be demonstrated.

\section{Networks beyond the triangle scenario} \label{sec:beyond_triangle}

To this point, we have considered causal compatibility constraints for the triangle scenario. It is natural to ask: how readily can such constraints be derived for other causal structures? In this section, we demonstrate that the method presented in \Cref{subsec:q_marginal_ineqs} can be applied to Cut inflations of a family of networks which includes the triangle scenario as a special case. 

We begin by summarizing the method for deriving the triangle-incompatibility witnesses. From a Cut inflation $G'$ of the triangle scenario $G$, a set of expressible subsets of the visible nodes of $G'$ 
\begin{align}
S = \{\{A_{1},B_{1}\}, \{A_{1},C_{1}\}, \{B_{1},C_{1}\},\{A_{1}\},\{B_{1}\},\{C_{1}\}\}
\end{align}
is obtained. For a given state $\rho_{ABC}$ on the visible nodes of $G$, by taking $S$ to be a set of marginal contexts, we can consider a set of states $\{\sigma_{{\bf X}} : {\bf X} \in S\}$, with the nature of the expressible sets of $G'$ allowing each of these states to be written in terms of marginals of $\rho_{ABC}$. The failure of marginal compatibility of $\{\sigma_{{\bf X}} : {\bf X} \in S\}$ demonstrates the triangle-incompatibility of $\rho_{ABC}$, and can be established by a violation of the operator inequality $\Delta(\{\sigma_{\bf X}: {\bf X}\in S\}) \ge 0$ proven by Hall \cite{william_hall_05} (recall \Cref{eq:Hall_ineqs} for the definition of $\Delta$). The triangle-incompatibility witness is then given by the operator $\Delta(\{\sigma_{\bf X}: {\bf X}\in S\})$ written in terms of the marginals of $\rho_{ABC}$.

There are two features of the Cut inflation of the triangle that allow the above method to work. First, the set of 
 expressible sets $S$ derived from the Cut inflation is such that $S = 2^{{\bf V}}\setminus {\bf V}$ for ${\bf V}$ an odd cardinality set of quantum systems, namely ${\bf V} = \{A_{1}, B_{1}, C_{1}\}$ in this case. The operator inequalities proven by Hall are of the form $\Delta(\{\sigma_{\bf X}: {\bf X}\in 2^{\bf V}\!\setminus\!{\bf V}  \}) \ge 0$ for odd cardinality ${\bf V}$, so the correspondence between $S$ and $2^{{\bf V}}\setminus {\bf V}$ is a requirement. The second feature is that {\em every} element of $S$ is an expressible set in $G'$. This allows every term $\Delta(\{\sigma_{\bf X}: {\bf X}\in 2^{\bf V}\!\setminus\!{\bf V}\})$ to be written in terms of marginals of the state $\rho_{{\bf V}}$, whose compatibility with the original causal model $G$ is being investigated, producing the desired $G$-incompatibility witness.

In the above case, the set ${\bf V}$ corresponds precisely to the set of visible nodes of the inflation $G'$, i.e. every subsystem specified in ${\bf V}$ is a single subsystem in $\textsf{Vnodes}(G')$. However, this is not a requirement: at least \textit{a priori} the set ${\bf V}$ could contain elements that are composite systems comprising multiple visible nodes from $G'$. Such a scenario is not relevant when considering Cut inflations of the triangle scenario due to the small cardinality of $\textsf{Vnodes}$ in this case ($|{\bf V}| = 3$ is the minimal case for which $\Delta(\{\sigma_{\bf X}: {\bf X}\in 2^{\bf V}\!\setminus\!{\bf V}  \}) \ge 0$ is non-trivial), but for networks with larger numbers of nodes such scenarios can arise, as we demonstrate in the examples below.

Before turning to the examples, let us comment further on how the method for deriving candidate $G$-incompatibility witnesses works in the case where ${\bf V}$ contains composite systems.\footnote{Note that we are not claiming that these witnesses are always guaranteed to be non-trivial; \textit{a priori} some of the expressions obtained may be positive semidefinite for all states. One example of a trivial witness derived using the nonfanout inflation of the triangle scenario known as the Ring inflation is presented in \Cref{app:ring_inflation}.} Suppose that an inflation $G'$ of some network $G$ has visible nodes $\textsf{Vnodes}(G') = \{ A_{1}^{[1]}, ..., A_{1}^{[n]}\}$\footnote{Regarding notation: the square-bracketed superscript is used here to have a way of labelling visible nodes with a numerical index rather than by using different letters. It is distinct to how round-bracketed superscripts have been used to label the factor spaces in a factorization of the Hilbert space of a latent node.} for some $n \ge 3$. Let $\{\mathcal{S}_{1},...,\mathcal{S}_{k}\}$ denote a set partition of $\textsf{Vnodes}(G')$ where $k \ge 3$ is odd.\footnote{Recall that a set partition $\{\mathcal{S}_{1},...,\mathcal{S}_{k}\}$ of a set $\mathcal{X}$ is a set of subsets of $\mathcal{X}$ such that $\cup_{i=1}^{k} \mathcal{S}_{i} = \mathcal{X}$ and $\mathcal{S}_{i} \cap \mathcal{S}_{j} = \emptyset$ for all $i\neq j \in \{1,...,k\}$.} By taking ${\bf V} := \{\mathcal{S}_{1},...,\mathcal{S}_{k}\}$ and by considering, for each ${\bf X} \subseteq {\bf V}$, the state $\sigma_{{\bf X}}$ to be a state on the composite system comprised of the visible nodes $\cup_{\mathcal{S}_{j} \in {\bf X}} S_{j} \subseteq \textsf{Vnodes}(G')$, we can again consider the operator $\Delta(\{\sigma_{\bf X}: {\bf X}\in 2^{\bf V}\!\setminus\!{\bf V}  \})$ as acting on the Hilbert space associated to $\textsf{Vnodes}(G')$. Moreover, if each element of $2^{{\bf V}}\setminus {\bf V}$ is injectable or expressible, since $k$ is odd, it is again possible to derive an incompatibility witness from the inequality $\Delta(\{\sigma_{\bf X}: {\bf X}\in 2^{\bf V}\!\setminus\!{\bf V}  \}) \ge 0$.

In the following two subsections, we consider two examples that illustrate the above method for larger numbers of visible nodes. In the first subsection, we consider the pentagon scenario, shown in \Cref{fig:pent_hex}(a), which has a similar network structure to the triangle scenario but has $5$ visible nodes instead of $3$. By considering a Cut inflation of the pentagon scenario, as in \Cref{fig:pent_hex}(b), we can obtain pentagon-incompatibility witnesses both by taking ${\bf V}$ to be the full set of visible nodes, meaning that $|{\bf V}| = 5$, as well as by taking ${\bf V}$ to consist of composite systems, in which case $|{\bf V}| = 3$. In the second subsection, we consider the woven hexagon scenario, shown in \Cref{fig:pent_hex}(c), consisting of $6$ visible nodes and exhibiting a higher degree of connectivity between latent and visible nodes than that of the triangle and pentagon scenarios. Using a Cut inflation for this network (cf \Cref{fig:pent_hex}(d)), we are only able to produce an incompatibility witness by taking ${\bf V}$ to consist of composite systems of visible nodes (with the result that $|{\bf V}| = 3$).

\begin{figure*}[htbp] 
\centering
\includegraphics[width=0.4\textwidth]{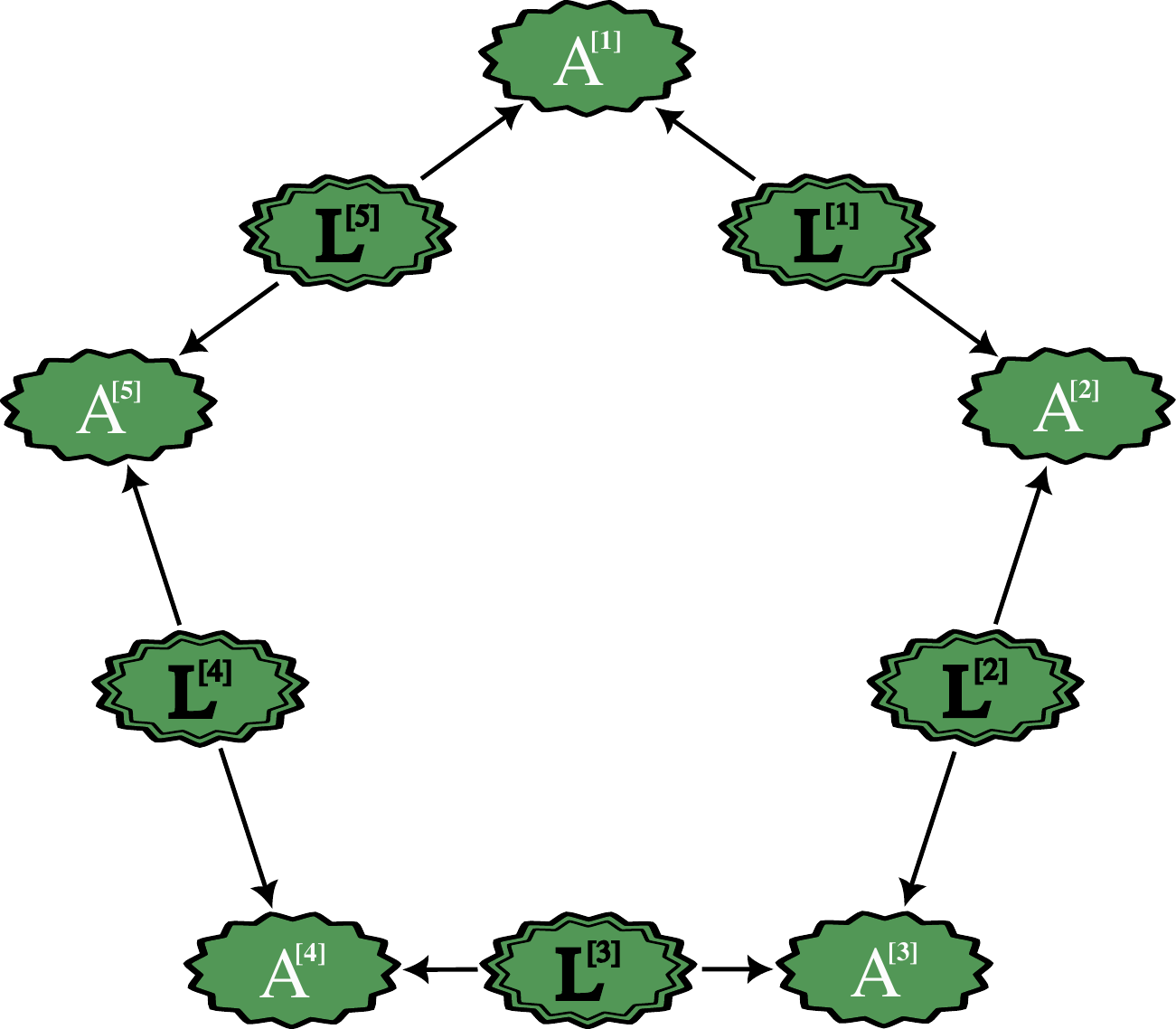} 
\hfill
\includegraphics[width=0.4\textwidth]{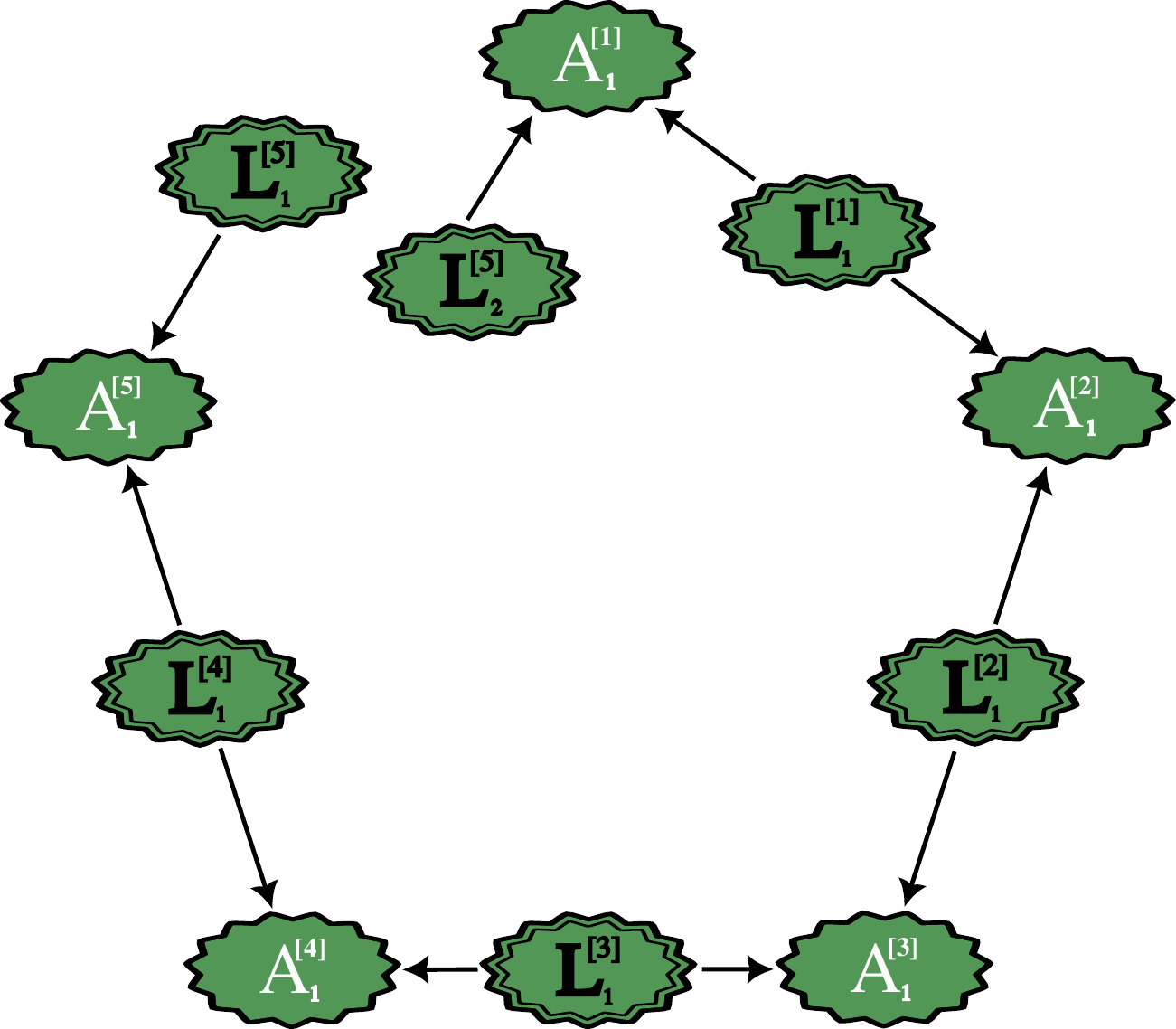}
\vspace{0.5cm}
\parbox{0.4\textwidth}{\centering (a)} \hfill
\parbox{0.4\textwidth}{\centering (b)}
\vspace{0.1cm}
\includegraphics[width=0.4\textwidth]{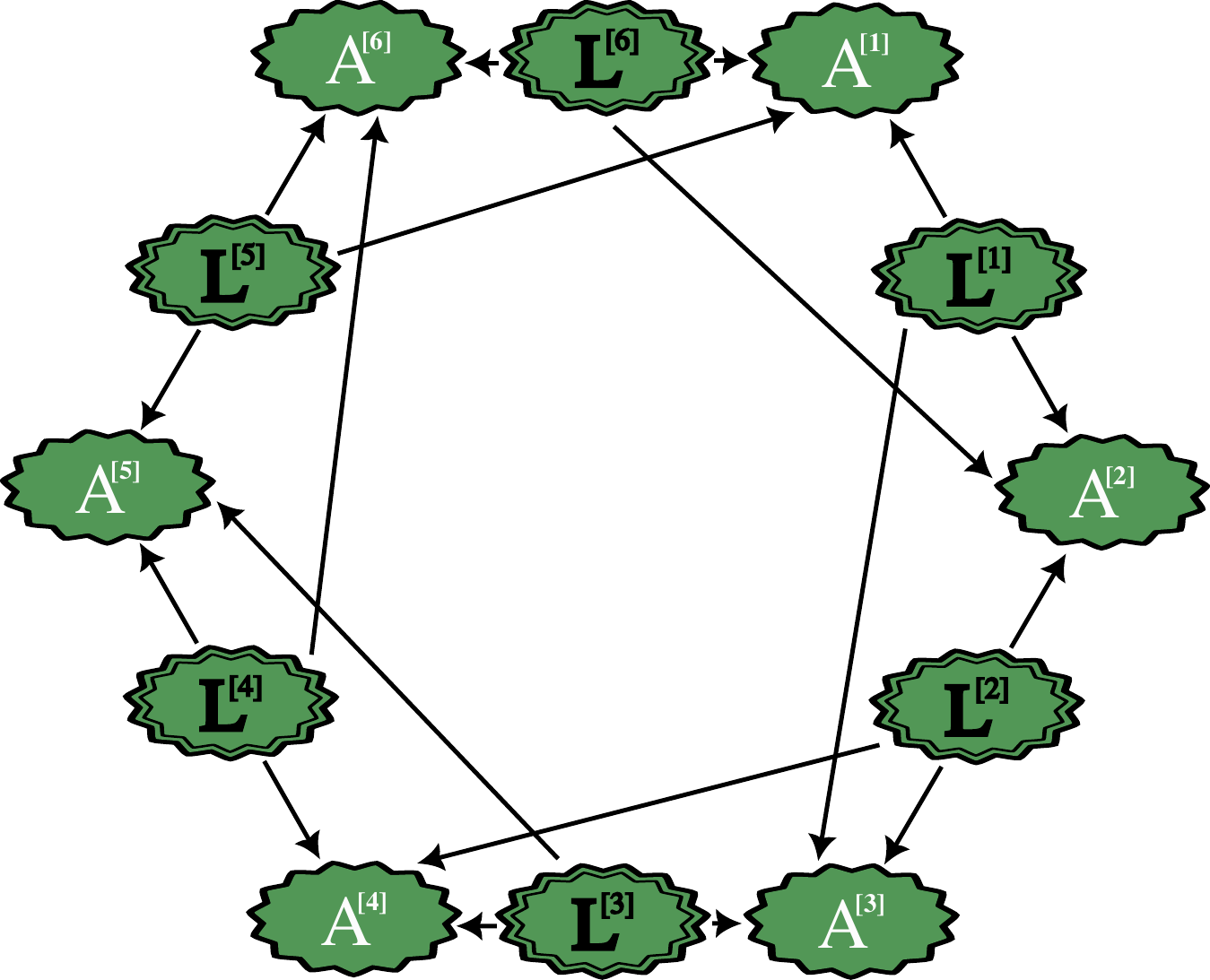} 
\hfill
\includegraphics[width=0.4\textwidth]{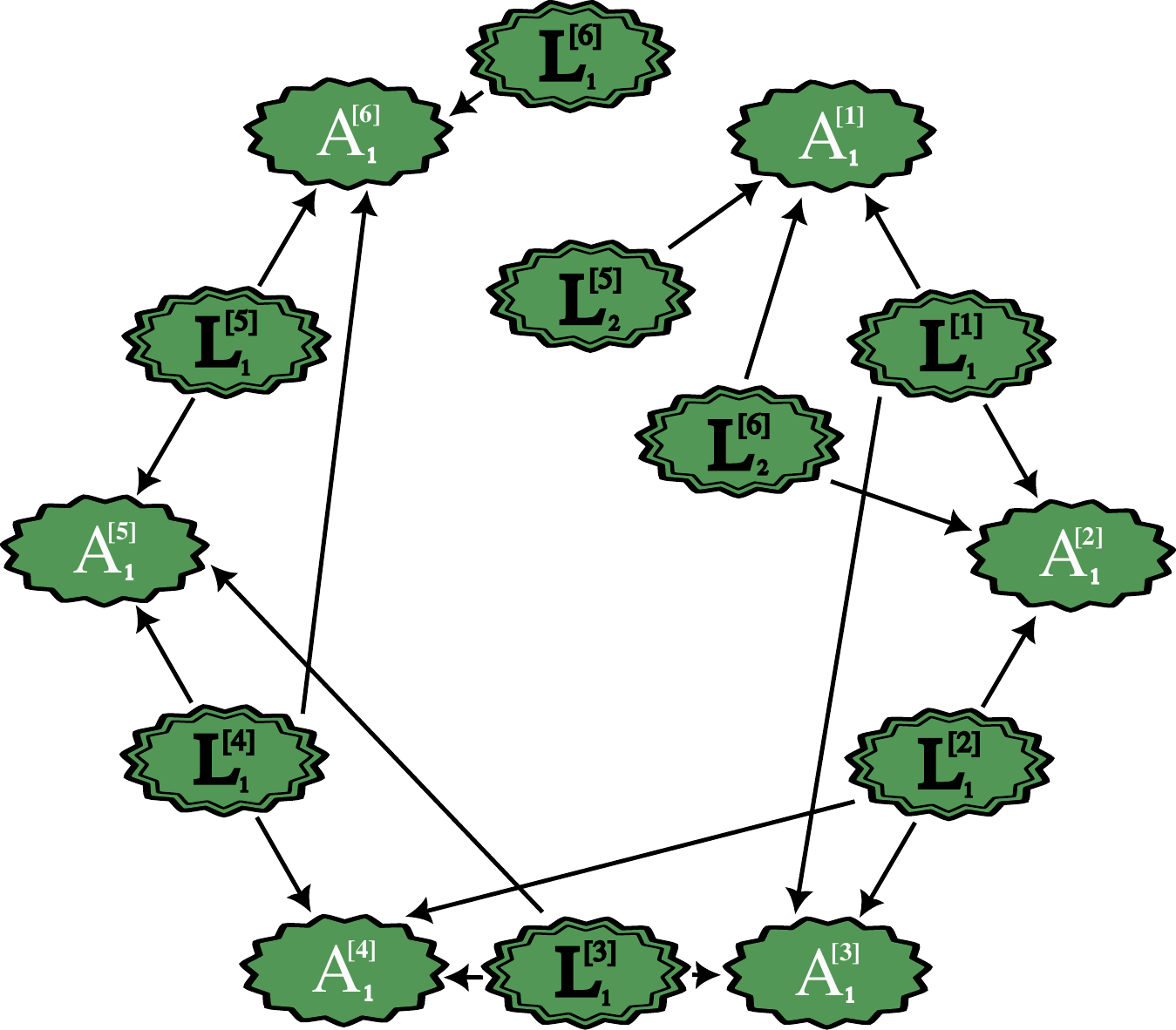} 
\parbox{0.4\textwidth}{\centering (c)} \hfill
\parbox{0.4\textwidth}{\centering (d)}
\caption{(a) The pentagon scenario, corresponding to $G(5,2)$ in the notation of the family of networks in the main text. (b) A Cut inflation of the pentagon scenario for the cut between $A_{1}^{[5]}$ and $A_{1}^{[1]}$, corresponding to $G'(5,2,5)$ in the main text. (c) The woven hexagon scenario, corresponding to $G(6,3)$ in the main text. (d) A Cut inflation of the woven hexagon scenario, corresponding to $G'(6,3,6)$ in the main text.}
\label{fig:pent_hex}
\end{figure*}

\subsection{Example: pentagon scenario}

Let $G$ denote the pentagon scenario, shown in \Cref{fig:pent_hex}(a), and let $G'$ denote the Cut inflation of $G$, shown in \Cref{fig:pent_hex}(b). We will start to use notation here consistent with that used below in the treatment of the general family of networks for which incompatibility witnesses can be derived. In particular, we have that the visible and latent nodes for $G$ are 
\begin{equation}
\begin{gathered}
\textsf{Vnodes}(G) = \{A^{[1]}, A^{[2]},A^{[3]},A^{[4]},A^{[5]}\}, \\
\textsf{Lnodes}(G) = \{L^{[1]}, L^{[2]},L^{[3]},L^{[4]},L^{[5]} \}, 
\end{gathered}
\end{equation}
while the visible and latent nodes for the Cut inflation $G'$ are
\begin{equation}
\begin{gathered}
\textsf{Vnodes}(G') = \{A_{1}^{[1]}, A_{1}^{[2]},A_{1}^{[3]},A_{1}^{[4]},A_{1}^{[5]}\}, \\
\textsf{Lnodes}(G') = \{L_{1}^{[1]}, L_{1}^{[2]},L_{1}^{[3]},L_{1}^{[4]},L_{1}^{[5]},L_{2}^{[5]} \}.
\end{gathered}
\end{equation}
Note that we are considering a cut between the visible nodes $A_{1}^{[5]}$ and $A_{1}^{[1]}$ of $G'$, as depicted in \Cref{fig:pent_hex}(b), however Cut inflations of $G$ can be defined for cuts introduced between any other pair of visible nodes, just as with the triangle scenario.

For the choice of labelling of the visible and latent nodes made here, the connectivity of $G$ is given as: $L^{[j]}$ is the parent of $A^{[j]}$ and $A^{[j+1]}$ for each $j \in \{1,...,5\}$, where any index value greater than $5$ is taken modulo $5$ (i.e., $L^{[5]}$ is the parent of both $A^{[5]}$ and $A^{[1]}$). The connectivity of $G'$ is given as: $L_{1}^{[j]}$ is the parent of $A_{1}^{[j]}$ and $A_{1}^{[j+1]}$ for $j \in \{1, ..., 4\}$, $L_{1}^{[5]}$ is the parent of $A_{1}^{[5]}$ alone and $L_{2}^{[5]}$ is the parent of $A_{1}^{[1]}$ alone. As is demonstrated in \Cref{app:fam_networks}, the sets 
\begin{gather}
\{A_{1}^{[1]},A_{1}^{[2]},A_{1}^{[3]},A_{1}^{[4]}  \}, \quad  \{A_{1}^{[2]},A_{1}^{[3]},A_{1}^{[4]},A_{1}^{[5]}   \}
\end{gather}
and all subsets thereof are injectable in $G'$, while the sets 
\begin{equation}
\begin{gathered}
\{A_{1}^{[1]},A_{1}^{[2]},A_{1}^{[3]},A_{1}^{[5]}  \}, \quad \{A_{1}^{[1]},A_{1}^{[2]},A_{1}^{[4]},A_{1}^{[5]}  \}, \\ \{A_{1}^{[1]},A_{1}^{[3]},A_{1}^{[4]},A_{1}^{[5]}  \}.
\end{gathered}
\end{equation}
and all subsets thereof are expressible. 

With these sets, we can derive a number of pentagon-incompatibility witnesses. For example, taking ${\bf V_{1}} = \{A_{1}^{[1]},...,A_{1}^{[5]}\}$, we get that $|{\bf V_{1}}| = 5$ and every element of $2^{{\bf V_{1}}}\setminus {\bf V_{1}}$ is injectable or expressible. We are thus able to derive a pentagon-incompatibility witness from $\Delta(\{\sigma_{\bf X}: {\bf X}\in 2^{\bf V_{1}}\!\setminus\!{\bf V_{1}}  \})$. Explicitly, if $\rho := \rho_{A^{[1]}A^{[2]}A^{[3]}A^{[4]}A^{[5]}}$ is a state on the composite Hilbert space associated to the visible nodes of $G$, then the witness, denoted $I_{{\bf V_{1}}}^{\textrm{pent}}$, is
\begin{align}
I_{{\bf V_{1}}}^{\textrm{pent}}(\rho)&:= \mathbb{1} - \rho_{A^{[1]}} - \rho_{A^{[2]}} - \rho_{A^{[3]}} - \rho_{A^{[4]}} - \rho_{A^{[5]}} \nonumber \\
&+ \rho_{A^{[1]}A^{[2]}} + \rho_{A^{[1]}}\otimes \rho_{A^{[3]}} + \rho_{A^{[1]}}\otimes \rho_{A^{[4]}}  \nonumber \\
& + \rho_{A^{[1]}}\otimes \rho_{A^{[5]}}+ \rho_{A^{[2]}A^{[3]}}  + \rho_{A^{[2]}}\otimes \rho_{A^{[4]}} + \rho_{A^{[2]}}\otimes \rho_{A^{[5]}} \nonumber \\
&+ \rho_{A^{[3]}A^{[4]}}  + \rho_{A^{[3]}}\otimes \rho_{A^{[5]}} + \rho_{A^{[4]}A^{[5]}} \nonumber \\
&- \rho_{A^{[1]}A^{[2]}A^{[3]}} - \rho_{A^{[1]}A^{[2]}}\otimes \rho_{A^{[4]}} - \rho_{A^{[1]}A^{[2]}}\otimes \rho_{A^{[5]}}  \nonumber \\
& - \rho_{A^{[1]}} \otimes \rho_{A^{[3]}A^{[4]}} - \rho_{A^{[1]}}\otimes \rho_{A^{[3]}} \otimes \rho_{A^{[5]}}  \nonumber \\
& - \rho_{A^{[1]}} \otimes \rho_{A^{[4]}A^{[5]}} - \rho_{A^{[2]}A^{[3]}A^{[4]}} - \rho_{A^{[2]}A^{[3]}}\otimes \rho_{A^{[5]}} \nonumber \\ & - \rho_{A^{[2]}} \otimes \rho_{A^{[4]}A^{[5]}}  - \rho_{A^{[3]}A^{[4]}A^{[5]}} + \rho_{A^{[1]}A^{[2]}A^{[3]}A^{[4]}}  \nonumber \\
&+ \rho_{A^{[1]}A^{[2]}A^{[3]}} \otimes \rho_{A^{[5]}} + \rho_{A^{[1]}A^{[2]}} \otimes \rho_{A^{[4]}A^{[5]}} \nonumber \\
&+ \rho_{A^{[1]}}\otimes \rho_{A^{[3]}A^{[4]}A^{[5]}} + \rho_{A^{[2]}A^{[3]}A^{[4]}A^{[5]}}.
\end{align}
Note that the above expression makes use of the causal structure of $G'$: every $\sigma_{{\bf X}}$ for which ${\bf X} \in 2^{\bf V_{1}}\!\setminus\!{\bf V_{1}}$ contains nodes that have independent ancestry in $G'$, is written as a tensor product of marginals of $\rho$. For example, the independence arising from the cut means that $\sigma_{A_{1}^{[1]}A_{1}^{[5]}} = \rho_{A^{[1]}} \otimes \rho_{A^{[5]}}$ just as with the triangle scenario, but also the independence of, e.g., the visible nodes $A_{1}^{[2]}$ and $A_{1}^{[4]}$ ensure that $\sigma_{A_{1}^{[2]}A_{1}^{[4]}} = \rho_{A^{[2]}} \otimes \rho_{A^{[4]}}$.

If, instead, we define ${\bf V_{2}} := \{A_{1}^{[1]}A_{1}^{[3]}, A_{1}^{[2]}A_{1}^{[5]}, A_{1}^{[4]} \}$, that is, by considering the partition $\{\mathcal{S}_{1}, \mathcal{S}_{2}, \mathcal{S}_{3}\}$ of $\textsf{Vnodes}(G')$ where $\mathcal{S}_{1} := \{ A_{1}^{[1]}, A_{1}^{[3]}\}$, $\mathcal{S}_{2} := \{A_{1}^{[2]}, A_{1}^{[5]}\}$ and $\mathcal{S}_{3} := \{ A_{1}^{[4]}\}$, we obtain an \textit{a priori} distinct pentagon-incompatibility witness:
\begin{align}
I_{{\bf V_{2}}}^{\textrm{pent}}(\rho) := &\mathbb{1} - \rho_{A^{[1]}} \otimes \rho_{A^{[3]}} - \rho_{A^{[2]}}\otimes \rho_{A^{[5]}} - \rho_{A^{[4]}} \nonumber \\
& + \rho_{A^{[1]}A^{[2]}A^{[3]}}\otimes \rho_{A^{[5]}} + \rho_{A^{[1]}} \otimes \rho_{A^{[3]}A^{[4]}} \nonumber \\
&+ \rho_{A^{[2]}}\otimes \rho_{A^{[4]}A^{[5]}}.
\end{align}
There are indeed states whose incompatibility with the pentagon scenario is witnessed by $I_{{\bf V_{1}}}^{\textrm{pent}}$ and $I_{{\bf V_{2}}}^{\textrm{pent}}$. Consider the $5$-qubit analogues to the GHZ and W states:
\begin{equation}
\begin{gathered}
\ket{\textrm{GHZ}_{5}} := \frac{1}{\sqrt{2}}(\ket{00000}+\ket{11111}),\\
\ket{\textrm{W}_{5}} := \frac{1}{\sqrt{5}}\sum_{\substack{i,j,k,l,m \in \{0,1\},\\ i+j+k+l+m=1}} \ket{ijklm}.
\end{gathered}
\end{equation}
Both $I_{{\bf V_{1}}}^{\textrm{pent}}$ and $I_{{\bf V_{2}}}^{\textrm{pent}}$ witness the pentagon-incompatibility of $\ket{\textrm{W}_{5}}$: $I_{{\bf V_{1}}}^{\textrm{pent}}(\ket{\textrm{W}_{5}}\!\!\bra{\textrm{W}_{5}})$ has an eigenvalue of $\approx -0.3942$ and $I_{{\bf V_{2}}}^{\textrm{pent}}(\ket{\textrm{W}_{5}}\!\!\bra{\textrm{W}_{5}})$ has an eigenvalue of $\approx -0.1000$. However, only $I_{{\bf V_{1}}}^{\textrm{pent}}$ witnesses the pentagon-incompatibility of $\ket{\textrm{GHZ}_{5}}$: $I_{{\bf V_{1}}}^{\textrm{pent}}(\ket{\textrm{GHZ}_{5}}\!\!\bra{\textrm{GHZ}_{5}})$ has an eigenvalue of $-0.125$ while the smallest eigenvalue of $I_{{\bf V_{2}}}^{\textrm{pent}}(\ket{\textrm{GHZ}_{5}}\!\!\bra{\textrm{GHZ}_{5}})$ is $0$. This demonstrates that both witnesses are non-trivial and different, however the question of whether $I_{{\bf V_{2}}}^{\textrm{pent}}(\rho) \ngeq 0$ implies that $I_{{\bf V_{1}}}^{\textrm{pent}}(\rho) \ngeq 0$ remains open.

There are still further ways to define ${\bf V}$ to be different choices of set partitions of $\textsf{Vnodes}(G')$, which may lead to other distinct pentagon-incompatibility witnesses. Furthermore, just as with the triangle scenario, it is possible to consider different Cut inflations, of which there are now $5$ in total. Between these two degrees of freedom, in choosing the Cut inflation and in choosing how to define ${\bf V}$, there are potentially many pentagon-incompatibility witnesses that can be obtained. Not all will be distinct as different choices of $G'$ and ${\bf V}$ can lead to the same witness, and moreover, not every witness is guaranteed to be non-redundant with respect to the other witnesses. A full characterization of the set of non-trivial and non-redundant pentagon-witnesses obtained via the method presented here is left for future work.

\subsection{Example: woven hexagon scenario}

In the pentagon scenario, each of the latent nodes has the same number of children as do the latent nodes in the triangle scenario, namely $2$. It is in fact possible to derive incompatibility witnesses for network scenarios exhibiting greater connectivity, as we now demonstrate by considering the woven hexagon scenario (\Cref{fig:pent_hex}(c)). In this scenario, every latent node has three children which impacts both the definition of what can be considered a Cut inflation as well as the properties of the corresponding injectable and expressible sets. 

The defining feature of the Cut inflations of both the triangle and pentagon scenarios is that the visible nodes between which the cut occurs, share no parents. This will be the guiding principle for defining the Cut inflations for the woven hexagon scenario here as well as the Cut inflations of the general family of networks considered below. 

Let $G$ denote the woven hexagon scenario, which has visible and latent nodes given by
\begin{equation}
\begin{gathered}
\textsf{Vnodes}(G) = \{A^{[1]}, A^{[2]},A^{[3]},A^{[4]},A^{[5]},A^{[6]}\}, \\
\textsf{Lnodes}(G) = \{L^{[1]}, L^{[2]},L^{[3]},L^{[4]},L^{[5]},L^{[6]} \}.
\end{gathered}
\end{equation}
Let $G'$ denote the Cut inflation of $G$ for the cut between $A_{1}^{[6]}$ and $A_{1}^{[1]}$, which has visible and latent nodes
\begin{equation}
\begin{gathered}
\textsf{Vnodes}(G') = \{A_{1}^{[1]}, A_{1}^{[2]},A_{1}^{[3]},A_{1}^{[4]},A_{1}^{[5]},A_{1}^{[6]}\}, \\
\textsf{Lnodes}(G') = \{L_{1}^{[1]}, L_{1}^{[2]},L_{1}^{[3]},L_{1}^{[4]},L_{1}^{[5]},L_{1}^{[6]},L_{2}^{[5]},L_{2}^{[6]} \}. 
\end{gathered}
\end{equation}
Similarly to the pentagon scenario, the connectivity of $G$ is given by: $L^{[j]}$ is a parent of $A^{[j]}, A^{[j+1]},A^{[j+2]}$ for each $j \in \{1,...,6\}$ with every index value greater than $6$ being taken modulo $6$. The connectivity of $G'$ is given by: $L_{1}^{[j]}$ is a parent of $A_{1}^{[j]}, A_{1}^{[j+1]}, A_{1}^{[j+2]}$ for all $j \in \{1,...,4\}$, $L_{1}^{[5]}$ is a parent of $A_{1}^{[5]}$ and $A_{1}^{[6]}$, $L_{1}^{[6]}$ is a parent only of $A_{1}^{[6]}$, $L_{2}^{[5]}$ is a parent only of $A_{1}^{[1]}$ and $L_{2}^{[6]}$ is a parent of $A_{1}^{[1]}$ and $A_{1}^{[2]}$. The sets
\begin{gather}
\{A_{1}^{[1]}, A_{1}^{[2]}, A_{1}^{[3]}, A_{1}^{[4]} \}, \quad \{A_{1}^{[2]}, A_{1}^{[3]}, A_{1}^{[4]}, A_{1}^{[5]} \} \\ 
\{A_{1}^{[3]}, A_{1}^{[4]}, A_{1}^{[5]}, A_{1}^{[6]} \},
\end{gather}
and all subsets thereof, are injectable in $G'$ (see \Cref{app:fam_networks} for details), while the sets
\begin{gather}
\{A_{1}^{[1]}, A_{1}^{[2]}, A_{1}^{[3]}, A_{1}^{[6]} \}, \quad \{A_{1}^{[1]}, A_{1}^{[2]}, A_{1}^{[5]}, A_{1}^{[6]} \} \\ 
\{A_{1}^{[1]}, A_{1}^{[4]}, A_{1}^{[5]}, A_{1}^{[6]} \},
\end{gather}
and all subsets thereof, are expressible.

Defining ${\bf V} := \{A_{1}^{[1]}A_{1}^{[6]}, A_{1}^{[2]}A_{1}^{[3]}, A_{1}^{[4]}A_{1}^{[5]}\}$, that is, with respect to the partition $\{\mathcal{S}_{1}, \mathcal{S}_{2}, \mathcal{S}_{3}\}$ of $\textsf{Vnodes}(G')$ where $\mathcal{S}_{1} := \{A_{1}^{[1]}, A_{1}^{[6]} \}$, $\mathcal{S}_{2} := \{A_{1}^{[2]},A_{1}^{[3]}\}$ and $\mathcal{S}_{3} := \{A_{1}^{[4]}, A_{1}^{[5]}\}$, we find that the set of contexts $2^{{\bf V}}\setminus {\bf V}$ is 
\begin{gather}
\left\{A_{1}^{[1]}A_{1}^{[2]}A_{1}^{[3]}A_{1}^{[6]},A_{1}^{[2]}A_{1}^{[3]}A_{1}^{[4]}A_{1}^{[5]}, A_{1}^{[1]}A_{1}^{[4]}A_{1}^{[5]}A_{1}^{[6]}, \right. \\
\left. A_{1}^{[1]}A_{1}^{[6]}, A_{1}^{[2]}A_{1}^{[3]}, A_{1}^{[4]}A_{1}^{[5]} \right\},
\end{gather}
all of which are injectable or expressible. Thus, we are in a position to derive a woven hexagon-incompatibility witness, namely, 
\begin{align}
I_{{\bf V}}^{\textrm{w-hex}}(\rho) &:= \mathbb{1} - \rho_{A^{[1]}} \otimes \rho_{A^{[6]}} - \rho_{A^{[2]}A^{[3]}} - \rho_{A^{[4]}A^{[5]}} \nonumber \\
&+ \rho_{A^{[1]}} \otimes \rho_{A^{[2]}A^{[3]}} \otimes \rho_{A^{[6]}} + \rho_{A^{[1]}} \otimes \rho_{A^{[4]}A^{[5]}A^{[6]}} \nonumber \\
&+ \rho_{A^{[2]}A^{[3]}A^{[4]}A^{[5]}}.
\end{align}
This is non-trivial as we now show. As with the pentagon scenario, let us consider generalizations of the GHZ and W states:
\begin{equation}
\begin{gathered}
\ket{\textrm{GHZ}_{6}} := \frac{1}{\sqrt{2}}(\ket{000000}+\ket{111111}),\\
\ket{\textrm{W}_{6}} := \frac{1}{\sqrt{6}}\sum_{\substack{i,j,k,l,m,o \in \{0,1\},\\ i+j+k+l+m+o=1}} \ket{ijklmo}.
\end{gathered}
\end{equation}
We then have that $I_{{\bf V}}^{\textrm{w-hex}}(\ket{\textrm{GHZ}_{6}}\!\!\bra{\textrm{GHZ}_{6}})$ has an eigenvalue of $-0.125$ and $I_{{\bf V}}^{\textrm{w-hex}}(\ket{\textrm{W}_{6}}\!\!\bra{\textrm{W}_{6}})$ has an eigenvalue of $\approx -0.1028$, thereby witnessing the incompatibility of these two states with the woven hexagon scenario.

Again, there are a couple of possible choices for ${\bf V}$ such as by taking ${\bf V} = \{A_{1}^{[1]}A_{1}^{[2]}, A_{1}^{[3]}A_{1}^{[4]},A_{1}^{[5]}A_{1}^{[6]}\}$ instead. 

Unlike the pentagon scenario, there is no possibility of defining ${\bf V}$ such that $|{\bf V}| = 5$: this would mean that some term in the expression $\Delta(\{\sigma_{\bf X}: {\bf X}\in 2^{\bf V}\!\setminus\!{\bf V}  \})$ would correspond to a state on $5$ visible nodes of $G'$, but $G'$ has no injectable or expressible sets of size greater than $4$.

\subsection{Incompatibility witnesses from Cut inflations of a family of networks} 

The above examples demonstrate that it is indeed possible to derive incompatibility witnesses for more networks than just the triangle scenario, many of which require a choice of ${\bf V}$ where the elements are composites of the visible nodes of the inflation DAG. Furthermore, the witnesses derived for all the examples considered so far --- the triangle, pentagon and woven hexagon scenarios --- have made use of Cut inflations. In this section, we demonstrate that this extends to a family of networks that include the triangle, pentagon and woven hexagon scenarios as special cases.

The members of the family of networks that we consider will be denoted by $G(n,l)$, where $n > l \geq 2$. The network $G(n,l)$ has visible and latent nodes given by
\begin{equation}
\begin{gathered}
\textsf{Vnodes}(G(n,l)) = \{A^{[1]},...,A^{[n]} \}, \\
\textsf{Lnodes}(G(n,l)) = \{L^{[1]},...,L^{[n]}\}
\end{gathered}
\end{equation}
and has connectivity given by: $L^{[j]}$ is a parent of $A^{[j]}, A^{[j+1]},...,A^{[j+l-1]}$ for all $j \in \{1,...,n\}$, where every index with value $i$ such that $i > n$ or $i \leq 0$ is considered modulo $n$.\footnote{The convention of considering index values greater than $n$ or less than or equal to $0$ to be taken modulo $n$ is a convention that will be applied throughout this section and in the corresponding appendix, \Cref{app:fam_networks}. This also applies to set notation such as $\{j-l+2,...,j\}$ if it corresponds to the possible values an index can take.} So defined, $G(n,l)$ has $n$ visible nodes and $n$ latent nodes, with each latent node having precisely $l$ children and each visible node having precisely $l$ parents. With this notation, $G(3,2)$ is the triangle scenario, $G(5,2)$ is the pentagon scenario, and $G(6,3)$ is the woven hexagon scenario. 

Cut inflations of members of this family will be denoted $G'(n,l,j)$, where $j \in \{1,...,n\}$ indicates the location of the cut. As stated during the treatment of the woven hexagon scenario above, the defining feature of a Cut inflation is that the visible nodes that appear on either side of the cut do not share any parents. Accordingly, if $G'(n,l,j)$ denotes a Cut inflation of $G$ where the cut is between $A_{1}^{[j]}$ and $A_{1}^{[j+1]}$, the visible and latent nodes of $G'(n,l,j)$ are
\begin{equation}
\begin{gathered}
\textsf{Vnodes}(G'(n,l,j)) = \{A_{1}^{[1]},...,A_{1}^{[n]} \}, \\
\textsf{Lnodes}(G'(n,l,j)) = \{L_{1}^{[1]},...,L_{1}^{[n]},L_{2}^{[n-j+2]},...,L_{2}^{[j]} \}.
\end{gathered}
\end{equation}
The connectivity of $G'(n,l,j)$ is given by: $L_{1}^{[i]}$ is a parent of $A_{1}^{[i]}, A_{1}^{[i+1]}, ..., A_{1}^{[i+l-1]}$ if $i \notin \{j-l+2, ..., j \}$, $L_{1}^{[i]}$ is a parent of $A_{1}^{[i]}, A_{1}^{[i+1]},...,A_{1}^{[j]}$ if $i \in \{j-l+2, ..., j\}$, and $L_{2}^{[i]}$ is a parent of $A_{1}^{[j+1]}, A_{1}^{[j+2]}, ..., A_{1}^{[i + l - 1]}$ for $i \in \{j-l+2, ..., j\}$. With this notation, the Cut inflations considered above for the pentagon and woven hexagon scenarios are $G(5,2,5)$ and $G(6,3,6)$ respectively.

We then have the following result:
\begin{Theorem} \label{thm:family_of_networks} Let $k \in \mathbb{N}_{> 1}$ be odd. If $n \geq (l-1)k$, then there exists a partition $\{\mathcal{S}_{1},...,\mathcal{S}_{k}\}$ of $\textsf{Vnodes}(G'(n,l,j))$ such that, taking ${\bf V} := \{\mathcal{S}_{1},...,\mathcal{S}_{k}\}$, all elements of $2^{{\bf V}}\setminus {\bf V}$ are injectable or expressible in $G'(n,l,j)$.
\end{Theorem}
The proof is given in \Cref{app:fam_networks} and consists in generalizing some of the common features of the proofs for the injectable and expressible sets for the Cut inflations of the triangle, pentagon and woven hexagon scenarios for the general case. This theorem demonstrates that incompatibility witnesses can be found using the method presented in this work for every network in the set $\{G(n,l)| \frac{n}{l-1} \geq 3 \}$. A thorough analysis of which of these witnesses prove fruitful for testing $G(n,l)$-incompatibility, as well as an investigation into which other causal structures admit similar witnesses, is left for future work.

\section{Discussion} \label{sec:dicussion}

Demonstrating the compatibility of a probability distribution or quantum state with a causal structure is in general a difficult task. In this article we have presented a method for witnessing the incompatibility of quantum states based on a family of quantum marginal constraints (in the form of the operator inequalities of Eq.~(\ref{eq:Hall_ineqs})) in conjunction with the inflation technique. We have furthermore presented an array of results that demonstrate that this method can indeed be used to make progress on problems of causal compatibility.

One of the key benefits of this approach pertains to its scalability. As demonstrated in, e.g., Ref~\cite{Hall_SDP_07}, quantum marginal problems admit a semidefinite programming formulation. Accordingly, when using the inflation technique for investigating questions of causal compatibility, one can always consider the corresponding SDP. However, in the cases where a witness can be obtained by a marginal inequality, such as those in Eq.~(\ref{eq:Hall_ineqs}), there is a slight advantage in terms of computational complexity as we now describe. Denoting the product of the dimensions of the subsystems involved by $D$, the complexity of the SDP associated to the quantum marginal problem is $\mathcal{O}(D^{\frac{13}{2}})$~\cite{Hall_SDP_07}. For example, for the marginal problems derived from a choice of inflation of a causal structure with qubit visible nodes, the SDP complexity would be $\mathcal{O}(2^{\frac{13n}{2}})$, where $n$ is the number of visible nodes in the inflation DAG. Alternatively, the complexity of finding the eigendecomposition of a $D \times D$ matrix in practice is typically stated as being $\mathcal{O}(D^{3})$ (improvements to the exponent exist in certain cases). Thus, producing an incompatibility witness corresponding to an inflation  DAG with $n$ qubit visible nodes allows us to perform an incompatibility test with complexity $\mathcal{O}(2^{3n})$ or better. This is clearly not an exponential speed-up over the SDP case, but it is a speed-up nonetheless.

As a consequence of this favourable scaling, we believe the method presented here has pragmatic value. As experimental realizations of quantum networks improve (see, e.g.,~\cite{Krutyanskiy_23,pompili2021realization,Liao_18}), it will be increasingly important to have practical methods for verifying the structure of such networks. For example, using the methods presented above it may be possible to verify that a tomographically-specified network state could not have been produced in a causal structure representing the desired network,  thereby  revealing a deficiency in the experimentalist's account of the set-up. 
  Furthermore, an understanding of whether it is possible to witness incompatibility via a probability distribution arising from local measurements allows for a better appreciation of the requirements for verifying a network structure.

There are also potential benefits of the approach presented here for classical causal inference. Since the quantum compatibility problem subsumes the classical compatibility problem (as every distribution can be expressed as a diagonal density operator), a test for compatibility of quantum states is also a test for compatibility of distributions. Sometimes the more abstract problem leads to progress on the more concrete problem, so it is conceivable that investigations of the sort described here and comparisons of the classical and fully quantum causal inference problems more generally, might ultimately yield dividends for classical causal inference.

The notion of a causal model can be generalized not just from classical theories to quantum, but to any generalized probabilistic theory (GPT) (see e.g.,~\cite{plavala2023general}). Boxworld~\cite{Barrett_07}, the toy theory of Ref.~\cite{Spekkens2007Evidence}, and real quantum theory\cite{caves2002unknown,wootters1990local} are prominent examples.  These theories are generally studies as foils to quantum theory, clarifying the meaning of the latter. As such, an interesting problem for future research is to study the causal compatibility problem within generalized probabilistic theories of interest and compare to the quantum case.   As emphasized in this article, if one wishes to derive such constraints using the inflation technique, then one must have some constraints coming from the marginal problem to leverage.  This, therefore, motivates a study of the marginal problem in various GPTs of interest.

\begin{acknowledgments}
We would like to thank Pedro Lauand, Phil LeMaitre, Alejandro Pozas-Kerstjens and Antonio Ac\'{i}n for helpful discussions, and Simon Milz and anonymous reviewers for pointing out several typos. We would also like to acknowledge early discussions on the problem of determining which quantum states are compatible with the triangle scenario with the participants of a 2019 Perimeter Scholars International winter school project led by RWS and EW, namely, R\'{e}mi Faure, Maria Julia Maristany, Tom\'{a}\v{s} Gonda, and T. C. Fraser.  Finally, special thanks to T. C. Fraser for pointing us to the work by Butterly {\em et al.} and Hall on the quantum marginal problem.
This research was funded in whole or in part by the Austrian Science Fund (FWF) through DK-ALM: W$1259$-N$27$. For open access purposes, the authors have applied a CC BY public copyright license to any author-accepted manuscript version arising from this submission. This research was supported in part by Perimeter Institute for Theoretical Physics. Research at Perimeter Institute is supported in part by the Government of Canada through the Department of Innovation, Science and Economic Development and by the Province of Ontario through the Ministry of Colleges and Universities.
\end{acknowledgments}

\bibliography{triangle_compat}

\newpage

\onecolumngrid
\appendix

\section{Proofs} \label{app:proofs} 

This appendix contains the proofs of the results presented in the main text.

\subsection{Proof of~\Cref{lem:xy_minus_x_y}} \label{app:proof_lem_xy_minus_x_y}

The proof of the following lemma uses properties of the max-divergence, some of which are provided below. The reader is referred to the relevant sections of, e.g.,~\cite{khatri2020principles,tomamichel2015quantum} for further details.

Let $\widetilde{D}_{\alpha}(\cdot \| \cdot)$ denote the sandwiched Rényi relative entropy where $\alpha \in (0, 1) \cup (1, \infty)$ \cite{wilde2014strong,muller2013quantum}. The relevance of this relative entropy for the purposes of~\Cref{lem:xy_minus_x_y} is via its connection to the max relative entropy \cite{datta2009min} defined via
\begin{align}
D_{\max}(\rho\|\sigma) = \log_{2} \inf\{\lambda \in \mathbb{R}: \rho \leq \lambda \sigma\}. \label{eq:d_max_def}
\end{align}
The connection between the sandwiched Rényi relative entropy and the max relative entropy is established through the following proposition:
\begin{Proposition}[\cite{khatri2020principles}] \label{prop:7.61} The sandwiched Rényi relative entropy converges to the max relative entropy in the limit $\alpha \rightarrow \infty$
\begin{align}
\lim_{\alpha \rightarrow \infty} \widetilde{D}_{\alpha}(\rho \| \sigma) = D_{\max}(\rho \| \sigma).
\end{align}
\end{Proposition}
As a consequence of the above proposition, the max relative entropy shares a number of properties with the sandwiched Rényi relative entropies, including those stated below:
\begin{Proposition}[\cite{khatri2020principles}] \label{prop:7.35} The sandwiched Rényi relative entropy $\widetilde{D}_{\alpha}$ satisfies the following properties for every states $\rho$ and positive semi-definite operator $\sigma$ for $\alpha \in [1/2, 1) \cup (1, \infty)$.
\begin{enumerate}
    \item \label{item:positivity} If $\Tr[\sigma] \leq \Tr[\rho]=1$, then $\widetilde{D}_{\alpha}(\rho \| \sigma) \geq 0$.
    \item \label{item:faithfulness} \textit{Faithfulness:} If $\Tr[\sigma] \leq 1$, we have that $\widetilde{D}_{\alpha}(\rho \| \sigma) = 0$ if and only if $\rho = \sigma$.
    \item \label{item:negativity} If $\rho \leq  \sigma$, then $\widetilde{D}_{\alpha}(\rho \| \sigma) \leq 0$.
    \item For every PSD operator $\sigma'$ such that $\sigma' \geq \sigma$, we have $\widetilde{D}_{\alpha}(\rho \| \sigma) \geq \widetilde{D}_{\alpha}(\rho \| \sigma')$.
\end{enumerate}
\end{Proposition}
The first three items above are of the most relevance for the proof the lemma:

\begin{proof}[Proof of~\Cref{lem:xy_minus_x_y}] Let $\rho_{xy}$ be such that $\rho_{xy} \neq \rho_{x} \otimes \rho_{y}$. Suppose for a contradiction that $\rho_{xy}  \leq \rho_{x} \otimes \rho_{y}$. From the consequences of~\Cref{prop:7.61}, we know that the max relative entropy satisfies~\Cref{item:positivity,item:faithfulness,item:negativity} from~\Cref{prop:7.35}. Since $\rho_{xy}$ and $\rho_{x} \otimes \rho_{y}$ are both quantum states and hence of unit trace, from~\Cref{item:positivity} we can conclude that $D_{\max}(\rho_{xy} \| \rho_{x} \otimes \rho_{y}) \geq 0$. From the assumption that $\rho_{x,y} \neq \rho_{x} \otimes \rho_{y}$, we can conclude from~\Cref{item:faithfulness} that $D_{\max}(\rho_{xy} \| \rho_{x} \otimes \rho_{y}) \neq 0$ and hence must be strictly positive. However, the assumption that $\rho_{xy} \leq \rho_{x} \otimes \rho_{y}$ and~\Cref{item:negativity} imply that $D_{\max}(\rho_{xy} \| \rho_{x} \otimes \rho_{y}) \leq 0$, forming the desired contradiction. The conclusion that $\rho_{xy} \leq \rho_{x} \otimes \rho_{y}$ implies that $D_{\max}(\rho_{xy} \| \rho_{x} \otimes \rho_{y}) \leq 0$ can be seen directly from~\Cref{eq:d_max_def} since $\rho_{xy} \leq \rho_{x} \otimes \rho_{y}$ is equivalent to $\rho_{xy} \leq \lambda \rho_{x} \otimes \rho_{y}$ for $\lambda = 1$. It follows that
\begin{align}
\inf\{\lambda \in \mathbb{R}: \rho_{xy} \leq \lambda \rho_{x} \otimes \rho_{y}\} \leq 1
\end{align}
and so taking the logarithm produces a non-positive value for $D_{\max}(\rho_{xy} \| \rho_{x} \otimes \rho_{y})$.
\end{proof}

\subsection{Distribution-witnessability of W-state} \label{app:distn_witn_W_state}

Here we demonstrate that, unlike the state $\rho^{({\rm Wdistn})}$ that encodes the $W$-distribution, the triangle-incompatibility of the state $\rho^{({\rm W})}:=\ket{{\rm W}}\!\!\bra{{\rm W}}$ with $\ket{{\rm W}} = \frac{1}{\sqrt{3}}(\ket{001}+\ket{010} +\ket{100})$ is distribution-witnessable using the cut-inflation. 

Recall that we can write the triangle-incompatibility witness $I_{AB}$ as 
\begin{align}
I_{AB}(\rho^{({\rm W})}) = \Delta(\{\rho_{{\bm X}}^{({\rm W})} : {\bm X} \in 2^{{\bm V}}\setminus {\bm V} \}) + (\rho_{A}^{({\rm W})} \otimes \rho_{B}^{({\rm W})} - \rho_{AB}^{({\rm W})})\otimes I_{C}
\end{align}
where ${\bm V} = \{A,B,C\}$. We have that
\begin{align}
\Delta(\{\rho_{{\bm X}}^{({\rm W})} : {\bm X} \in 2^{{\bm V}}\setminus {\bm V} \}) = \begin{bmatrix}
0 & 0 & 0 & 0 & 0 & 0 & 0 & 0 \\
 0 & \frac{1}{3} & \frac{1}{3} & 0 & \frac{1}{3} & 0 & 0 & 0 \\
 0 & \frac{1}{3} & \frac{1}{3} & 0 & \frac{1}{3} & 0 & 0 & 0 \\
 0 & 0 & 0 & \frac{1}{3} & 0 & \frac{1}{3} & \frac{1}{3} & 0 \\
 0 & \frac{1}{3} & \frac{1}{3} & 0 & \frac{1}{3} & 0 & 0 & 0 \\
 0 & 0 & 0 & \frac{1}{3} & 0 & \frac{1}{3} & \frac{1}{3} & 0 \\
 0 & 0 & 0 & \frac{1}{3} & 0 & \frac{1}{3} & \frac{1}{3} & 0 \\
 0 & 0 & 0 & 0 & 0 & 0 & 0 & 0 
\end{bmatrix}
\end{align}
and 
\begin{align}
\rho_{A}^{({\rm W})} \otimes \rho_{B}^{({\rm W})} - \rho_{AB}^{({\rm W})} = \begin{bmatrix}
\frac{1}{9} & 0 & 0 & 0 \\
 0 & -\frac{1}{9} & -\frac{1}{3} & 0 \\
 0 & -\frac{1}{3} & -\frac{1}{9} & 0 \\
 0 & 0 & 0 & \frac{1}{9}
\end{bmatrix}.
\end{align}
Consequently, we have that
\begin{align}
\bra{++-}I_{AB}(\rho^{({\rm W})})\ket{++-} &= \bra{++-}\Delta(\{\rho_{{\bm X}}^{({\rm W})} : {\bm X} \in 2^{{\bm V}}\setminus {\bm V} \})\ket{++-} + \bra{++}\rho_{A}^{({\rm W})} \otimes \rho_{B}^{({\rm W})} - \rho_{AB}^{({\rm W})}\ket{++} \\
&= \frac{1}{12} - \frac{1}{6}
\end{align}
and so the probability distribution that results from measuring each subsystem of $\rho^{(W)}$ in the $X$-basis is triangle-incompatible.

\subsection{Proof of~\Cref{thm:main}} \label{app:proof_cor_pure_qubit_incompat}

Prior to proving the theorem, it is convenient to show the following lemma:

\begin{Lemma} \label{lem:lin_ind} If $\ket{\omega_{AB}} \in \H_{A} \otimes \H_{B}$, $\ket{\omega_{AC}} \in \H_{A} \otimes \H_{C}$ and $\ket{\omega_{BC}} \in \H_{B} \otimes \H_{C}$ are non-zero, then the sets
\begin{align}
\{ \ket{\omega_{AB}}\otimes \ket{0}_{C}, \ket{\omega_{AB}} \otimes \ket{1}_{C}, \ket{\omega_{AC}} \otimes \ket{0}_{B}, \ket{\omega_{AC}} \otimes \ket{1}_{B} \}, \label{eq:AB_AC_lin_ind} \\
\{ \ket{\omega_{AB}} \otimes \ket{0}_{C}, \ket{\omega_{AB}} \otimes \ket{1}_{C}, \ket{\omega_{BC}} \otimes \ket{0}_{A}, \ket{\omega_{BC}} \otimes \ket{1}_{A} \}, \\
\{ \ket{\omega_{AC}} \otimes \ket{0}_{B}, \ket{\omega_{AC}} \otimes \ket{1}_{B}, \ket{\omega_{BC}} \otimes \ket{0}_{A}, \ket{\omega_{BC}} \otimes \ket{1}_{A} \},
\end{align}
each contain at least three linearly independent vectors.
\end{Lemma}

\begin{proof} We demonstrate the proof for~\Cref{eq:AB_AC_lin_ind}; the proof for the other sets are analogous. Let $\ket{\omega_{AB}} = [a,b,c,d]^{\top}$ and $\ket{\omega_{AC}} = [a',b',c',d']^{\top}$ where at least one of $a,b,c,d$ and at least one of $a',b',c',d'$ are non-zero by assumption. Suppose that
\begin{align}
\ket{\omega_{AC}} \otimes \ket{0}_{B} &= \alpha \ket{\omega_{AB}} \otimes \ket{0}_{C} + \beta \ket{\omega_{AB}} \otimes \ket{1}_{C}, \\
\ket{\omega_{AC}} \otimes \ket{1}_{B} &= \gamma \ket{\omega_{AB}} \otimes \ket{0}_{C} + \delta \ket{\omega_{AB}} \otimes \ket{1}_{C}, 
\end{align}
where neither $\alpha$ and $\beta$ nor $\gamma$ and $\delta$ are both-zero. The first equation requires that $b$ and $d$ both be $0$, while the second equation requires that both $a$ and $c$ be $0$, forming a contradiction. It follows that at least one of $\ket{\omega_{AC}} \otimes \ket{0}_{B}$ and $\ket{\omega_{AC}} \otimes \ket{1}_{B}$ is linearly independent of $\ket{\omega_{AB}} \otimes \ket{0}_{C}$ and $\ket{\omega_{AB}} \otimes \ket{1}_{C}$. 
\end{proof}

Let ${\bf V} := \{A,B,C\}$. A further useful fact for the proof of~\Cref{thm:main} arises from the proof of Theorem $3.2$ in~\cite{butterley2006compatibility}: for $\H_{A} \cong \H_{B} \cong \H_{C} \cong \mathbb{C}^{2}$, we have that
\begin{align}
\Delta(\{\rho_{{\bf X}}: {\bf X} \in 2^{{\bf V}}\!\setminus\!{\bf V} \}) = \rho + \tau^{-1}\rho \tau \label{eq:Delta_rho_taurho}
\end{align}
where $\tau$ is the three qubit antiunitary defined by its action on pure states via
\begin{align}
\tau \sum_{i,j,k}\alpha_{i,j,k}\ket{ijk} = \sum_{i,j,k}(-1)^{i + j + k}\alpha_{ijk}^{*}\ket{\overline{i}\overline{j}\overline{k}}
\end{align}
with $i,j,k \in \{0,1\}$ and $\overline{x} := 1 - x$.

\begin{proof}[Proof of~\Cref{thm:main}] Suppose $\rho_{ABC} = \ket{\Psi}\!\!\bra{\Psi}$ where $\ket{\Psi}$ is genuinely tripartite entangled, meaning that each of the two-party marginals satisfies $\rho_{xy} \neq \rho_{x} \otimes \rho_{y}$. As a consequence of~\Cref{lem:xy_minus_x_y}, $\rho_{A} \otimes \rho_{B} - \rho_{AB}$, $\rho_{A} \otimes \rho_{C} - \rho_{AC}$ and  $\rho_{B} \otimes \rho_{C} - \rho_{BC}$ are not positive semi-definite and hence have a eigenvector with a negative eigenvalue. Let $\omega_{xy} \in \H_{x} \otimes \H_{y}$ denote an eigenvector of $\rho_{x}\otimes \rho_{y} - \rho_{xy}$ with negative eigenvalue, for $x\neq y \in \{A,B,C\}$. As a consequence of~\Cref{lem:lin_ind}, we have that
\begin{align}
\dim(\supp(\nu_{AB}^{-}\otimes  \mathbb{1}_{C}) \cup \supp(\nu_{AC}^{-}\otimes \mathbb{1}_{B} ) \cup \supp(\nu_{BC}^{-}\otimes  \mathbb{1}_{A})) \geq 3.
\end{align}
Noting that $\tau^{-1}\ket{\Psi}$ is orthogonal to $\ket{\Psi}$, it follows from~\Cref{eq:Delta_rho_taurho} that $\dim(\supp(\Delta(\{\rho_{{\bf X}}: {\bf X} \in 2^{{\bf V}}\!\setminus\!{\bf V} \}) )) = 2$ (i.e., for any orthonormal basis of $\mathcal{H}_{A} \otimes \mathcal{H}_{B} \otimes \mathcal{H}_{C}$ containing $\ket{\Psi}$ and $\tau^{-1}\ket{\Psi}$, the other $6$ basis vectors are in the kernel of $\Delta(\{\rho_{{\bf X}}: {\bf X} \in 2^{{\bf V}}\!\setminus\!{\bf V} \})$). It follows that there exists a $\ket{\Phi} \in \supp(\nu_{xy}^{-}\otimes \mathbb{1}_{z}) \cap \ker( \Delta(\{\rho_{{\bf X}}: {\bf X} \in 2^{{\bf V}}\!\setminus\!{\bf V} \}))$ for some $x,y,z \in \{A,B,C\}$ with $x\neq y \neq z \neq x$, and so the result follows from~\Cref{prop:supp_ker_intersect}. 
\end{proof}

\subsection{Proof of~\Cref{prop:state_vs_distr}} \label{app:prop_state_vs_distr}

\Cref{prop:state_vs_distr} presented in the main text establishes that there exist states whose triangle-incompatibility is not distribution-witnessable. The proof of this proposition makes use of a one-parameter family of pure states:
\begin{align}
\ket{\Psi(t)} = \frac{\sqrt{t-2}}{\sqrt{t}} \ket{100}  + \frac{1}{\sqrt{t}}\ket{001} + \frac{1}{\sqrt{t}} \ket{010},
\end{align}
where $t$ is a real parameter in the interval $[3,\infty)$. These are categorized as `tri-Bell' states according to the classification of Ref.~\cite{acin_et_al_00}.  Note that the boundary of this family, at $t=3$,  is the W state of Eq.~\eqref{eq:Wstate}. We will use the notation $\rho_{ABC}(t)$ for the member of this family for a given $t$. 

Since every element of this family of states (for finite $t$) fails to factorize across any bipartition, by Theorem~\ref{thm:main}, they are all incompatible with the triangle scenario.  Furthermore, this incompatibility can be demonstrated using the triangle-incompatibility witness $I_{AB}$, as we will now see.

For $\rho_{ABC}(t)$, the eigenvalues of $I_{AB}(\rho_{ABC}(t))$ are
\begin{align}
\left\{\frac{1}{t^{2}}(t-2), \frac{r_{1}}{t^{2}}, \frac{r_{2}}{t^{2}}, \frac{r_{3}}{t^{2}} \right\}
\end{align}
each with multiplicity $2$, where $r_{1}, r_{2}, r_{3}$ are the roots of the cubic polynomial
\begin{align}
f(x) = x^{3} &+ (-t^{2} + t - 2)x^{2} + (-4 + 8t - 5t^{2} + t^{3})x + (8 - 20t + 14t^{2} - 5t^{3} + t^{4}). 
\end{align}
For a given $t \geq 3$, the demonstration of a negative value for $r_{1}$, $r_{2}$ or $r_{3}$ suffices to prove that $\rho_{ABC}(t)$ is triangle-incompatible. Since $I_{AB}(\rho_{ABC}(t))$ is Hermitian for all $t$, it must be that $r_{1}, r_{2}, r_{3} \in \mathbb{R}$. Via Vieta's formula~\cite{girard1884invention} (see also~\cite{Vieta}), we get that
\begin{align}
r_{1}r_{2}r_{3} = -(8 - 20t + 14t^{2} - 5t^{3} + t^{4}).
\end{align}
As the quartic polynomial inside the brackets on the right-hand side is strictly positive for $t > 2$, it follows that $r_{1}r_{2}r_{3}$ is strictly negative, indicating at least one negative root. Thus $\rho_{ABC}(t)$ is triangle-incompatible for the range of values of $t$ that we consider (i.e. $t \geq 3$). Any normalized eigenvector corresponding to a negative eigenvalue is thus a state that witnesses this incompatibility.

The fact that there exist values of $t$ such that incompatibility of $\rho_{ABC}(t)$ is not distribution-witnessable was presented in the main text. To show that $\rho_{ABC}(t)$ for certain $t$ is not distribution-witnessable using the AB-Cut inflation is equivalent to demonstrating that the quantity
\begin{align}
\iota_{AB}(\rho_{ABC}(t)) &:= \min_{\ket{\phi_{A}},\ket{\phi_{B}}, \ket{\phi_{C}}} \bra{\phi_{A}\phi_{B}\phi_{C}}I_{AB}(\rho_{ABC}(t))\ket{\phi_{A}\phi_{B}\phi_{C}} 
\end{align}
is non-negative. To do so, we consider a relaxation of the above non-convex optimization to a convex one, namely 
\begin{align}
\tilde{\iota}_{AB}(\rho_{ABC}(t)) &:= \min_{\varrho_{ABC} \in S} \Tr[\varrho_{ABC} I_{AB}(\rho_{ABC}(t))], \label{eq:relaxed_AB} 
\end{align}
where $S$ denotes the set of all three-qubit states with positive partial transposes over each subsystem, that is:
\begin{align}
S := \{ \varrho_{ABC} : \varrho_{ABC}^{\top_{A}} \geq 0, \varrho_{ABC}^{\top_{B}} \geq 0, \varrho_{ABC}^{\top_{C}} \geq 0 \}.
\end{align}
This set $S$ contains the set of pure product states, meaning that $\iota_{ABC}(t)$ is guaranteed to be lower bounded by $\tilde{\iota}_{ABC}(t)$. As discussed in the main text, the latter quantity is indeed non-negative for e.g., $t \geq 7$.

\subsection{Proof of~\Cref{prop:mixing_value}} \label{app:proof_prop_mixing_value}

\begin{proof}[Proof of~\Cref{prop:mixing_value}]
Let us write $\rho_{ABC} = \ket{\Psi}\!\!\bra{\Psi}$ so $\hat{\rho}(\ket{\Psi}, p) = p \rho_{ABC} + \frac{1-p}{8}\mathbb{1}_{ABC}$. For convenience, we will at times write $\hat{\rho}$ for $\hat{\rho}(\ket{\Psi}, p)$. By linearity, we see that 
\begin{align}
 \Delta(\{\hat{\rho}_{{\bf X}}: {\bf X} \in 2^{{\bf V}}\!\setminus\!{\bf V} \}) = p\Delta(\{\rho_{{\bf X}}: {\bf X} \in 2^{{\bf V}}\!\setminus\!{\bf V} \}) + (1-p)\Delta\left(\left\{\frac{\mathbb{1}_{{\bf X}}}{2^{|{\bf X}|}}: {\bf X} \in 2^{{\bf V}}\!\setminus\!{\bf V}\right\}\right)
\end{align}
Furthermore, recalling~\Cref{eq:Hall_tripartite}, we have that  $\Delta\left(\left\{\frac{\mathbb{1}_{{\bf X}}}{2^{|{\bf X}|}}: {\bf X} \in 2^{{\bf V}}\!\setminus\!{\bf V}\right\}\right) = \frac{1}{4}\mathbb{1}_{ABC}$. We demonstrate that there exists a $q \in [3 - 2\sqrt{2}, 1)$ such that $I_{AB}(\hat{\rho}(\ket{\Psi}, p)) \geq 0$ for all $0 \leq p \leq q$; the proofs for the other triangle-incompatibility witnesses proceed analogously, and the result is obtained by taking the minimal value of $q$ for the three cases.

Since 
\begin{align}
\hat{\rho}_{A} \otimes \hat{\rho}_{B} &= (p \rho_{A} + \frac{1-p}{2} \mathbb{1}_{A} ) \otimes (p \rho_{B} + \frac{1-p}{2}  \mathbb{1}_{B}) \\
&= p^{2} \rho_{A} \otimes \rho_{B} + \frac{p(1-p)}{2} \rho_{A} \otimes  \mathbb{1}_{B} + \frac{p(1-p)}{2}  \mathbb{1}_{A}  \otimes \rho_{B} + \frac{(1-p)^{2}}{4}  \mathbb{1}_{AB} , \\
\hat{\rho}_{AB} &= p \rho_{AB} + \frac{(1-p)}{4} \mathbb{1}_{AB} ,
\end{align}
and $ I_{AB}(\hat{\rho}) = \Delta(\{\hat{\rho}_{{\bf X}}: {\bf X} \in 2^{{\bf V}}\!\setminus\!{\bf V} \}) + \hat{\rho}_{A} \otimes \hat{\rho}_{B} - \hat{\rho}_{AB}$, we get that
\begin{align}
 I_{AB}(\hat{\rho}) &= p \Delta(\{\rho_{{\bf X}}: {\bf X} \in 2^{{\bf V}}\!\setminus\!{\bf V} \}) + \frac{1-p}{4} \mathbb{1}_{ABC} + p^{2} \rho_{A} \otimes \rho_{B} \otimes  \mathbb{1}_{C} + \frac{p(1-p)}{2} \rho_{A} \otimes \mathbb{1}_{BC}   \nonumber \\
& \quad + \frac{p(1-p)}{2}\mathbb{1}_{AC}\otimes \rho_{B} + \frac{(1-p)^{2}}{4} \mathbb{1}_{ABC} - p \rho_{AB} \otimes \mathbb{1}_{C} -  \frac{(1-p)}{4}  \mathbb{1}_{ABC} \\
&= p  \Delta(\{\rho_{{\bf X}}: {\bf X} \in 2^{{\bf V}}\!\setminus\!{\bf V} \}) + p^{2} \rho_{A} \otimes \rho_{B} \otimes  \mathbb{1}_{C} + \frac{p(1-p)}{2} \rho_{A} \otimes  \mathbb{1}_{BC} + \frac{p(1-p)}{2}  \mathbb{1}_{AC} \otimes \rho_{B} \nonumber \\
& \quad + \frac{(1-p)^{2}}{4} \mathbb{1}_{ABC} - p \rho_{AB} \otimes  \mathbb{1}_{C}.
\end{align}
Noting that a sufficient condition for $A - B  \geq  0$ for $A$ and $B$ positive semidefinite matrices is that the maximal eigenvalue of $B$ is less than or equal to the minimal eigenvalue of $A$, we see that $\frac{(1-p)^{2}}{4} \mathbb{1}_{ABC} - p \rho_{AB} \otimes \mathbb{1}_{C}  \geq  0$ is guaranteed at least when $\frac{1}{4}(1-p)^{2} - p \geq 0$, i.e. for $0 \leq p \leq 3 - 2\sqrt{2}$. It follows that $I_{AB}(\hat{\rho})$ is guaranteed to be positive semidefinite at least for $0 \leq p \leq q := 3 - 2\sqrt{2}$.  The same reasoning holds for $I_{AC}$ and $I_{BC}$ with the same value of $q$.
\end{proof}

The above proof establishes a lower bound for $q$ that is generally not tight for a given pure state $\ket{\Psi}$. For $\hat{\rho}(\ket{GHZ},p)$, the eigenvalues of $I_{xy}(\hat{\rho})$ are the same for all $xy \in \{AB, AC, BC\}$ and are given by $\frac{1}{4}\{1, 1 \pm 2p \}$, where the first has multiplicity $4$ and the others multiplicity $2$. Thus, for $p \leq \frac{1}{2}$, the $I_{xy}(\hat{\rho})$ are all positive.

For $\hat{\rho}(\ket{W},p)$, the eigenvalues of $I_{xy}(\hat{\rho})$ are the same for all $xy \in \{AB, AC, BC\}$ and are given by 
\begin{align}
\frac{1}{36}\left\{9-p^{2}, 9 -6p + p^{2}, 9 + 3p \pm \sqrt{297p^{2} + 6p^{3} + p^{4}} \right\},
\end{align}
each with multiplicity $2$. These eigenvalues are all positive whenever $9 + 3p \geq \sqrt{297p^{2} + 6p^{3} + p^{4}}$, which occurs for $p \lessapprox 0.627$.

\subsection{Proof of~\Cref{prop:Toth_Acin_states}} \label{app:proof_prop_Toth_Acin_states}

Any three-qubit mixed state can be written as \cite{Sudbery_2001}
\begin{align}
\rho_{ABC} = \frac{1}{8}\mathbb{1}_{ABC} + \sum_{i,j,k = X,Y,Z} a_{i}\sigma_{A}^{i} + b_{j} \sigma_{B}^{j} + c_{k}\sigma_{C}^{k} + d_{ij}\sigma_{A}^{i}\otimes \sigma_{B}^{j} + e_{ik}\sigma_{A}^{i}\otimes \sigma_{C}^{k} + f_{jk}\sigma_{B}^{j}\otimes \sigma_{C}^{k} + g_{ijk}\sigma_{A}^{i} \otimes \sigma_{B}^{j} \otimes \sigma_{C}^{k},
\end{align}
where the $\sigma^{l}$ for $l = X,Y,Z$ denote the Pauli operators and where all the coefficients are real (recall that the Pauli matrices form an operator basis for the real vector space of Hermitian operators). Each term in the above expression is to be understood as an operator $\H_{A} \otimes \H_{B} \otimes \H_{C}$ - identity operators have been omitted for brevity. 

Substituting the expression for $\rho_{ABC}$ into~\Cref{quantumABCutBellWignerInequality} and the corresponding equations for the $AC$ and $BC$ triangle-incompatibility witnesses, we get that
\begin{align}
 I_{AB}(\rho) &= \frac{1}{4} \mathbb{1}_{ABC} + \sum_{i,j,k=X,Y,Z} 16a_{i}b_{j} \sigma_{A}^{i} \otimes \sigma_{B}^{j} + 2e_{ik} \sigma_{A}^{i} \otimes \sigma_{C}^{k} + 2f_{jk}\sigma_{B}^{j} \otimes \sigma_{C}^{k} \\
 I_{AC}(\rho) &= \frac{1}{4} \mathbb{1}_{ABC} + \sum_{i,j,k=X,Y,Z} 2d_{ij} \sigma_{A}^{i} \otimes \sigma_{B}^{j} + 16a_{i}c_{k} \sigma_{A}^{i} \otimes \sigma_{C}^{k} + 2f_{jk}\sigma_{B}^{j} \otimes \sigma_{C}^{k}\\
 I_{BC}(\rho)  &= \frac{1}{4} \mathbb{1}_{ABC} + \sum_{i,j,k=X,Y,Z} 2d_{ij} \sigma_{A}^{i} \otimes \sigma_{B}^{j} + 2e_{ik} \sigma_{A}^{i} \otimes \sigma_{C}^{k} + 16b_{j}c_{k}\sigma_{B}^{j} \otimes \sigma_{C}^{k}.
\end{align}
For the family of states $\rho^{(3,c)}$, we see that $a_{i} = b_{j} = c_{k} = 0$ for all $i,j,k$ and that $d_{ij} = -\frac{c}{16} \delta_{ij}$, $e_{ik} = -\frac{c}{16}\delta_{ik}$ and $f_{jk} = \frac{1}{24}\delta_{jk}$. Thus,
\begin{align}
I_{AB}(\rho^{(3,c)}) &= \frac{1}{4} \mathbb{1}_{ABC}  + \sum_{k=X,Y,Z} \frac{1}{12}\sigma_{B}^{k}\otimes \sigma_{C}^{k} - \frac{c}{8}\sigma_{A}^{i}\otimes \sigma_{C}^{k} \\
&= \begin{bmatrix}
\frac{1}{3}-\frac{c}{8} & 0 & 0 & 0 & 0 & 0 & 0 & 0 \\
 0 & \frac{c}{8}+\frac{1}{6} & \frac{1}{6} & 0 & -\frac{c}{4}
   & 0 & 0 & 0 \\
 0 & \frac{1}{6} & \frac{1}{6}-\frac{c}{8} & 0 & 0 & 0 & 0 &
   0 \\
 0 & 0 & 0 & \frac{c}{8}+\frac{1}{3} & 0 & 0 & -\frac{c}{4} &
   0 \\
 0 & -\frac{c}{4} & 0 & 0 & \frac{c}{8}+\frac{1}{3} & 0 & 0 &
   0 \\
 0 & 0 & 0 & 0 & 0 & \frac{1}{6}-\frac{c}{8} & \frac{1}{6} &
   0 \\
 0 & 0 & 0 & -\frac{c}{4} & 0 & \frac{1}{6} &
   \frac{c}{8}+\frac{1}{6} & 0 \\
 0 & 0 & 0 & 0 & 0 & 0 & 0 & \frac{1}{3}-\frac{c}{8}
\end{bmatrix}.
\end{align}
This operator has eigenvalues $\frac{1}{24}(8-3c)$ and $\frac{1}{24}(4 + 3c \pm 2\sqrt{4 + 6c + 9c^{2}})$, the first with multiplicity $4$ and the others each with multiplicity $2$. It is readily verified that $4 + 3c - 2\sqrt{4 + 6c + 9c^{2}}$ is strictly negative for $c \neq 0$, establishing the result.

\subsection{Proof of~\Cref{prop:GHZ_224}} \label{app:proof_prop_GHZ_224}

For $\rho_{ABC} = \ket{\Psi}\!\!\bra{\Psi}$ where
\begin{align}
\ket{\Psi} = \alpha_{0}e^{i\phi_{0}}\ket{000} + \alpha_{4}\ket{110} + \alpha_{5}\ket{111} + \alpha_{6}\ket{112} +\alpha_{7}\ket{113}
\end{align}
we have that
\begin{align}
I_{AC}(\rho_{ABC}) = \begin{bmatrix} 
\Gamma_{1}& \boldsymbol{0} & \boldsymbol{0} & e^{i\phi_{0}}\alpha_{0}\alpha_{4} \mathbb{1} \\
\boldsymbol{0}& \Gamma_{2} & \boldsymbol{0} & \boldsymbol{0} \\
\boldsymbol{0}& \boldsymbol{0} & -\Gamma_{2} & \boldsymbol{0}\\
e^{-i\phi_{0}}\alpha_{0}\alpha_{4}  \mathbb{1}  & \boldsymbol{0} & \boldsymbol{0} & \mathbb{1} -\Gamma_{1}
\end{bmatrix}
\end{align}
where here  $\mathbb{1}$ is the $4\times 4$ identity matrix and  
\begin{gather}
\Gamma_{1} := \begin{bmatrix}
\alpha _0^4+\left(\alpha _4^2-1\right) \alpha _0^2-\alpha _4^2+1 & \left(\alpha _0^2-1\right) \alpha _4 \alpha _5 & \left(\alpha _0^2-1\right) \alpha _4 \alpha _6 & \left(\alpha _0^2-1\right) \alpha _4 \alpha_{7} \\
\left(\alpha _0^2-1\right) \alpha _4 \alpha _5 & \left(\alpha _0^2-1\right) \left(\alpha_5^2-1\right) & \left(\alpha _0^2-1\right) \alpha _5 \alpha _6 & \left(\alpha _0^2-1\right)\alpha _5 \alpha_{7} \\
\left(\alpha _0^2-1\right) \alpha _4 \alpha _6 & \left(\alpha _0^2-1\right) \alpha _5 \alpha _6 &\left(\alpha _0^2-1\right) \left(\alpha _6^2-1\right) & \left(\alpha _0^2-1\right) \alpha _6\alpha_{7} \\
\left(\alpha _0^2-1\right) \alpha _4 \alpha_{7} & \left(\alpha _0^2-1\right) \alpha _5 \alpha_{7} & \left(\alpha _0^2-1\right) \alpha _6 \alpha_{7} & \left(1-\alpha _0^2\right) \left(1 - \alpha_{7}^{2}\right) \\
\end{bmatrix} \\
\Gamma_{2} := \begin{bmatrix}
-\alpha _0^2 \left(1-\alpha _0^2-\alpha _4^2\right) & \alpha _0^2 \alpha _4 \alpha _5 & \alpha _0^2 \alpha _4 \alpha _6 & \alpha _0^2 \alpha _4 \alpha_{7} \\
\alpha _0^2 \alpha _4 \alpha _5 & \alpha _0^2 \alpha _5^2 & \alpha _0^2 \alpha _5 \alpha _6 & \alpha _0^2 \alpha _5 \alpha_{7} \\
\alpha _0^2 \alpha _4 \alpha _6 & \alpha _0^2 \alpha _5 \alpha _6 & \alpha _0^2 \alpha _6^2 & \alpha _0^2 \alpha _6 \alpha_{7} \\
\alpha _0^2 \alpha _4 \alpha_{7} & \alpha _0^2 \alpha _5 \alpha_{7} & \alpha _0^2 \alpha _6 \alpha_{7} & \alpha _0^2 \alpha_{7}^{2} \\
\end{bmatrix}.
\end{gather}
By inspection, one notices that 
\begin{align}
\braket{010| I_{AC}(\rho_{ABC}) |010} = -\alpha _0^2 \left(1-\alpha _0^2-\alpha _4^2\right)
\end{align}
which is negative whenever $0< \alpha_{0}^{2} + \alpha_{4}^{2} < 1$.

\subsection{Proof of~\Cref{thm:family_of_networks}} \label{app:fam_networks}

In this section, we prove \Cref{thm:family_of_networks} regarding Cut inflations $G'(n,l,j)$ of the networks $G(n,l)$ where $n > l \ge 2$. Recall that $G(n,l)$ has visible and latent nodes 
\begin{equation}
\begin{gathered}
\textsf{Vnodes}(G(n,l)) = \{A^{[1]},...,A^{[n]}\}, \\ \textsf{Lnodes}(G(n,l)) = \{L^{[1]},...,L^{[n]}\},
\end{gathered}
\end{equation}
while $G'(n,l,j)$ has visible and latent nodes
\begin{equation}
\begin{gathered}
\textsf{Vnodes}(G'(n,l,j)) = \{A_{1}^{[1]},...,A_{1}^{[n]} \}, \\
\textsf{Vnodes}(G'(n,l,j)) = \{L_{1}^{[1]},...,L_{1}^{[n]},L_{2}^{[j-l+2]},L_{2}^{[j-l+3]},...,L_{2}^{[j]} \}.
\end{gathered}
\end{equation}
With this choice of notation, along with the fact that every index value greater than $n$ or less than or equal to $0$ is taken modulo $n$, it is easy to see that any proof of involving injectable and expressible set of $G'(n,l,j)$ can quickly become heavily entrenched in keeping track of the sub- and superscript labelling of nodes. To ease the notational burden, we introduce some further terminology and notation, and make some simplifying choices.

The main simplifying choice is that we will only consider the Cut inflation $G'(n,l,n)$, that is, the Cut inflation with the cut between the visible node $A_{1}^{[n]}$ and $A_{1}^{[1]}$ from now on. The proof for any other choice of cut proceeds analogously: by making the substitution $[t] \mapsto [t - j]$ for every superscript (with the usual convention of index values modulo $n$), the case of the Cut inflation $G'(n,l,j)$ is mapped to the case $G'(n,l,n)$ that we consider here.

Next, it will be convenient to consider the set $\textsf{Lnodes}(G'(n,l,n))$ equipped with the total order given by:
\begin{align}
L_{2}^{[n-l+2]} < L_{2}^{[n-l+3]} < ... < L_{2}^{[n]} < L_{1}^{[1]} < L_{1}^{[2]} < ... < L_{1}^{[n]}. \label{eq:Lnodes_ordering}
\end{align}
We will call a subset $\bm{X'} \subseteq \textsf{Lnodes}(G'(n,l,n))$ \textit{totally adjacent} with respect to the above total order, if it possible to write the elements of $\bm{X'}$ as a sequence $x'_{1},...,x'_{|\bm{X'}|}$ such that (i) $x'_{1} < x'_{2} < ... < x'_{|\bm{X'}|}$ in the order of \Cref{eq:Lnodes_ordering} and (ii) for each $t \in \{1, ..., |\bm{X'}| - 1\}$, there does not exist a $v \in \textsf{Lnodes}(G'(n,l,n))$ such that $x'_{t} < v < x'_{t+1}$. The sequence $x'_{1},...,x'_{|\bm{X'}|}$ will be called the \textit{adjacency sequence} of $\bm{X'}$.

We will also make use of the notation $\textsf{Anc}_{G'(n,l,j))}(\bm{U'}) \subseteq \textsf{Lnodes}(G'(n,l,j))$ to denote the set of ancestors in $G'(n,l,j)$ for $\bm{U'} \subseteq \textsf{Vnodes}(G'(n,l,j))$.

Finally, let us recall the definition of injectable and expressible sets. For an inflation $G'$ of a causal model $G$, a subset $\bm{U'} \subseteq \textsf{Vnodes}(G')$ is \textit{injectable} if there exists a subset $\bm{U} \subseteq \textsf{Vnodes}(G)$ such that $\bm{U'}$ and $\bm{U}$ are the same up to copy-indices  and moreover the ancestral subgraph of $\bm{U'}$ in $G'$, denoted $\mathsf{ansubgraph}_{G'}(\bm{U'})$, is the same up to copy-indices as the ancestral subgraph of $\bm{U}$ in $G$, denoted $\mathsf{ansubgraph}_{G}(\bm{U})$. Furthermore, a subset $\bm{U'} \subseteq \textsf{Vnodes}(G')$ is \textit{ai-expressible} if it can be written as the union of injectable sets that are ancestrally independent. The notion of ai-expressibility is a special case of the general notion of expressible set (c.f. \cite{WolfeSpekkensFritz_2019}), however, as we require neither the more general notion here nor to distinguish between the different cases, we drop the `ai-' from the terminology and refer to the sets simply as ``expressible''. Any subset of an injectable or expressible set is also injectable or expressible.

As we will be considering injectable and expressible sets for the inflation $G'(n,l,n)$, let us recall the connectivity of this model: 
\begin{itemize}
    \item every latent node $L_{2}^{[s]}$, where $s \in \{n-l+2,...,n\}$ is the parent of $A_{1}^{[1]},...,A_{1}^{[s+l-1]}$;
    \item every latent node $L_{1}^{[s]}$, where $s \in \{1,...,n-l+1\}$, is the parent of $A_{1}^{[s]},...,A_{1}^{[s + l - 1]}$;
    \item every latent node $L_{1}^{[s]}$, where $s \in \{n-l+2,...,n\}$, is the parent of $A_{1}^{[s]},...,A_{1}^{[n]}$.
\end{itemize}
Before proving the theorem, we consider the following lemmata:
\begin{Lemma} \label{lem:adj_seq_subset} Let $\bm{U'} \subseteq \textsf{Vnodes}(G'(n,l,n))$ be non-empty. Suppose that $\bm{X'} := \textsf{Anc}_{G'(n,l,n))}(\bm{U'})$ is totally adjacent with respect to the order on $\textsf{Lnodes}(G'(n,l,n))$ (i.e. \Cref{eq:Lnodes_ordering}) with adjacency sequence $x'_{1},...,x'_{|\bm{X'}|}$. Then either 
\begin{enumerate}
    \item[(a)] $x'_{1} = L_{2}^{[s]}$ and $x'_{|\bm{X'}|} = L_{1}^{[s+|\bm{X'}|-1]}$ for some $s \in \{n-l+2, ..., n\}$, or
    \item[(b)] $x'_{1} = L_{1}^{[s]}$ and $x'_{|\bm{X'}|} = L_{1}^{[s + |\bm{X'}|-1]}$ for some $s \in \{1, ..., n-l+1\}$.
\end{enumerate} 
In either case, $\bm{U'} \subseteq \{A_{1}^{[s+l-1]}, ..., A_{1}^{[s + |\bm{X'}|-1]} \}$ for the relevant value of $s$.
\end{Lemma}

\begin{proof} Let $\bm{U'} \subseteq \textsf{Vnodes}(G'(n,l,n))$ be non-empty and $\bm{X'} := \textsf{Anc}_{G'(n,l,n))}(\bm{U'})$ be totally adjacent with respect to the order on $\textsf{Lnodes}(G'(n,l,n))$. Since $\bm{U'}$ contains at least one visible node and since each visible node in $G'(n,l,n)$ has at least $l$ parents, $|\bm{X'}| \ge l$. Note that $\textsf{Lnodes}(G'(n,l,n))$ contains $l-1$ elements with a subscript `$2$', namely $L_{2}^{[n-l+2]}, ..., L_{2}^{[n]}$, meaning that any totally adjacent subset of $\textsf{Lnodes}(G'(n,l,n))$ containing at least $l$ elements has adjacency sequence that must terminate with an element with a subscript `$1$'. The possible cases then split into the cases (a) and (b) in the statement of the lemma. 

Due to the assumption of total adjacency for $\bm{X'}$, the sequence $x'_{1}, ..., x'_{|\bm{X'}|}$ contains $|\bm{X'}|$ adjacent elements in the ordering of \Cref{eq:Lnodes_ordering}. If, as in case (a), the sequence commences with $x_{1} = L_{2}^{[s]}$ for some $s \in \{n-l+2, ..., n\}$, then the sequence is
\begin{align}
L_{2}^{[s]},L_{2}^{[s+1]},\dots,L_{2}^{[n]},L_{1}^{[1]},L_{1}^{[2]},\dots,L_{1}^{[s+|\bm{X'}| - 1]}. 
\end{align}
If, as in case (b), the sequence commences with $L_{1}^{[s]}$ for $s \in \{1, ..., n - l + 1 \}$ (note that if $s \in \{n-l+2, ...,n \}$, then the sequence couldn't be at least $l$ elements long and totally adjacent with respect to the order \Cref{eq:Lnodes_ordering}), then the sequence is
\begin{align}
L_{1}^{[s]},\dots,L_{1}^{[s+|\bm{X'}|-1]}.
\end{align}

We now show that $\bm{U'} \subseteq \{A_{1}^{[s+l-1]}, ..., A_{1}^{[s + |\bm{X'}|-1]} \}$. If $\{A_{1}^{[s+l-1]}, ..., A_{1}^{[s + |\bm{X'}|-1]} \} = \textsf{Vnodes}(G'(n,l,n))$, then this is trivially true so we assume that $\{A_{1}^{[s+l-1]}, ..., A_{1}^{[s + |\bm{X'}|-1]} \} \neq \textsf{Vnodes}(G'(n,l,n))$ from now on. In particular, this means that either $s \neq n-l+2$, resulting in $A_{1}^{[s+l-1]} = A_{1}^{[r]}$ for some $r > 1$ after applying the convention on the indices, or $s \neq n - |\bm{X'}| + 1$, resulting in $A_{1}^{[s + |\bm{X'}| - 1]} = A_{1}^{[r]}$ for some $r < n$, or both.

We proceed by contradiction. Suppose $\bm{U'} \nsubseteq \{A_{1}^{s+l-1}, ..., A_{1}^{[s + |\bm{X'}|-1]} \}$ and recall that either (a') $\bm{X'} = \{L_{2}^{[s]}, L_{2}^{[s+1]},...,L_{2}^{[n]},L_{1}^{[1]}, L_{1}^{[2]},...,L_{1}^{[s+|\bm{X'}|-1]} \}$ for some $s \in \{n-l+2,...,n\}$, or (b') $\bm{X'} = \{L_{1}^{[s]},L_{1}^{[s+1]},...,L_{1}^{[s+|\bm{X'}|-1]} \}$ for some $s \in \{1, ..., n-l+1\}$. In case (a'), if $s \neq n - |\bm{X'}| + 1$ and there exists an $A_{1}^{[t]} \in \bm{U'}$ with $t \in \{s+|\bm{X'}|, ..., n\}$\footnote{Recall that the convenient on index values also applies when considering ranges of values, i.e., if $s + |\bm{X'}| > n$ it is taken modulo $n$ in the range $\{s+|\bm{X'}|, ..., n\}$.}, then since $L_{1}^{[t]}$ is a parent of $A_{1}^{[t]}$, this would require that $L_{1}^{[t]} \in \bm{X'}$, forming a contradiction. If $s \neq n-l+2$ and there exists an $A_{1}^{[t]} \in \bm{U'}$ with $t \in \{1,...,s+l-2\}$, then since $L_{2}^{[t-l+1]}$ is a parent of $A_{1}^{[t]}$, this would require that $L_{2}^{[t-l+1]} \in \bm{X'}$ which also forms a contradiction. Since at least one of these two possibilities must occur under the assumption that $\{A_{1}^{s+l-1}, ..., A_{1}^{[s + |\bm{X'}|-1]} \} \neq \textsf{Vnodes}(G'(n,l,n))$, this completes the case for (a').

The situation for (b') is analogous. If $s \neq n - |\bm{X'}|+1$ and there exists an $A_{1}^{[t]} \in \bm{U'}$ with $t \in \{s+|\bm{X'}|, ..., n\}$, then since $L_{1}^{[t]}$ is a parent of $A_{1}^{[t]}$, this would require that $L_{1}^{[t]} \in \bm{X'}$, forming a contradiction, same as above. If $s \neq n-l+2$ and there exists an $A_{1}^{[t]} \in \bm{U'}$ with $t \in \{1,...,s+l-2\}$, then either (I) $A_{1}^{[t]}$ has a parent $L_{2}^{[t-l+1]}$ if $t \in \{1, ..., l-1\}$ or (II) $A_{1}^{[t]}$ has a parent $L_{1}^{[t-l+1]}$ if $t \in \{l, ..., s+l-2\}$. In case (I), this would require that $L_{2}^{[t-l+1]} \in \bm{X'}$ and in case (II), this would require that $L_{1}^{[t-l+1]} \in \bm{X'}$, both of which form contradictions. 
\end{proof}

\begin{Lemma} \label{lem:Gprime_inj} Let $\textsf{Lnodes}(G'(n,l,n))$ be equipped with the ordering of \Cref{eq:Lnodes_ordering} and let $\bm{U'} \subseteq \textsf{Vnodes}(G'(n,l,n))$. If $\textsf{Anc}_{G'(n,l,n))}(\bm{U'})$ is totally adjacent with respect to the ordering on $\textsf{Lnodes}(G'(n,l,n))$ and if $|\textsf{Anc}_{G'(n,l,n))}(\bm{U'})| = n$, then $\bm{U'}$ is an injectable set.
\end{Lemma}

\begin{proof} Let $\bm{U'} \subseteq \textsf{Vnodes}(G'(n,l,n))$ be such that $\bm{X'} := \textsf{Anc}_{G'(n,l,n))}(\bm{U'})$ is totally ordered and that $|\bm{X'}| = n$. From \Cref{lem:adj_seq_subset}, we know that either $\bm{X'} = \{L_{2}^{[s]},...,L_{2}^{[n]},L_{1}^{[1]},...,L_{1}^{[n+s-1]} \}$ for some $s \in \{n-l+2,...,n\}$ or $\bm{X'} = \{L_{1}^{[s]},...,L_{1}^{[n + s -1]}\}$ for some $s \in \{1, ..., n-l+1\}$, and moreover that $\bm{U'} \subseteq \{A_{1}^{[s+l-1]}, ..., A_{1}^{[n+s-1]} \}$ for the relevant value of $s$. Since any subset of an injectable set is also injectable, it suffices to show that $\bm{W'} := \{A_{1}^{[s+l-1]}, ..., A_{1}^{[n+s-1]} \}$ is injectable in $G'(n,l,n)$.

Recall that $\bm{W'} \subseteq \textsf{Vnodes}(G'(n,l,n))$ is injectable if and only if there exists a subset $\bm{W} \subseteq \textsf{Vnodes}(G(n,l))$ such that $\bm{W'} \sim \bm{W}$ and $\mathsf{ansubgraph}_{G'(n,l,n)}(\bm{W'}) \sim \mathsf{ansubgraph}_{G(n,l)}(\bm{W})$, where $\sim$ denotes sameness up to copy-index and where $\mathsf{ansubgraph}_{G'(n,l,n)}(\bm{W'})$ denotes the ancestral subgraph of $\bm{W'}$ in $G'(n,l,n)$ (similarly for $\mathsf{ansubgraph}_{G(n,l)}(\bm{W})$). Defining $\bm{W} := \{A^{[s+l-1]}, ..., A^{[n+s-1]}\}$, we have that $\bm{W'} \sim \bm{W}$. 

Next, recall that $\textsf{Lnodes}(G(n,l)) = \{L^{[1]},...,L^{[n]}\}$. Due to the choice of ordering in \Cref{eq:Lnodes_ordering} and to the fact that $|\bm{X'}| = n$, we have that $\bm{X'} \sim \textsf{Lnodes}(G(n,l))$ in either case for $\bm{X'}$. Since each visible node $A^{[t]}$ in $G(n,l)$ has parents $L^{[t-l+1]},L^{[t-l+2]},...,L^{[t]}$, we see that $\bm{X} := \textsf{Anc}_{G(n,l)}(\bm{W}) = \{L^{[s]},...,L^{[n+s-1]} \} = \textsf{Lnodes}(G(n,l))$. Hence, we have that $\bm{X'} \sim \bm{X}$. Recall that, by definition of inflation, the ancestral subgraph of any single visible node in an inflation is  the same up to copy-indices as the ancestral subgraph of the corresponding visible node in the original causal model. Since $\bm{X'}$ contains no two elements with a common superscript but different subscripts, any common parents of any subset of visible nodes of $\bm{W'}$ will be the unique latent nodes with the relevant superscript, which ensures that $\mathsf{ansubgraph}_{G'(n,l,n)}(\bm{W'}) \sim \mathsf{ansubgraph}_{G(n,l)}(\bm{W})$.
\end{proof}

\begin{Lemma} \label{lem:Gprime_expr} Let $\textsf{Lnodes}(G'(n,l,n))$ be given the ordering in \Cref{eq:Lnodes_ordering}. Let $\bm{U'_{1}}, \bm{U'_{2}}\subseteq \textsf{Vnodes}(G'(n,l,n))$. If 
\begin{enumerate}
    \item $A_{1}^{[n]} \in \bm{U'_{1}}$ and $A_{1}^{[1]} \in \bm{U'_{2}}$,
    \item $\textsf{Anc}_{G'(n,l,n))}(\bm{U'_{1}})$ and $\textsf{Anc}_{G'(n,l,n))}(\bm{U'_{2}})$ are totally adjacent, and
    \item $\textsf{Anc}_{G'(n,l,n))}(\bm{U'_{1}})$ and $\textsf{Anc}_{G'(n,l,n))}(\bm{U'_{2}})$ form a set partition of $\textsf{Lnodes}(G'(n,l,n))$,
\end{enumerate}
then $\bm{U'_{1}} \cup \bm{U'_{2}}$ is an expressible set.
\end{Lemma}

\begin{proof} Let $\bm{U'_{1}}, \bm{U'_{2}}\subseteq \textsf{Vnodes}(G'(n,l,n))$ be such that $A_{1}^{[n]} \in \bm{U'_{1}}$ and $A_{1}^{[1]} \in \bm{U'_{2}}$, $\bm{X_{1}'} := \textsf{Anc}_{G'(n,l,n))}(\bm{U'_{1}})$ and $\bm{X'_{2}} := \textsf{Anc}_{G'(n,l,n))}(\bm{U'_{2}})$ are totally adjacent, and $\bm{X'_{1}}$ and $\bm{X'_{2}}$ form a set partition of $\textsf{Lnodes}(G'(n,l,n))$. Recall that for $\bm{U'_{1}} \cup \bm{U'_{2}}$ to be expressible, it suffices to have that $\bm{U'_{1}}$ and $\bm{U'_{2}}$ are injectable and are ancestrally independent. However, the latter requirement follows directly from the assumption that $\bm{X'_{1}}$ and $\bm{X'_{2}}$ form a set partition of $\textsf{Lnodes}(G'(n,l,n))$ (recall that mutual exclusivity is a requirement in the definition of set partition). 

To show that $\bm{U'_{1}}$ and $ \bm{U'_{2}}$ are indeed injectable, we first note the following bounds on $\bm{X'_{1}}$ and $\bm{X'_{2}}$. Since $A_{1}^{[n]} \in \bm{U'_{1}}$, it must be the case that $|\bm{X'_{1}}| \geq l$ since $A_{1}^{[n]}$ has $l$ parents in $G'(n,l,n)$. By the same reasoning $|\bm{X'_{2}}| \geq l$. Since $|\textsf{Lnodes}(G'(n,l,n))| = n+l-1$, by the assumption that $\bm{X'_{1}}, \bm{X'_{2}}$ partition $\textsf{Lnodes}(G'(n,l,n))$ it must also be the case that $|\bm{X'_{1}}|, |\bm{X'_{2}}| < n$. We now show that both $\bm{U'_{1}}$ and $\bm{U'_{2}}$ are subsets of injectable sets and are hence injectable.

The case for $\bm{U'_{1}}$ proceeds as follows. Since $A_{1}^{[n]} \in \bm{U'_{1}}$ and $L_{1}^{[n]}$ is a parent of $A_{1}^{[n]}$, it must be that $L_{1}^{[n]} \in \bm{X'_{1}}$. Since $\bm{X'_{1}}$ is totally adjacent by assumption, the adjacency sequence $x'_{1}, ..., x'_{|\bm{X'_{1}}|}$ of $\bm{X'_{1}}$ must be of the form
\begin{align}
x'_{1}, \dots, x'_{|\bm{X'_{1}}|} = L_{1}^{[n - |\bm{X'_{1}}| + 1]} , \dots , L_{1}^{[n]}.
\end{align}
The visible node $A_{1}^{[n - |\bm{X'_{1}}| + l - 1]}$ is not in $\bm{U'_{1}}$ since $L_{1}^{[n - |\bm{X'_{1}|}]}$ is a parent of $A_{1}^{[n - |\bm{X'_{1}}| + l - 1]}$ but is not in $\bm{X'_{1}}$. However, all other parents of $A_{1}^{[n - |\bm{X'_{1}}| + l - 1]}$ are in $\bm{X'_{1}}$, so then $\textsf{Anc}_{G'(n,l,n))}(\bm{U'_{1}} \cup \{A_{1}^{[n - |\bm{X'_{1}}| + l - 1]} \}) = \bm{X'_{1}} \cup \{ L_{1}^{[n - |\bm{X'_{1}|}]}\}$. Similarly, $A_{1}^{[n - |\bm{X'_{1}}| + l - 2]}$ is not in $\bm{U'_{1}} \cup \{A_{1}^{[n - |\bm{X'_{1}}| + l - 1]} \}$ since $L_{1}^{[n - |\bm{X'_{1}|-1}]}$ is a parent of $A_{1}^{[n - |\bm{X'_{1}}| + l - 2]}$ but is not in $\bm{X'_{1}} \cup \{ L_{1}^{[n - |\bm{X'_{1}}|]}\}$, while all other parents of $A_{1}^{[n - |\bm{X'_{1}}| + l - 2]}$ are. So then $\textsf{Anc}_{G'(n,l,n))}(\bm{U'_{1}} \cup \{ A_{1}^{[n - |\bm{X'_{1}}| + l - 2]},A_{1}^{[n - |\bm{X'_{1}}| + l - 1]}\}) = \bm{X'_{1}} \cup \{L_{1}^{[n - |\bm{X'_{1}| - 1}]},L_{1}^{[n - |\bm{X'_{1}|}]}\}$. Proceeding in this fashion, we can define 
\begin{align}
\bm{W'_{1}} := \bm{U'_{1}} \cup \{A_{1}^{[l]} , \dots,A_{1}^{[n - |\bm{X'_{1}}| + l - 1]} \}
\end{align}
and observe that $\textsf{Anc}_{G'(n,l,n))}(\bm{W'_{1}})$ has $n$ elements and is totally adjacent (the adjacency sequence is $L_{1}^{[1]}, ..., L_{1}^{[n]}$). Thus $\bm{W'_{1}}$ is injectable by \Cref{lem:Gprime_inj}.

The case for $\bm{U'_{2}}$ proceeds similarly. Since $A_{1}^{[1]} \in \bm{U'_{2}}$ and $L_{2}^{[n-l+2]}$ is a parent of $A_{1}^{[1]}$, it must be that $L_{2}^{[n-l+2]} \in \bm{X'_{2}}$. Since $\bm{X'_{2}}$ is totally adjacent by assumption, the adjacency sequence $y'_{1},...,y'_{|\bm{X'_{2}}|}$ of $\bm{X'_{2}}$ must be of the form
\begin{align}
y'_{1},...,y'_{|\bm{X'_{2}}|} = L_{2}^{[n-l+2]}, ..., L_{2}^{[n]},L_{1}^{[1]}, ..., L_{1}^{[|\bm{X'_{2}}| - l+1]}
\end{align}
The visible node $A_{1}^{[|\bm{X'_{2}}| - l+2]}$ cannot be in $\bm{U'_{2}}$ since $L_{1}^{[|\bm{X'_{2}}| - l+2]}$ is a parent of $A_{1}^{[|\bm{X'_{2}}| - l+2]}$ but is not in $\bm{X'_{2}}$. However, all other parents of $A_{1}^{[|\bm{X'_{2}}| - l+2]}$ are in $\bm{X'_{2}}$, so it follows that $\textsf{Anc}_{G'(n,l,n))}(\bm{U'_{2}} \cup \{A_{1}^{[|\bm{X'_{2}}| - l+2]} \}) = \bm{X'_{2}} \cup \{L_{1}^{[|\bm{X'_{2}}| - l+2]} \}$. Just as with the case above, we continue appending visible nodes to $\bm{U'_{2}}$ until we obtain a provably injectable set. Explicitly, we define $\bm{W'_{2}} := \bm{U'_{2}} \cup \{A_{1}^{[|\bm{X'_{2}}| - l+2]}, ..., A_{1}^{[n - l+1]} \} $ and get that $\textsf{Anc}_{G'(n,l,n))}(\bm{W'_{2}}) = \{L_{2}^{[n-l+2]},\dots,L_{2}^{[n]},L_{1}^{[1]},\dots,L_{1}^{[n-l+1]} \}$ is of size $n$ and totally adjacent (with adjacency sequence $L_{2}^{[n-l+2]}, \dots, L_{2}^{[n]},L_{1}^{[1]}, \dots,L_{1}^{[n-l+1]}$). By \Cref{lem:Gprime_inj}, $\bm{W'_{2}}$ is injectable.
\end{proof}

Let us now turn to the proof of \Cref{thm:family_of_networks}, which we restate here for convenience, with the minor change of replacing $G'(n,l,j)$ for a general $j \in \{1,...,n\}$ with $G'(n,l,n)$ pursuant to the comments made earlier: 
\begin{Theorem}  Let $k \in \mathbb{N}_{>1}$ be odd. If $n \geq (l-1)k$, then there exists a partition $\{\mathcal{S}_{1},...,\mathcal{S}_{k}\}$ of $\textsf{Vnodes}(G'(n,l,n))$ such that, taking ${\bf V} := \{\mathcal{S}_{1},...,\mathcal{S}_{k}\}$, all elements of $2^{{\bf V}} \setminus {\bf V}$ are injectable or expressible in $G'(n,l,n)$.
\end{Theorem}

\begin{proof}[Proof of \Cref{thm:family_of_networks}] Consider $G'(n,l,n)$ for $n \geq (l-1)k$ for some odd $k \in \mathbb{N}_{ > 1}$ and consider $\textsf{Lnodes}(G'(n,l,n))$ equipped with the order as in \Cref{eq:Lnodes_ordering}. 

Define $\bm{U'}(r) := \{A_{1}^{[r]},...,A_{1}^{[n-l+r]}\}$ for $r \in \{1,...,l\}$. Then
\begin{align}
\textsf{Anc}_{G'(n,l,n)}(\bm{U'}(r)) = \begin{cases} \{L_{2}^{[n-l+1+r]}, ...,L_{2}^{[n]},L_{1}^{[1]},...,L_{1}^{[n-l+r]} \}, &\text{ if } r \in \{1,...,l-1\}, \\
\{L_{1}^{[1]},...,L_{1}^{[n]} \}, &\text{ if } r = l.
\end{cases}
\end{align}
In either case, $\textsf{Anc}_{G'(n,l,n)}(\bm{U'}(r))$ satisfies the conditions of \Cref{lem:Gprime_inj} and hence $\bm{U'}(r)$ is injectable. 

Now define $\bm{U'_{1}}(t) := \{A_{1}^{[n-t+1]},A_{1}^{[n-t+2]},...,A_{1}^{[n]} \}$ and $\bm{U'_{2}}(t') = \{A_{1}^{[1]}, A_{1}^{[2]}, ..., A_{1}^{[t']} \}$ for $t,t' \in \{1, ..., n-l\}$. We get that
\begin{align}
\textsf{Anc}_{G'(n,l,n)}(\bm{U'_{1}}(t)) &= \{L_{1}^{[n-l+2-t]},\dots ,L_{1}^{[n]}\}, \\
\textsf{Anc}_{G'(n,l,n)}(\bm{U'_{2}}(t')) &= \{L_{2}^{[n-l+2]}, \dots ,L_{2}^{[n]},L_{1}^{[1]}, \dots, L_{1}^{[t']} \},
\end{align}
both of which are totally adjacent with respect to \Cref{eq:Lnodes_ordering}. If in addition $t + t' = n-l+1$, then $\textsf{Anc}_{G'(n,l,n)}(\bm{U'_{1}}(t))$ and $\textsf{Anc}_{G'(n,l,n)}(\bm{U'_{2}}(t'))$ partition $\textsf{Lnodes}(G'(n,l,n))$, in which case $\bm{U'_{1}}(t) \cup \bm{U'_{2}}(t')$ is expressible.

Let us now define, for $r \in \{1, ..., l\}$,
\begin{align}
\mathcal{A}(r) := \begin{cases} \{A_{1}^{[1]},...,A_{1}^{[l-1]}\}, & \text{ if } r = l, \\
\{A_{1}^{[n-l+2]},...,A_{1}^{[n]}\}, & \text{ if } r = 1, \\
\{A_{1}^{[1]}, ..., A_{1}^{[r-1]} \} \cup \{A_{1}^{[n-l+ r + 1]}, ..., A_{1}^{[n]} \}, & \text{ if } 1 < r < l,
\end{cases}
\end{align}
and, for $t \in \{1,...,n-l\}$,
\begin{align}
\mathcal{B}(t) := \{ A_{1}^{[n-l+2-t]},...,A_{1}^{[n-t]}\}.
\end{align}
Note that, for any $r$ and $t$ in the allowed ranges, $|\mathcal{A}(r)| = |\mathcal{B}(t)| = l-1$. Furthermore, we get that
\begin{gather}
\textsf{Vnodes}(G'(n,l,n))) \setminus \mathcal{A}(r) = \bm{U'}(r), \\
\textsf{Vnodes}(G'(n,l,n))) \setminus \mathcal{B}(t) = \bm{U'_{1}}(t) \cup \bm{U'_{2}}(n-l+1-t),
\end{gather}
so if we are able to construct a set partition of $\textsf{Vnodes}(G'(n,l,n))$ out of $\mathcal{A}(r)$, $\mathcal{B}(t)$ or sets containing them for suitable values for $r,t$, then we are guaranteed that their set complements are injectable or expressible.

Since a set partition is made up of mutually disjoint sets, let us consider the conditions under which the $\mathcal{A}(r)$ and $\mathcal{B}(t)$ are disjoint. The sets $\mathcal{A}(r_{1})$ and $\mathcal{A}(r_{2})$ satisfy $\mathcal{A}(r_{1}) \cap \mathcal{A}(r_{2}) = \emptyset$ if $r_{1} = 1$ and $r_{2} = l$. The sets $\mathcal{B}(t_{1})$ and $\mathcal{B}(t_{2})$ satisfy $\mathcal{B}(t_{1}) \cap \mathcal{B}(t_{2}) = \emptyset$ if $|t_{1} - t_{2}| > l-2$. The sets $\mathcal{A}(r)$ and $\mathcal{B}(t)$ satisfy $\mathcal{A}(r) \cap \mathcal{B}(t) = \emptyset$ if $r = 1$ and $t > l-2$, if $r = l$ and $t < n -2l + 3$, or if $r \neq 1,l$ and $l-1 < t+r < n-l+3$.

Let us now define $\widehat{\mathcal{S}}_{1} := \mathcal{A}(1)$ and $\widehat{\mathcal{S}}_{q} := \mathcal{B}((q-1)(l-1))$ for $q \in \{2,..., \kappa\}$ where $\kappa := 1 +  \lfloor \frac{n-2l+3}{l-1}\rfloor$. To see that $\kappa$ is well-defined (i.e., $\kappa \ge 2$), we note that by assumption $n \geq k(l-1)$ and $l-1 > 0$, so
\begin{align}
\frac{n-2l+3}{l-1} \ge \frac{k(l-1) - 2l + 3}{l-1} \ge k-2 + \frac{1}{l-1} \ge 1
\end{align}
since $k$ is both odd and greater than $1$. In fact, the above reasoning demonstrates that $\kappa \geq k-1$. Note that the definitions of the $\widehat{\mathcal{S}}_{i}$ so far ensure that they are mutually disjoint: $\widehat{\mathcal{S}}_{1} \cap \widehat{\mathcal{S}}_{q} = \emptyset$  since $(q-1)(l-1) \geq l-2$ for all $q \in \{2, ..., \kappa\}$ and $\widehat{\mathcal{S}}_{q_{1}} \cap \widehat{\mathcal{S}}_{q_{2}} = \emptyset$ for all $q_{1} \neq q_{2} \in \{1,...,\kappa\}$ since $|(q_{1} - 1)(l-1) - (q_{2} - 1)(l-1)| = |(q_{1} - q_{2})(l-1)| > l -2$. Furthermore, we have that
\begin{align}
\bigcup_{i = 1}^{\kappa} \widehat{\mathcal{S}}_{i} = \{A_{1}^{[n-\kappa (l - 1) + 1]}, ..., A_{1}^{[n]}\}.
\end{align}
Let us then define $\widehat{S}_{\kappa + 1} := \{A_{1}^{[1]}, ..., A_{1}^{[n-\kappa (l -1)]}\}$. As it stands, $\{\widehat{\mathcal{S}}_{1}, ..., \widehat{\mathcal{S}}_{\kappa + 1} \}$ forms a valid set partition of $\textsf{Vnodes}(G'(n,l,n))$. Furthermore, we know that $\textsf{Vnodes}(G'(n,l,n)) \setminus \widehat{\mathcal{S}}_{q}$ is injectable or expressible for each $q \in \{1, ..., \kappa\}$ since
\begin{align}
\textsf{Vnodes}(G'(n,l,n)) \setminus \widehat{\mathcal{S}}_{1} = \bm{U'}(1)
\end{align}
and
\begin{align}
\textsf{Vnodes}(G'(n,l,n)) \setminus \widehat{\mathcal{S}}_{q} = \bm{U'_{1}}((q-1)(l-1)) \cup \bm{U'_{2}}(n-q(l-1)), \quad q \in \{2, ..., \kappa\}
\end{align}
by construction. To see that $\textsf{Vnodes}(G'(n,l,n)) \setminus \widehat{\mathcal{S}}_{\kappa +1}$ is injectable also, we note that
\begin{align}
n - \kappa[l-1] \geq n - \frac{n - 2l + 3}{l-1}(l-1) \ge 2l-3 \ge l-1
\end{align}
where the last inequality uses that $l \ge 2$. In particular, this means that $\widehat{\mathcal{S}}_{\kappa+1} \supseteq \mathcal{A}(l)$ and hence that $\textsf{Vnodes}(G'(n,l,n)) \setminus \widehat{\mathcal{S}}_{\kappa +1} \subseteq \bm{U'}(l)$, from which injectibility follows (recall that subsets of injectable sets are injectable).

Finally, it remains only to construct a partition using precisely $k$ subsets. From above, we know that $\kappa \ge k-1$ meaning that $\kappa+1 \ge k$, so the partition $\{\widehat{\mathcal{S}}_{1}, ..., \widehat{\mathcal{S}}_{\kappa+1} \}$ has at least as many elements as required. If $\kappa + 1 = k$, taking $\mathcal{S}_{q} = \widehat{\mathcal{S}}_{q}$ for all $q \in \{1, ..., \kappa + 1\}$ finishes the job. In the case where $\kappa + 1 > k$, we can simply define a partition $\{\mathcal{S}_{1}, ..., \mathcal{S}_{k}\}$ where one or more of the $\mathcal{S}_{i}$ consist of a union of $\widehat{\mathcal{S}}_{q}$. For example, by taking $\mathcal{S}_{q} = \widehat{\mathcal{S}}_{q}$ for all $q = 1, ..., k-1$ and $\mathcal{S}_{k} := \cup_{i=k}^{\kappa + 1} \widehat{\mathcal{S}}_{i}$. In this case, $\textsf{Vnodes}(G'(n,l,n)) \setminus \mathcal{S}_{q} \subseteq \textsf{Vnodes}(G'(n,l,n)) \setminus \widehat{\mathcal{S}}_{q}$ for all $q \in \{1, ...,k\}$ and so remain injectable or expressible. Thus, for ${\bf V} := \{\mathcal{S}_{1}, ..., \mathcal{S}_{k}\}$, each element of $2^{{\bf V}} \setminus {\bf V}$ is injectable or expressible as required.
\end{proof}

Using \Cref{lem:Gprime_inj} and \Cref{lem:Gprime_expr}, we can now justify the claims made in the main text regarding the injectable and expressible sets of the pentagon and woven hexagon scenarios. For the Cut inflation $G'(5,2,5)$ of the pentagon scenario $G(5,2)$, we consider $\textsf{LNodes}(G'(5,2,5))$ with the order
\begin{align}
L_{2}^{[5]} < L_{1}^{[1]} < L_{1}^{[2]} < L_{1}^{[3]} < L_{1}^{[4]} < L_{1}^{[5]} . \label{eq:pentagon_order}
\end{align}
For $\bm{U'} := \{A_{1}^{[1]}, A_{1}^{[2]}, A_{1}^{[3]}, A_{1}^{[4]} \}$, we have that
\begin{align}
\textsf{Anc}_{G'(5,2,5)}(\{A_{1}^{[1]}, A_{1}^{[2]}, A_{1}^{[3]}, A_{1}^{[4]} \}) = \{L_{2}^{[5]}, L_{1}^{[1]},L_{1}^{[2]},L_{1}^{[3]},L_{1}^{[4]}\}
\end{align}
which has $5$ elements and is totally adjacent with adjacency sequence $L_{2}^{[5]}, L_{1}^{[1]},L_{1}^{[2]},L_{1}^{[3]},L_{1}^{[4]}$, so by \Cref{lem:Gprime_inj}, $\bm{U'}$ is injectable. For $\bm{U'} := \{A_{1}^{[2]}, A_{1}^{[3]}, A_{1}^{[4]}, A_{1}^{[5]}\}$, we have that
\begin{align}
\textsf{Anc}_{G'(5,2,5)}(\{A_{1}^{[2]}, A_{1}^{[3]}, A_{1}^{[4]}, A_{1}^{[5]}\}) = \{L_{1}^{[1]}, L_{1}^{[2]},L_{1}^{[3]},L_{1}^{[4]},L_{1}^{[5]}\}
\end{align}
which also has $5$ elements and is totally adjacent (with adjacency sequence $L_{1}^{[1]}, L_{1}^{[2]},L_{1}^{[3]},L_{1}^{[4]},L_{1}^{[5]}$), so is injectable. 

Defining $\bm{U'_{1}} := \{ A_{1}^{[1]}, A_{1}^{[2]}, A_{1}^{[3]}\}$ and $\bm{U'_{2}} := \{A_{1}^{[5]} \}$, we get that 
\begin{equation}
\begin{gathered}
\textsf{Anc}_{G'(5,2,5)}(\{ A_{1}^{[1]}, A_{1}^{[2]}, A_{1}^{[3]}\}) = \{L_{2}^{[5]}, L_{1}^{[1]}, L_{1}^{[2]}, L_{1}^{[3]} \}, \\
\textsf{Anc}_{G'(5,2,5)}(\{A_{1}^{[5]} \}) = \{L_{1}^{[4]}, L_{1}^{[5]}\},
\end{gathered}
\end{equation}
both of which are totally adjacent with respect to the order in \Cref{eq:pentagon_order} and which form a partition of $\textsf{Lnodes}(G'(5,2,5))$. So by \Cref{lem:Gprime_expr}, $\{A_{1}^{[1]}, A_{1}^{[2]}, A_{1}^{[3]}, A_{1}^{[5]} \}$ is expressible. The cases where $\bm{U'_{1}} := \{A_{1}^{[1]}, A_{1}^{[2]}\}$ and $\bm{U'_{2}} = \{A_{1}^{[4]}, A_{1}^{[5]} \}$, and where $\bm{U'_{1}} := \{A_{1}^{[1]} \}$ and $\bm{U'_{2}} = \{A_{1}^{[3]}, A_{1}^{[4]}, A_{1}^{[5]} \}$, proceed analogously.

For the Cut inflation $G'(6,3,6)$ of the woven hexagon scenario $G(6,3)$, we consider the order on $\textsf{Lnodes}(G'(6,3,6))$ given by
\begin{align}
L_{2}^{[5]} < L_{2}^{[6]} < L_{1}^{[1]} < L_{1}^{[2]} < L_{1}^{[3]} < L_{1}^{[4]} < L_{1}^{[5]} < L_{1}^{[6]}.
\end{align}
Analogously to the pentagon scenario, \Cref{lem:Gprime_inj} can be used to show that $\{A_{1}^{[1]}, A_{1}^{[2]}, A_{1}^{[3]}, A_{1}^{[4]} \}$, $\{A_{1}^{[2]}, A_{1}^{[3]}, A_{1}^{[4]}, A_{1}^{[5]} \}$ and $\{A_{1}^{[3]}, A_{1}^{[4]}, A_{1}^{[5]}, A_{1}^{[6]} \}$ are injectable, and \Cref{lem:Gprime_expr} can be used to show that $\{A_{1}^{[1]}, A_{1}^{[2]}, A_{1}^{[3]} \} \cup \{A_{1}^{[6]} \}$, $\{A_{1}^{[1]}, A_{1}^{[2]} \} \cup \{A_{1}^{[5]}, A_{1}^{[6]} \}$ and $\{A_{1}^{[1]}\} \cup \{A_{1}^{[4]}, A_{1}^{[5]}, A_{1}^{[6]} \}$ are expressible.

\section{Higher-dimensional generalized Schmidt decomposition} \label{app:hdim_gen_schmidt}

The following is a summary of the generalized Schmidt decomposition presented in~\cite[Thm 2,][]{carteret2000multipartite}. After presenting the general case, we consider the specific case considered in the main text. 

Consider $\H = \bigotimes_{i=1}^{n} \H_{i}$ with $n \geq 3$ and let $d_{i} := \dim \H_{i}$ and $D := \dim \H$. Let us label the component Hilbert spaces such that $2 \leq d_{1} \leq ... \leq d_{n}$. Let $\{I_{1}, ..., I_{N}\}$ be the ordered set of length $(n-1)$ tuples $(i_{1},...,i_{n-1})$ with $1 \leq i \leq d_{k}$ which excludes tuples of the form $(i,...,i)$ and $(d_{1}, ..., d_{k},i,...,i)$ with $d_{k} \leq i < d_{k+1}$ and $1 \leq k \leq n-2$. The ordering of $\{I_{1},...,I_{N}\}$ is lexicographical. Define the set $\mathcal{B}_{0}$ to be the set of $n$-tuples $(i_{1}, ..., i_{n})$ with $(i_{1}, ..., i_{n-1}) = I_{k}$ and $i_{n} = d_{n-1} + l$ where $1 \leq k \leq \min\{N, d_{n} - d_{n-1} \}$ and $k \leq l \leq d_{n} - d_{n-1}$. Define also the following sets of $n$-tuples:
\begin{align}
\mathcal{B}_{1} &:= \left\{ (d_{1}, ..., d_{k-1}, j, d_{k+1},...,d_{n-1}, d_{n-1}) : 1 \leq j < d_{k}, 1 \leq k \leq n-2 \right\} \\
\mathcal{B}_{2} &:= \left\{(d_{1}, ..., d_{n-2}, j, d_{n-1}) : 1 \leq j < d_{n-2} \right\}, \\
\mathcal{B}_{3} &:= \left\{(d_{1},...,d_{n-2},j,1) : d_{n-2} \leq j \leq d_{n-1} \right\}, \\
\mathcal{B}_{4} &:= \left\{ (j,j,...,j): 2 \leq j \leq d_{1} \right\} \cup \left\{ (d_{1},...,d_{k}, j, ..., j): 1 \leq k < n-1, d_{k} < j \leq d_{k+1} \right\} \nonumber \\
& \quad \quad \cup \left\{(d_{1}, ..., d_{n-1},j) : d_{n-1} < j \leq \min\left\{\frac{D}{d_{n}}, d_{n} \right\} \right\}.
\end{align}

Any $\ket{\Psi} \in \H$ can be written as 
\begin{align}
\ket{\Psi} = \sum_{i_{1}, ..., i_{n}} \alpha_{i_{1}...i_{n}}\ket{\psi_{i_{1}}^{(1)}...\psi_{i_{n}}^{(n)}}
\end{align}
where $\{\ket{\psi_{j}}^{(k)}\}_{j = 1}^{d_{k}}$ is a basis for $\H_{k}$ for each $k$ and where the coefficients $\alpha_{i_{1}...i_{n}}$ satisfy:
\begin{enumerate}
    \item \label{cond_1} $\alpha_{ii...ijii...i} = 0$ if $1 \leq i < d_{1}$ and $i < j$;
    \item \label{cond_2} $\alpha_{d_{1}...d_{k}ii...ijii...i} = 0$ if $d_{k} \leq i < d_{k+1}$ and $i < j$, $1 \leq k \leq n-2$;
    \item \label{cond_3} $\alpha_{I}$ for every $n$-tuple index $I \in \mathcal{B}_{0}$;
    \item \label{cond_4} the coefficients with indices in the sets $\mathcal{B}_{x}$, $x = 1,...,4$ are real and non-negative;
    \item \label{cond_5} for $i = 1,...,d_{n-1}$, define
    \begin{align}
    R_{i} = | \alpha_{d_{1}...d_{k}i...i}|,
    \end{align}
    where $k$ is such that $d_{k} < i \leq d_{k+1}$. Then $R_{1} \geq ... \geq R_{d_{n-1}}$.
\end{enumerate}

In~\Cref{sec:higher_dim_nodes}, we consider pure states in $\H = \H_{A} \otimes \H_{B} \otimes \H_{C} \cong \mathbb{C}^{2} \otimes \mathbb{C}^{2} \otimes \mathbb{C}^{4}$. This means that $d_{1} = d_{2} =2$, $d_{3} = 4$, $D = 16$ and $N = 3$. We have that 
\begin{gather}
\{I_{1}, I_{2}, I_{3} \} = \{(1,2), (2,1), (2,2)\}, \\
\mathcal{B}_{0} = \{(1,2,3),(1,2,4),(2,1,4) \},\\
\mathcal{B}_{1} = \{(1,2,2)\},\\
\mathcal{B}_{2} = \{(2,1,2) \},\\
\mathcal{B}_{3} = \{(2,2,1) \},\\
\mathcal{B}_{4} = \{(2,2,2) \} \cup \emptyset \cup \{(2,2,3), (2,2,4)\}.
\end{gather}
From~\Cref{cond_1}, we get that $\alpha_{112} = \alpha_{121} = \alpha_{211} = \alpha_{113} = \alpha_{114} = 0$.~\Cref{cond_2} doesn't rule out any coefficients.~\Cref{cond_3} ensures that $\alpha_{123} = \alpha_{124} = \alpha_{214} = 0$.~\Cref{cond_4} gives that $\alpha_{122}, \alpha_{212}, \alpha_{221}, \alpha_{222}, \alpha_{223}, \alpha_{224} \in \mathbb{R}_{\geq 0}$ and~\Cref{cond_5} requires that $|\alpha_{111}| \geq |\alpha_{222}|$.

To arrive at the form given in the main text (\Cref{eq:hd_gen_schmidt_224}), we make a few notational modifications. Firstly, we identify the bases $\{\ket{\psi_{j}}^{(k)}\}_{j = 1}^{d_{k}}$ for $k = A,B,C$ with the computational basis in each case (i.e. $\{\ket{0}, \ket{1}\}$ for $A$ and $B$, and $\{\ket{0}, ..., \ket{3}\}$ for $C$). We write each coefficient not constrained to by real and positive in polar form and relabel all non-zero coefficients $\alpha_{ijk}$ to $\alpha_{x}$ for $x = 0, ..., 7$ where the ordering corresponds to the lexicographical ordering of the corresponding computational basis states.

\section{Some inflations produce trivial witnesses} \label{app:ring_inflation}

\begin{figure}[htbp]
\begin{center}
\includegraphics[width=0.50\textwidth]{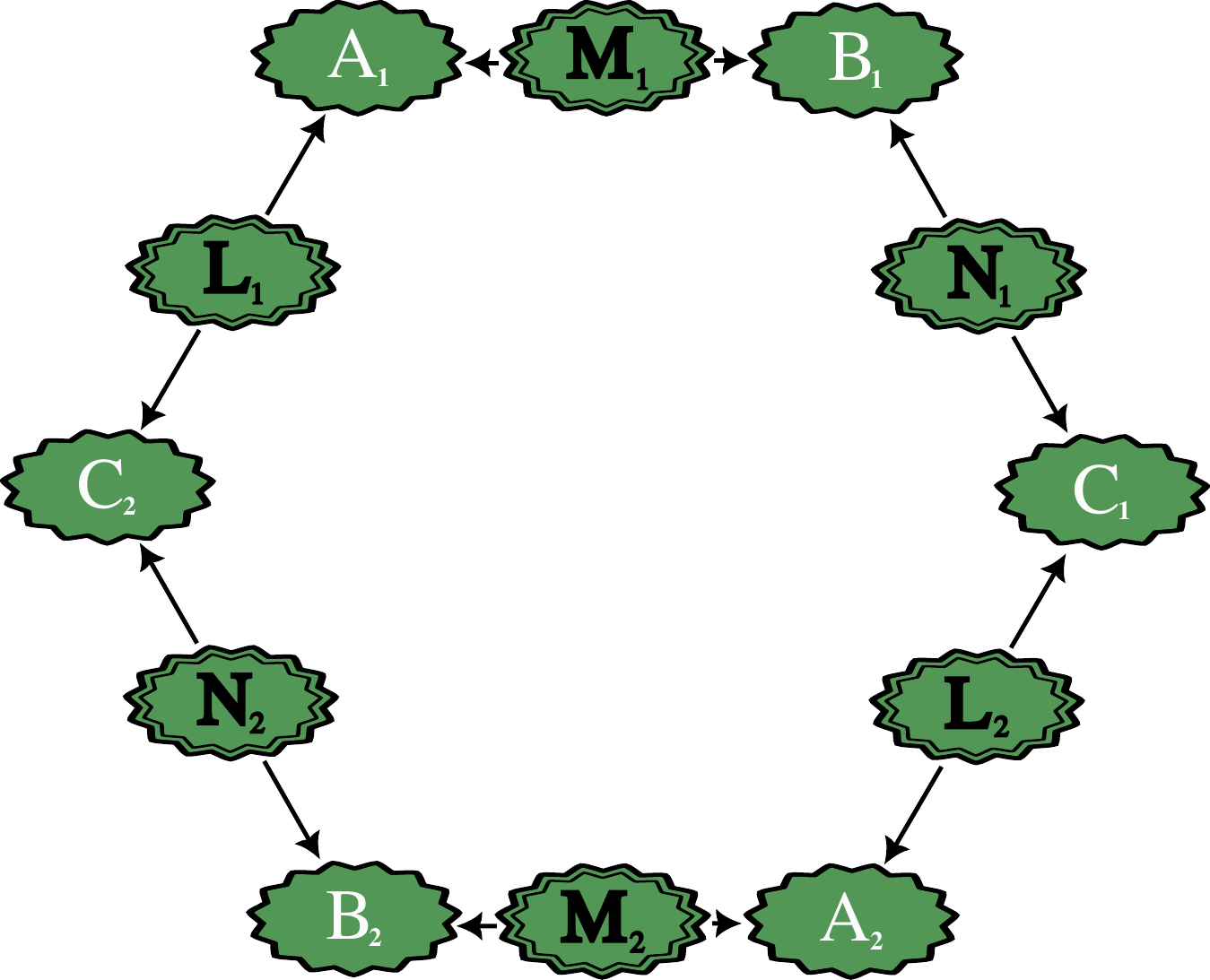}
\caption{The Ring inflation of the triangle scenario. Note that, despite the similar structure, this DAG is distinct from that of the network $G(6,2)$ from the family of networks considered in \Cref{sec:beyond_triangle}, as it is an inflation of the triangle scenario and hence has node labels that are the same up to copy-indices as those of the triangle scenario (i.e., the network $G(3,2)$).}
\label{fig:ring_inf}
\end{center}
\end{figure}

For the networks considered in the main text, we focused on a single type of nonfanout inflation, namely that of the Cut inflation. However, for a given network $G$, various choices of nonfanout inflation are often available. It is then natural to ask whether these other choices also lead to non-trivial $G$-incompatibility witnesses. Here, we demonstrate that this is not guaranteed to be the case.

Consider the nonfanout inflation of the triangle scenario depicted in \Cref{fig:ring_inf}, which is known as the Ring inflation \cite{gisin2020constraints, navascues2020genuine}. For this inflation, the sets
\begin{align}
\{A_{1}, B_{1} \}, \quad \{B_{1},C_{1}\}, \quad \{C_{1}, A_{2} \}, \quad \{A_{2}, B_{2} \}, \quad \{B_{2}, C_{2} \}, \quad \{C_{2}, A_{2}\},
\end{align}
are all injectable. Furthermore, any pair of these sets that differ only in their copy-indices, for example $\{A_{1},B_{1}\}$ and $\{A_{2},B_{2}\}$, share no common ancestors, meaning that the sets
\begin{align}
\{A_{1},B_{1},A_{2},B_{2}\}, \quad \{A_{1},C_{1},A_{2},C_{2}\}, \quad \{B_{1},C_{1},B_{2},C_{2}\}
\end{align}
and all subsets thereof, are expressible. 

Let ${\bf V} := \{\mathcal{S}_{1}, \mathcal{S}_{2}, \mathcal{S}_{3}\}$ with $\mathcal{S}_{1} := \{A_{1},A_{2} \}$, $\mathcal{S}_{2} := \{B_{1},B_{2} \}$ and $\mathcal{S}_{3} := \{C_{1}, C_{2}\}$. The set ${\bf V}$ defines a partition of the visible nodes of the Ring inflation and moreover, $2^{{\bf V}} \setminus {\bf V}$ consists of injectable and expressible sets as established above. As per \Cref{sec:beyond_triangle}, we are thus in a position to write down a candidate triangle-incompatibility witness based on the inequality from Hall (cf \Cref{eq:Hall_op_ineq}). For $\rho_{ABC}$ a state on the Hilbert spaces associated to the visible nodes of the triangle scenario, the triangle-incompatibility witness for the choice of ${\bf V}$ defined above is
\begin{align}
I^{\textrm{Ring}}(\rho_{ABC}) = \mathbb{1} - \rho_{A} \otimes \rho_{A} - \rho_{B}\otimes \rho_{B} - \rho_{C} \otimes \rho_{C} + \rho_{AB} \otimes \rho_{AB} + \rho_{AC} \otimes \rho_{AC} + \rho_{BC} \otimes \rho_{BC}
\end{align}
where each term in the above expression is to be understood as an operator on $\mathcal{H}_{A_{1}} \otimes \mathcal{H}_{B_{1}} \otimes \mathcal{H}_{C_{1}} \otimes \mathcal{H}_{A_{2}} \otimes \mathcal{H}_{B_{2}} \otimes \mathcal{H}_{C_{2}}$ (that is, the identity operators have been omitted from the notation as in the main text).

Despite being expressed entirely in terms of marginals of $\rho_{ABC}$, just like the triangle-incompatibility witnesses derived from the Cut inflation, $I^{\textrm{Ring}}$ is a trivial triangle-incompatibility witness. That is, $I^{\textrm{Ring}}(\rho_{ABC})$ is positive semidefinite for all $\rho_{ABC}$, meaning that the triangle-incompatibility of no state $\rho_{ABC}$ can be witnessed. To see that $I^{\textrm{Ring}}(\rho_{ABC})$ is positive semidefinite, let us define $\mathcal{H}_{\bm{A}} := \mathcal{H}_{A_{1}} \otimes \mathcal{H}_{A_{2}}$ and similarly for $\mathcal{H}_{\bm{B}}$ and $\mathcal{H}_{\bm{C}}$. Let $\sigma$ be a state on $\mathcal{H}_{\bm{A}} \otimes \mathcal{H}_{\bm{B}} \otimes \mathcal{H}_{\bm{C}}$ given by $\sigma = \rho_{ABC} \otimes \rho_{ABC}$, with the ordering of the Hilbert spaces such that the two copies of $\rho_{ABC}$ corresponding to $\mathcal{H}_{A_{1}}\otimes \mathcal{H}_{B_{1}} \otimes \mathcal{H}_{C_{1}}$ and $\mathcal{H}_{A_{2}}\otimes \mathcal{H}_{B_{2}} \otimes \mathcal{H}_{C_{2}}$ respectively. For ${\bf V'} := \{\bm{A}, \bm{B}, \bm{C}\}$, \Cref{eq:Hall_op_ineq} guarantees that
\begin{align}
\Delta(\{\sigma_{{\bf X}} : {\bf X} \in 2^{{\bf V'}}\setminus {\bf V'}  \}) \ge 0.
\end{align}
However, $\Delta(\{\sigma_{{\bf X}} : {\bf X} \in 2^{{\bf V'}}\setminus {\bf V'}  \})$ is precisely the operator $I^{\textrm{Ring}}(\rho_{ABC})$:
\begin{align}
\Delta(\{\sigma_{{\bf X}} : {\bf X} \in 2^{{\bf V'}}\setminus {\bf V'}  \}) &= \mathbb{1} - \sigma_{\bm{A}} - \sigma_{\bm{B}} - \sigma_{\bm{C}} + \sigma_{\bm{AB}} + \sigma_{\bm{AC}} + \sigma_{\bm{BC}} \\
&= \mathbb{1} - \rho_{A} \otimes \rho_{A} - \rho_{B}\otimes \rho_{B} - \rho_{C} \otimes \rho_{C} + \rho_{AB} \otimes \rho_{AB} + \rho_{AC} \otimes \rho_{AC} + \rho_{BC} \otimes \rho_{BC} \\
&= I^{\textrm{Ring}}(\rho_{ABC}).
\end{align}
This example indicates that some care needs to be taken when leveraging the inequality of \Cref{eq:Hall_op_ineq} in the pursuit of non-trivial $G$-incompatibility witnesses for a given network scenario $G$. An analysis of the conditions under which a non-trivial witness can be derived is left for future work.

\section{Example illustrating Proposition~\ref{mainlemma}} \label{app:example_prop_1}

In this section, we provide an example that illustrates the content of Proposition~\ref{mainlemma} from the main text, which we have restated here:
\begin{Proposition}  
Let $G'$ be an inflation of the network $G$.
In this case, if the family of states $\{ \rho_{\bm{U}} : \bm{U} \in \mathsf{ImagesInjectableSets}(G) \}$ is compatible with $G$, then the family of states $\{ \sigma_{\bm{Y}'} : \bm{Y}' \in \mathsf{ExpressibleSets}(G') \}$, where 
 \begin{align}
 \sigma_{\bm{Y}'} := \bigotimes_{\bm{U} \in \mathsf{ImagesInjectableComponents}(\bm{Y}')} \rho_{\bm{U}},
 \end{align} 
 is compatible with $G'$. 
\end{Proposition}

In our example, let $G$ be the triangle scenario and let $G'$ be the $AB$-Cut inflation of $G$. Recall that the singleton sets $\{A_{1}\}$, $\{B_{1}\}$ and $\{C_{1}\}$ are all injectable sets by the definition of inflation and hence are elements of $\mathsf{ExpressibleSets}(G')$. Similarly, the sets $\{A_{1}, C_{1}\}$ and $\{B_{1},C_{1}\}$ are also injectable in the $AB$-Cut inflation of the triangle, while $\{A_{1}, B_{1}\}$ is not injectable but is expressible, with injectable components given by $\{A_{1}\}$ and $\{B_{1}\}$. 

Suppose that $\{ \rho_{\bm{U}} : \bm{U} \in \mathsf{ImagesInjectableSets}(G) \}$ is compatible with $G$. In particular, this means that there exists a state $\rho_{ABC}$ and a set of parameters 
\begin{align}
\texttt{par} := \{\E_{A|\LA \MA}, \; \E_{B|\MB \NB},\; \E_{C|\LC \NC},\; \rho_{L}, \;\rho_{M}, \;\rho_{N} \},
\end{align}
where the maps $\E_{A|\LA \MA}$, $\E_{B|\MB \NB}$, $\E_{C|\LC \NC}$ and states $\rho_{L} = \rho_{\LA \LC}$, $\rho_{M} = \rho_{\MA \MB}$, $\rho_{N} = \rho_{\NB \NC}$ are as in \Cref{Qcausalmodels}, such that
\begin{align}
\rho_{ABC} = \E_{A|\LA \MA} \otimes \E_{B|\MB \NB} \otimes \E_{C|\LC \NC}(\rho_{\LA \LC} \otimes \rho_{\MA \MB} \otimes \rho_{\NB \NC}).
\end{align}
Let us now consider the set of parameters $\texttt{par}'$ given by $\texttt{par}' = \textsf{Inflation}_{G\rightarrow G'}(\texttt{par})$. That is, we consider the latent spaces of $G'$ to be given by $\mathcal{H}_{L_{1}} \cong \mathcal{H}_{L}$, $\mathcal{H}_{M_{1}} \cong \mathcal{H}_{\MB}$, $\mathcal{H}_{M_{2}} \cong \mathcal{H}_{\MA}$ and $\mathcal{H}_{N_{1}} \cong \mathcal{H}_{N}$ and the visible space to be $\mathcal{H}_{A_{1}} \cong \mathcal{H}_{A}$, $\mathcal{H}_{B_{1}} \cong \mathcal{H}_{B}$ and $\mathcal{H}_{C_{1}} \cong \mathcal{H}_{C}$. The set $\texttt{par}'$ is then 
\begin{align}
\texttt{par}' = \{\E_{A_{1}|L_{1}^{(A_{1})} M_{2}},\; \E_{B_{1}|M_{1} N_{1}^{(B_{1})}},\; \E_{C_{1}|L_{1}^{(C_{1})}N_{1}^{(C_{1})}}, \; \rho_{L_{1}},\; \rho_{M_{1}},\; \rho_{M_{2}},\; \rho_{N_{1}} \},
\end{align}
where 
\begin{gather}
\E_{A_{1}|L_{1}^{(A_{1})} M_{2}} \equiv \E_{A|\LA \MA}, \quad \E_{B_{1}|M_{1} N_{1}^{(B_{1})}} \equiv \E_{B|\MB \NB}, \quad \E_{C_{1}|L_{1}^{(C_{1})}N_{1}^{(C_{1})}} \equiv \E_{C|\LC \NC}, \\
\rho_{L_{1}}\equiv \rho_{\LA \LC}, \quad \rho_{M_{1}} \equiv \rho_{\MB}, \quad \rho_{M_{2}} \equiv \rho_{\MA}, \quad \rho_{N_{1}} \equiv \rho_{\NB \NC},
\end{gather} 
with ``$\equiv$'' denoting equality up to the substitution of isomorphic Hilbert spaces (i.e., $\E_{A| \LA \MA}$ is a map from $\mathcal{H}_{\LA} \otimes \mathcal{H}_{\MA}$ to $\mathcal{H}_{A}$ where as $\E_{A_{1}|L_{1}^{(A_{1})} M_{2}}$ is a map from $\mathcal{H}_{L_{1}^{(A_{1})}} \otimes \mathcal{H}_{M_{2}} \cong \mathcal{H}_{\LA} \otimes \mathcal{H}_{\MA}$ to $\mathcal{H}_{A_{1}} \cong \mathcal{H}_{A}$, and similarly for the other parameters). Consequently, we can consider the state $\sigma_{A_{1}B_{1}C_{1}} \in \mathcal{L}(\mathcal{H}_{A_{1}} \otimes \mathcal{H}_{B_{1}} \otimes \mathcal{H}_{C_{1}})$ given by
\begin{align}
\sigma_{A_{1}B_{1}C_{1}} := \E_{A_{1}|L_{1}^{(A_{1})} M_{2}} \otimes \E_{B_{1}|M_{1} N_{1}^{(B_{1})}} \otimes \E_{C_{1}|L_{1}^{(C_{1})}N_{1}^{(C_{1})}}(\rho_{L_{1}} \otimes \rho_{M_{1}} \otimes \rho_{M_{2}} \otimes \rho_{N_{1}}).
\end{align}
If we consider the marginals of $\sigma_{A_{1}B_{1}C_{1}}$, we see that
\begin{align}
\sigma_{A_{1}} &= \textrm{Tr}_{B_{1}C_{1}}[\E_{A_{1}|L_{1}^{(A_{1})} M_{2}} \otimes \E_{B_{1}|M_{1} N_{1}^{(B_{1})}} \otimes \E_{C_{1}|L_{1}^{(C_{1})}N_{1}^{(C_{1})}}(\rho_{L_{1}} \otimes \rho_{M_{1}} \otimes \rho_{M_{2}} \otimes \rho_{N_{1}})] \\
&= \E_{A_{1}|L_{1}^{(A_{1})} M_{2}} (\rho_{L_{1}^{(A_{1})}} \otimes \rho_{M_{2}}) \\
&\equiv \E_{A|\LA \MA} (\rho_{\LA} \otimes \rho_{\MA}) \\
&= \rho_{A}
\end{align}
and similarly for the marginals $\sigma_{B_{1}}$ and $\sigma_{C_{1}}$. Likewise,
\begin{align}
\sigma_{A_{1}C_{1}} &= \textrm{Tr}_{B_{1}}[\E_{A_{1}|L_{1}^{(A_{1})} M_{2}} \otimes \E_{B_{1}|M_{1} N_{1}^{(B_{1})}} \otimes \E_{C_{1}|L_{1}^{(C_{1})}N_{1}^{(C_{1})}}(\rho_{L_{1}} \otimes \rho_{M_{1}} \otimes \rho_{M_{2}} \otimes \rho_{N_{1}})] \\
&= \E_{A_{1}|L_{1}^{(A_{1})} M_{2}}  \otimes \E_{C_{1}|L_{1}^{(C_{1})}N_{1}^{(C_{1})}}(\rho_{L_{1}} \otimes \rho_{M_{2}} \otimes \rho_{N_{1}^{(C_{1})}}) \\
&\equiv \E_{A|\LA \MA}  \otimes \E_{C|\LC \NC}(\rho_{L} \otimes \rho_{\MA} \otimes \rho_{\NC}) \\
& = \rho_{AC}
\end{align}
and similarly for $\sigma_{B_{1}C_{1}}$. In the case of $\sigma_{A_{1}B_{1}}$, we get that
\begin{align}
\sigma_{A_{1}B_{1}} &= \textrm{Tr}_{C_{1}}[\E_{A_{1}|L_{1}^{(A_{1})} M_{2}} \otimes \E_{B_{1}|M_{1} N_{1}^{(B_{1})}} \otimes \E_{C_{1}|L_{1}^{(C_{1})}N_{1}^{(C_{1})}}(\rho_{L_{1}} \otimes \rho_{M_{1}} \otimes \rho_{M_{2}} \otimes \rho_{N_{1}})] \\
&= \E_{A_{1}|L_{1}^{(A_{1})} M_{2}} \otimes \E_{B_{1}|M_{1} N_{1}^{(B_{1})}}(\rho_{L_{1}^{(A_{1})}} \otimes \rho_{M_{1}} \otimes \rho_{M_{2}
} \otimes \rho_{N_{1}^{(B_{1})}}) \\
&= \E_{A_{1}|L_{1}^{(A_{1})} M_{2}}(\rho_{L_{1}^{(A_{1})}} \otimes \rho_{M_{1}}) \otimes \E_{B_{1}|M_{1} N_{1}^{(B_{1})}}(\rho_{M_{1}} \otimes \rho_{N_{1}^{(B_{1})}}) \\
&\equiv \rho_{A} \otimes \rho_{B}.
\end{align}
Since the states $\sigma_{A_{1}} \equiv \rho_{A}$, $\sigma_{B_{1}} \equiv \rho_{B}$, $\sigma_{C_{1}} \equiv \rho_{C}$, $\sigma_{A_{1}B_{1}} \equiv \rho_{A} \otimes \rho_{B}$, $\sigma_{A_{1}C_{1}} \equiv \rho_{AC}$ and $\sigma_{B_{1}C_{1}} \equiv \rho_{BC}$ are all marginals of $\sigma_{A_{1}B_{1}C_{1}}$ which is compatible with $G'$, it follows that the set $\{\rho_{A},\;  \rho_{B}, \; \rho_{C},\; \rho_{A} \otimes \rho_{B},\; \rho_{AC},\; \rho_{BC} \}$ is compatible with $G'$. Proposition~\ref{mainlemma} establishes this implication for general networks $G$ and inflations $G'$ thereof.

\end{document}